\documentclass[a4paper, 11pt]{article} 
\usepackage[a4paper]{geometry}
\usepackage[english]{babel}
\usepackage[T1]{fontenc}
\usepackage[fleqn]{amsmath}
\usepackage{mathtools}
\usepackage{relsize}
\usepackage{bm}

\usepackage[pdftex]{graphicx}
\usepackage{caption}
\usepackage{subcaption}
\usepackage{tabularx}
\usepackage{amssymb}
\usepackage{amstext}
\usepackage{mathrsfs}
\usepackage{verbatim}
\usepackage{amsthm}
\usepackage{thmtools}

\usepackage{hyperref}
\numberwithin{equation}{section}
\numberwithin{equation}{subsection}
\renewcommand*{\theequation}{%
  \ifnum\value{subsection}=0 %
    \thesection
  \else
    \thesubsection
  \fi
  .\arabic{equation}%
}

\def\RR{{\mathbb{R}}}

\def\NN{{\mathbb{N}}}
\def\ZZ{{\mathbb{Z}}}
\def\dd{{\mathrm{d}}}
\def\DD{{\mathrm{D}}}

\def\diam{\mathrm{diam}\,}
\def\signN{{{\hat{n}}}}
\newcommand{\bpt}[1]{{\bm{#1}}}
\newcommand{\taubp}[0]{{\mathbf{T}}}

\newcommand{\forbidden}[0]{\textsc{Inaccessible}}

\newcommand{\threemat}[6]{\left(\begin{array}{ccc}
#1 & #4 & #6\\
 & #2 & #5\\
 & & #3
\end{array}\right)}
\newcommand{\sixmat}[1]{\left(\begin{array}{ccc}
#1
\end{array}\right)}
\newcommand\numberthis{\addtocounter{equation}{1}\tag{\theequation}}
\declaretheorem[name=Lemma, style=plain, numberwithin=section]{lemma}
\declaretheorem[name=Definition, style=definition, sibling=lemma]{definition}
\declaretheorem[name=Remark, style=remark, sibling=lemma ]{remark}
\declaretheorem[name=Theorem, style=plain]{thm}

\declaretheorem[name=Proposition, style=plain, sibling=lemma]{proposition}
\declaretheorem[name=Question, style=plain, sibling=lemma]{question}
\newtheorem{innerhlemma}{Lemma}
\newenvironment{hlemma}[1]
  {  \renewcommand\theinnerhlemma{#1}\innerhlemma}
  {\endinnerhlemma}
\newtheorem{innerhthm}{Theorem}
\newenvironment{hthm}[1]
  {  \renewcommand\theinnerhthm{#1}\innerhthm}
  {\endinnerhthm}
\newtheorem{innerhproposition}{Proposition}
\newenvironment{hproposition}[1]
  {  \renewcommand\theinnerhproposition{#1}\innerhproposition}
  {\endinnerhthm}

\def\[{\begin{equation*}}
\def\]{\end{equation*}}
\author{Bernhard~Brehm\footnote{Free University Berlin, \texttt{bbrehm@math.fu-berlin.de}}}
\title{Bianchi \textsc{VIII} and \textsc{IX} vacuum cosmologies: 
Almost every solution forms particle horizons and converges to the Mixmaster attractor}
\date{\vspace{-5ex}}
\begin{document}
\maketitle
\begin{abstract}
Bianchi models are posited by the BKL picture to be essential building blocks towards an understanding of generic cosmological singularities. 
We study the behaviour of spatially homogeneous anisotropic vacuum spacetimes of Bianchi type \textsc{VIII} and \textsc{IX}, as they approach the big bang singularity. 

It is known since 2001 that generic Bianchi \textsc{IX} spacetimes converge towards the so-called Mixmaster attractor as time goes towards the singularity. We extend this result to the case of Bianchi \textsc{VIII} vacuum.

The BKL picture suggests that particle horizons should form, i.e.~spatially separate regions should causally decouple. 
We prove that this decoupling indeed occurs, for Lebesgue almost every Bianchi \textsc{VIII} and \textsc{IX} vacuum spacetime.
\end{abstract}

\setcounter{tocdepth}{2}
\tableofcontents

\section{Introduction}\label{sect:intro}
\paragraph{Spatially homogeneous cosmological models.}
The behaviour of cosmological models is governed by the Einstein field equations, coupled with equations describing the presence of matter. Simpler models are obtained under symmetry assumption. The class of models studied in this work, Bianchi models, assume spatial homogeneity, i.e.~``every point looks the same''. Then, one only needs to describe the behaviour over time of any single point, and the partial differential Einstein field equations become a system of ordinary differential equations. 

Directional isotropy assumes that ``every spatial direction looks the same''. This leads to the well-known FLRW (Friedmann-Lemaitre-Robertson-Walker) models. These models describe 
an initial (``big bang'') singularity, followed by an expansion of the universe, slowed down by ordinary and dark matter and accelerated by a competing positive cosmological constant (``dark energy'').

We will assume spatial homogeneity, but relax the assumption of directional isotropy.
Spatial homogeneity assumes that there is a Lie-group $G$ of spacetime isometries, which foliates the spacetime into three-dimensional space-like hypersurfaces on which $G$ acts transitively: For every two points $\bpt x, \bpt y$ in the same hypersurface there is a group element $f\in G$ such that $f\cdot \bpt x=\bpt y$. The resulting ordinary differential equations depend on the Lie-algebra of $G$, the so-called Killing fields. 
The three-dimensional Lie-algebras have been classified by Luigi Bianchi in 1898, hence the name``Bianchi-models''; for a commented translation see \cite{Bianchi2001, jantzen2001editor}, and for a modern treatment see \cite[Section 1]{wainwright2005dynamical}.

The two most studied classes of spatially homogeneous anisotropic cosmological models are the Bianchi-types $\textsc{IX}$ ($so(3)$) and $\textsc{VIII}$ ($sl(2,\RR)$), which are the focus of this work.
Both of these models exhibit a big-bang like singularity in at least one time-direction, and a universe that initially expands from this singularity, until, in the case of Bianchi \textsc{IX}, it recollapses into a time-reversed big bang (``big crunch'')
The big-bang singularity is present even in the vacuum case, where matter is absent and only gravity self-interacts. According to conventional wisdom, ``matter does not matter'' near the singularity. For this reason, we simplify our analysis by considering only the vacuum case.

Note that the symmetry assumptions already restrict the global topology of the space-like hypersurfaces, and that the isotropic FLRW-models are not contained as a special case: The only homogeneous isotropic vacuum model is flat Minkowski space.

For a detailed introduction to Bianchi models, we refer to \cite{wainwright2005dynamical}. A short derivation of the governing ordinary differential Wainwright-Hsu equations \eqref{eq:ode} is given in Section \ref{sect:append:derive-eq}, and physical interpretations of some of our results are given in Section \ref{sect:gr-phys-interpret}. For an excellent survey on Bianchi cosmologies, we refer to \cite{heinzle2009mixmaster}, and for further physical questions we refer to \cite{uggla2003past, heinzle2009cosmological}.
%
\paragraph{The Taub-spaces.}
The dynamical behaviour of Bianchi \textsc{VIII} and \textsc{IX} spacetimes is governed by the so-called Wainwright-Hsu equations. 
This dynamical system contains an invariant set $\mathcal T$ of codimension two, which we call the Taub-spaces. Solutions in this set are also called Taub-NUT spacetimes, or LRS (locally rotationally symmetric) spacetimes. The latter name is descriptive, in the sense that these spacetimes have additional (partial, local) isotropy. Taub spacetimes behave, in several ways, different from general (i.e.~non-Taub) Bianchi spacetimes. For a detailed description, we refer to \cite{ringstrom2001bianchi}.

\paragraph{The Mixmaster attractor.}
The Bianchi dynamical system also contains an invariant set $\mathcal A$, called the Mixmaster attractor, consisting of Bianchi Type \textsc{II} and \textsc{I} solutions. There are good heuristic arguments that $\mathcal A$ really is an attractor for time approaching the singularity, and that the dynamics on and near $\mathcal A$ can be considered chaotic (sometimes also called ``oscillatory''). 
It has been rigorously proven only in Bianchi \textsc{IX} models that $\mathcal A$ is actually attracting, with the exception of some Taub solutions, c.f.~\cite{ringstrom2001bianchi}. 

\paragraph{Stable Foliations.}
One may ask for a more precise description of how solutions get attracted to $\mathcal A$. 
Certain heteroclinic chains in $\mathcal A$ are known to attract hypersurfaces of codimension one, c.f.~\cite{liebscher2011ancient, beguin2010aperiodic}.
Reiterer and Trubowitz \cite{reiterer2010bkl} claim related results, for a much wider class of heteroclinic chains, but with less focus on the regularity or codimension of the attracted sets.
This general class of constructions, i.e.~partial stable foliations over specific solutions in $\mathcal A$, is not the focus of this work. Instead, we describe and estimate solutions and the solution operator (i.e.~flow) directly, without explicitly focussing on the symbolic description of $\mathcal A$.

\paragraph{The question of particle horizons.}
One of the most salient features of relativity is \emph{causality}: The state of the world at some point in spacetime is only affected by states in its past light-cone and can only affect states in its future light-cone. 

Suppose for this paragraph that we orient our spacetime $M$ such that the big bang singularity is situated in the past.
Two points $\bpt p, \bpt q\in M$ are said to \emph{causally decouple} towards the singularity if their past light-cones are disjoint, i.e.~if there is no past event which causally influences both $\bpt p$ and $\bpt q$. The past communication cone of $\bpt p$ is defined to be the set of points, which do not causally decouple from $\bpt p$. The cosmic horizon, also called particle horizon, is the boundary of the past communication cone. Hence, everything beyond the horizon is causally decoupled.

\[\numberthis\label{eq:intro:comcone}\begin{aligned}
\omit\rlap{Past light cone of $\bpt p$:}&\\
J^-(\bpt p) &= \{\bpt q:\, \text{there is } \gamma:[0,1]\to M\,\,\text{with}\, \gamma(0)=\bpt p, \gamma(1)=\bpt q,\,\text{time-like past directed}\}\\
\omit\rlap{Future light cone of $\bpt p$:}&\\
J^+(\bpt p) &= \{\bpt q:\, \text{there is } \gamma:[0,1]\to M\,\,\text{with}\, \gamma(0)=\bpt p, \gamma(1)=\bpt q,\,\text{time-like future directed}\}\\
\omit\rlap{Past communication cone of $\bpt p$:}&\\
J^+(J^-(\bpt p)) &= \bigcup_{\bpt q \in J^-(\bpt p)}J^+(\bpt q) = \{\bpt q:\, {J^-(\bpt q)}\cap {J^{-}(\bpt p)}\neq\emptyset\}\\
\omit\rlap{Past cosmic horizon of $\bpt p$:}&\\
\partial J^+(J^-(\bpt p))& = \mathrm{closure}\,J^{+}(J^-(\bpt p))\,\setminus\, \mathrm{interior}\, J^{+}(J^-(\bpt p)). 
\end{aligned}\]
In Figure \ref{fig:intro:comm-cone}, we illustrate the formation of particle horizons. An example where no particle horizons form is given by the (flat, connected) Minkowski-space $M=\RR\times \RR^3$: There is no singularity, and $M= J^+(J^-(\bpt p))$ for all $\bpt p$, and hence $\partial J^+(J^-(\bpt p)) = \emptyset$.  

Apart from the question of convergence to $\mathcal A$, the next important physical question in the context of Bianchi cosmologies is that of locality of light-cones:
\begin{enumerate}
\item Do nonempty particle horizons $\partial J^+(J^-(\bpt p))\neq \emptyset$ form towards the singularity? Does this happen if $\bpt p$ is sufficiently near to the singularity?
\item Are the spatial hypersurfaces $\{t= t_0=\mathrm{const}\}\cap (J^+(J^-(\bpt p)), \partial J^+(J^-(\bpt p))$, considered as three-dimensional manifolds with boundary, homeomorphic to the three dimensional unit ball $(B_1(0), \partial B_1(0))$? 
\item Are the past communication cones of $\bpt p$ spatially bounded? Do they shrink down to a point, as $\bpt p$ and $t_0$ go towards the singularity?
\end{enumerate}
The first question is formulated completely independently of the foliation of the spacetime into space-like hypersurfaces. The second question depends on the foliation, but is at least easy to clearly state. The third question is not clearly stated here, because it requires us to choose a way of comparing spatial extents of the communication cones at different times. The subtleties of this are discussed in Section \ref{sect:gr-phys-interpret}. 

Since this work is concerned only with proving affirmative answers to these questions, we will conflate them: We say that a solution forms a particle horizon if all these questions are answered with ``yes''.

\begin{figure}[hbt]
\centering
\begin{subfigure}[b]{\textwidth}
\includegraphics[width=\textwidth]{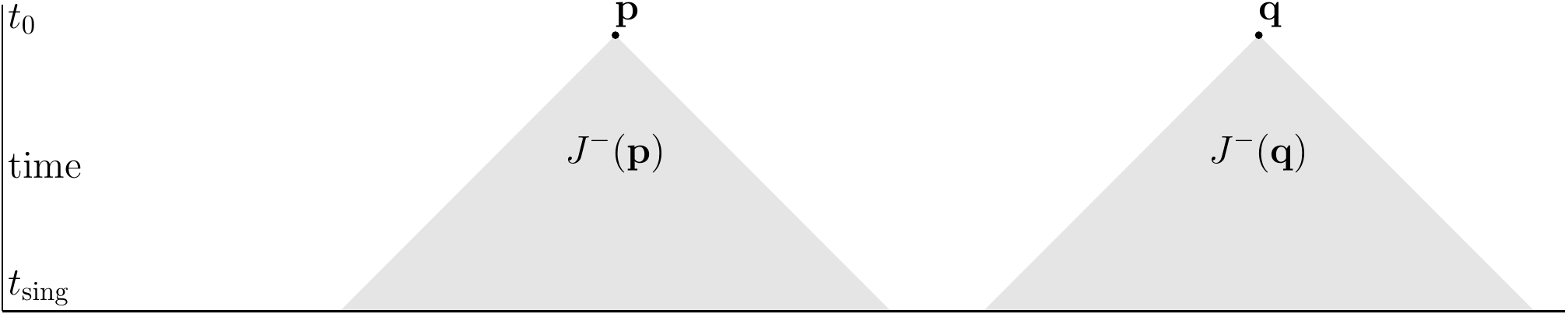}
\caption{Two points $\bpt p$, $\bpt q$ that decouple towards the singularity. Their past light-cones are disjoint.}
\end{subfigure}
\begin{subfigure}[b]{\textwidth}
\includegraphics[width=\textwidth]{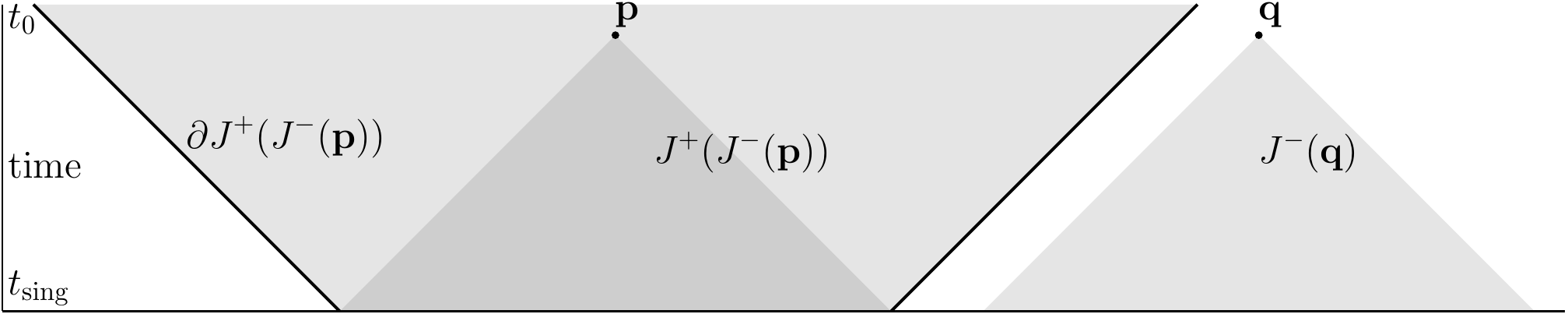}
\caption{Two points $\bpt p$, $\bpt q$ that decouple. The point $\bpt q$ lies outside the communication cone of $\bpt p$.}
\end{subfigure}
\begin{subfigure}[b]{\textwidth}
\includegraphics[width=\textwidth]{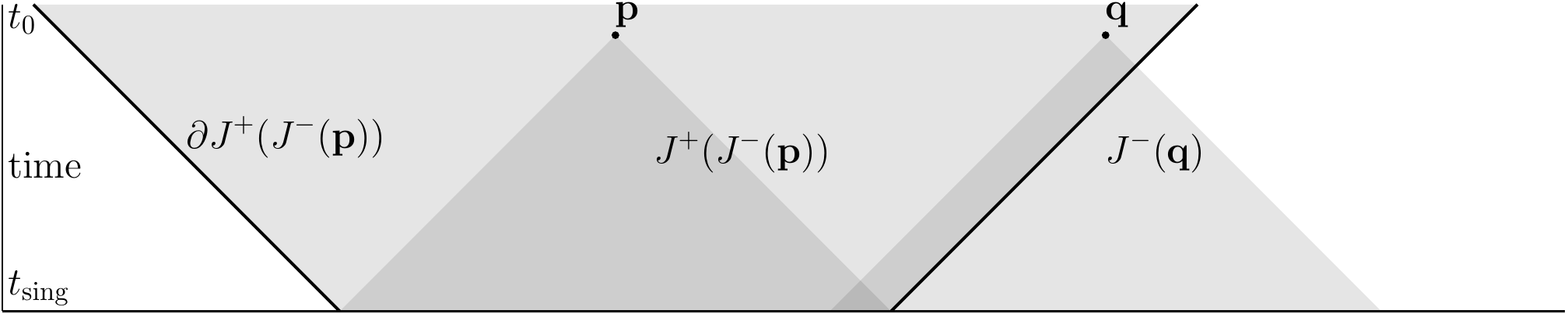}
\caption{Two points $\bpt p$, $\bpt q$ that do not decouple towards the singularity. Their past light-cones have nonempty intersection (the darkest shaded region). The point $\bpt q$ lies inside the communication cone of $\bpt p$.}
\end{subfigure}
\caption{Examples of decoupling and non-decoupling towards the singularity at $t=t_{\mathrm{sing}}$, and of particle horizons.}
\label{fig:intro:comm-cone}
\end{figure}

Originally, Misner \cite{misner1969mixmaster} suggested that no particle horizons should form in Bianchi \textsc{IX}. This was proposed as a possible explanation of the observed approximate homogeneity of the universe: If the homogeneity is due to past mixing, then different observed points in our current past light-cone must themselves have a shared causal past. Misner later changed his mind to the current consensus intuition that typical Bianchi \textsc{VIII} and \textsc{IX} solutions should form particle horizons. Some more details on this are given in Section \ref{sect:gr-phys-interpret}. Further discussion of these questions can be found in e.g.~\cite[Chapter 5]{wald1984general}, \cite[Chapter 5]{hawking1973large}.

%

\paragraph{The BKL picture.}
Spatially homogeneous spacetimes, and especially the question of particle horizons, play an essential role in the so-called BKL picture (also often called BKL-conjecture). The BKL picture is due to Belinskii, Khalatnikov and Lifshitz (\cite{belinskii1970oscillatory}), and describes generic cosmological singularities in terms of homogeneous spacetimes. This picture roughly claims the following: 
\begin{enumerate}
\item Generic cosmological singularities ``are curvature-dominated'', i.e.~behave like the vacuum case. More succinctly, ``matter does not matter''.
\item Generic cosmological singularities are ``oscillatory'', which means that the directions that get stretched or compressed switch over time.
\item Generic cosmological singularities ``locally behave like'' spatially homogeneous ones, especially Bianchi \textsc{IX} and \textsc{VIII}. By this, one means that:
\begin{enumerate}
\item Different regions causally decouple towards the singularity, i.e.~particle horizons form.
\item If one restricts attention to a single communication cone, then, as time goes towards the singularity, the spacetime can be well approximated by a homogeneous one.
\item Different spatial regions may have different geometry towards the singularity (since they decouple). This kind of behaviour has been described as ``foam-like''.
\end{enumerate}
\end{enumerate}
Boundedness of the communication cones, i.e.~formation of particle horizons, in spatially homogeneous models is a necessary condition for the consistency of the BKL-picture: $(3.a)$ claims that different spatial regions causally decouple towards the big bang, and $(3.b)$ claims that such decoupled regions behave ``like they were homogeneous''; hence, homogeneous solutions should better allow for different regions to spatially decouple.

\paragraph{Previous results.}
One way of viewing the formation of particle horizons is as a race between the expansion of the universe (shrinking towards the singularity) and the eventual blow-up at the singularity. If the blow-up is faster than the expansion, particle horizons form; otherwise, they don't. In the context of the Wainwright-Hsu equations, the question can be boiled down to: \emph{Do solutions converge to $\mathcal A$ sufficiently fast? If yes, then particle horizons form. If not, then the questions of particle horizons may have subtle answers.}

The aforementioned solutions constructed in \cite{liebscher2011ancient, beguin2010aperiodic}, with initial conditions on certain hypersurfaces of codimension one, all converge essentially uniformly exponentially to $\mathcal A$, which is definitely fast enough for particle horizons to form. 

Reiterer and Trubowitz claim in \cite{reiterer2010bkl} that the solutions constructed therein also converge to $\mathcal A$ fast enough for this to happen. The claimed results in \cite{reiterer2010bkl}
are somewhat nontrivial to parse; let us give a short overview: They construct solutions converging to certain parts of the Mixmaster attractor $\mathcal A$. These parts of the Mixmaster attractor have full (one-dimensional) volume, and all these constructed solutions form particle horizons. Claims about full-dimensional measure (or Hausdorff-dimensions, etc) are not made in \cite{reiterer2010bkl}. 

The solutions constructed in \cite{liebscher2011ancient, beguin2010aperiodic} were the first known nontrivial solutions that could be proven to form particle horizons. 
It is still unknown whether there exist nontrivial counterexamples, i.e. non-Taub solutions that fail to form particle horizons.

\paragraph{Main results.}
The first main result of this work extends Ringstr\"om's Bianchi \textsc{IX} attractor theorem (Theorem \ref{farfromA:thm:b9-attract}, c.f.~also \cite{ringstrom2001bianchi, heinzle2009new}) to the case of Bianchi \textsc{VIII} vacuum. It can be summarized in the following:
\begin{hthm}{\ref{thm:local-attractor}, \ref{thm:b9-attractor-global}, \ref{thm:b8-attractor-global} and \ref{thm:b8-attractor-global-genericity}}[Paraphrased Attractor Theorem]
With certain exceptions, solutions in Bianchi \textsc{IX} and \textsc{VIII} vacuum converge to the Mixmaster attractor $\mathcal A$.

Lower bounds on the speed of convergence are given, but are insufficient to ensure the formation of particle horizons.

The dynamics of the exceptional solutions is described. The set of initial conditions corresponding to the exceptional solutions is nongeneric, both in the sense of Lebesgue (it is a set of zero Lebesgue measure) and Baire (it is a meagre set).
\end{hthm}
Apart from the applicability to Bianchi \textsc{VIII}, this also extends the Ringstr\"om's previous result by providing lower bounds on the speed of convergence, and provides a new proof. 

The most important result of this work is the following:
\begin{hthm}{\ref{thm:horizon-formation}}[Almost sure formation of particle horizons]
Almost every solution in Bianchi \textsc{VIII} and \textsc{IX} vacuum forms particle horizons towards the big bang singularity. A more rigorous formulation of the theorem is on page \pageref{thm:horizon-formation}.
\end{hthm}
The question remains open of whether particle horizons form for initial conditions that are generic in the sense of Baire\footnote{A set is called generic in the sense of Baire, if it is \emph{co-meagre}, i.e.~it contains a countable intersection of open and dense sets. Then its complement is called \emph{meagre}. By construction, countable intersections of co-meagre sets are co-meagre and countable unions of meagre sets are meagre. Baire's Theorem states that co-meagre subsets of complete metric spaces are always dense and especially nonempty.}. 
We strongly suspect that the answer is no, i.e.~particle horizons fail to form for a co-meagre set of initial conditions, for reasons which will be explained in future work. 


%
\paragraph{Structure of this work.}
We will give the Wainwright-Hsu equations in Section \ref{sect:eqs}, as well as some notation and transformations that will be needed later on. The most referenced equations are also summarized in Appendix \ref{sect:eq-cheat-sheet}, page \pageref{sect:eq-cheat-sheet}, for easier reference. A derivation of the Wainwright-Hsu equations from the Einstein field equations of general relativity is given Section \ref{sect:append:derive-eq}.

An overview of the dynamical behaviour and some first proofs will be given in Section \ref{sect:farfromA}. Sections \ref{sect:near-A} and \ref{sect:near-taub} will describe in detail two different regimes in the neighborhood of $\mathcal A$; these two descriptions are synthesized into general attractor theorems in Section \ref{sect:global-attract}. 
The measure-theoretic results are all contained in Section \ref{sect:volume-form}. Section \ref{sect:gr-phys-interpret} relates dynamical properties of solutions to the Wainwright-Hsu equations to physical properties of the corresponding spacetime.

\paragraph{Strategy.}
Our analysis of the behaviour of solutions of the Wainwright-Hsu equation is structured around two invariant objects: The Mixmaster-attractor $\mathcal A$ and the Taub-spaces $\mathcal T$. We will measure the distances from these sets by functions $\delta(\bpt x) \sim d(\bpt x, \mathcal A)$ and $r(\bpt x)\sim d(\bpt x, \mathcal T)$.

There exist standard heuristics based on normal hyperbolicity, cf.~\cite{heinzle2009mixmaster}. According to these, solutions near $\mathcal A$: 
  \begin{enumerate}
  \item can be described by the so-called Kasner-map (see Section \ref{sect:farfromA}),
  \item converge exponentially to $\mathcal A$, and 
  \item their associated spacetimes form particle-horizons. 
  \end{enumerate}
These heuristics break down near the Taub-spaces $\mathcal T \cap \mathcal A$, where two eigenvalues pass through zero.

In Section \ref{sect:far-from-taub} we will formally prove the validity of the heuristic description of solutions near $\mathcal A$, that stay bounded away from $\mathcal T$.
More precisely, we will show in Proposition \ref{prop:farfromtaub-main} that for any $\epsilon_T>0$, there exists $\epsilon_d>0$ such that all the previously mentioned hyperbolicity heuristics apply for solutions $\bpt x: [0,T] \to \{\bpt y:\, \delta(\bpt y) <\epsilon_d,\,r(\bpt y) > \epsilon_T\}$.

To provide a complete picture of the dynamics, we still need to control solutions in a neighborhood of $\mathcal T$. This is the goal of Section \ref{sect:near-taub}.
It is well known, that $\mathcal T$ is transient, i.e.~solutions may approach $\mathcal T$ but cannot converge to $\mathcal T$; they must leave a neighborhood of $\mathcal T$ again. 
Every component of $\mathcal T\cap \mathcal A$ consists of two equilibria connected by a heteroclinic orbit $-\taubp \to +\taubp$; the standard heuristics described in Section \ref{sect:far-from-taub} break down at $+\taubp$, but continue to work near $-\taubp$ and near the heteroclinic orbit $-\taubp \to +\taubp$.
Thus, solutions can leave the region controlled by normal hyperbolicity only by approaching first $-\taubp$ and then $+\taubp$.

\label{paragraph:intro:strategy:quotients} The most important quantity in the local analysis near $+\taubp$ is the quotient $\frac{\delta}{r}$ of the distances to $\mathcal A$ and $\mathcal T$. As long as this quotient is small, we can prove exponential decay of $\delta(t)$ (Proposition \ref{prop:neartaub:main}) and slow growth of $r(t)$; hence, the quotient $\frac{\delta}{r}$ continues to decrease.
The structure of this kind of estimate is not surprising; in fact, 
it is trivial (by varying $\epsilon_d$ and $\epsilon_T$) to construct continuous nonnegative functions $\hat\delta$ and $\hat r$ with $\mathcal A=\{\bpt x: \hat\delta(\bpt x)=0\}$ and $\mathcal T=\{\bpt x: \hat r(\bpt x)=0\}$, such that estimates of the above form are true near $+\taubp$.
However, our particular choice of functions $\delta$ and $r$ (explicitly given in Section \ref{sect:polar-coords}) allows the \emph{same} quotient $\frac{\delta}{r}$ to be controlled near $-\taubp$ and near $+\taubp$.
We do not know of any reason to a priori expect this fortuitous fact; it is, however, easily verified by direct calculation.

The analysis from Section \ref{sect:near-taub} fits together with the analysis from Section \ref{sect:far-from-taub}: Solutions near $\mathcal A$ (i.e.~$\delta\ll 1$) that leave regions with $r>\epsilon_T$ to enter the neighborhood of $\mathcal T$ with $r\le\epsilon_T$ must have, at the moment where $r=\epsilon_T$, a very small quotient $\delta/r < \epsilon_T^{-1}\epsilon_d \ll 1$. 
This gives rise to a local attractor Theorem \ref{thm:local-attractor}: If $\widetilde \delta=\max(\delta, \frac{\delta}{r})$ is small enough for some initial condition, then $\widetilde \delta$ converges to zero. Hence, we provide a family of forward invariant neighborhoods of $\mathcal A\setminus \mathcal T$ attracted to $\mathcal A$. 
Using the local attractor Theorem \ref{thm:local-attractor}, it is rather straightforward to adapt the ``global'' (i.e.~away from $\mathcal A$) arguments from \cite{ringstrom2001bianchi} in order to produce the global attractor Theorems \ref{thm:b9-attractor-global} and \ref{thm:b8-attractor-global}.

Note that as in previous works \cite{ringstrom2001bianchi}, a ``global attractor theorem'' does not imply that \emph{all} solutions converge to $\mathcal A$, but rather it describes the exceptions, i.e.~solutions where the local attractor Theorem \ref{thm:local-attractor} does not eventually hold. 
In Bianchi \textsc{IX} models, i.e.~in Theorem \ref{thm:b9-attractor-global}, these exceptions are exactly solutions contained in the lower-dimensional Taub-spaces $\mathcal T$ (originally shown in \cite{ringstrom2001bianchi}; we provide an alternative proof).
In Bianchi \textsc{VIII} models, i.e.~in Theorem \ref{thm:b8-attractor-global}, these exceptions are either contained in the lower-dimensional Taub-spaces $\mathcal T$ or must follow the very particular asymptotics described in Theorem \ref{thm:b8-attractor-global}, case $\textsc{Except}$.
As we show in Theorem \ref{thm:b8-attractor-global-genericity}, this exceptional case $\textsc{Except}$ applies at most for a set of initial conditions, which is meagre (small Baire category) and has Lebesgue measure zero. It is currently unknown 
whether this exceptional case \textsc{Except} is possible at all.

These genericity results (Theorem \ref{thm:b8-attractor-global-genericity}) rely on measure theoretic considerations in Section \ref{sect:volume-form}: We provide a volume-form $\omega_4$, which is expanded under the flow (given in \eqref{eq:omega5-def}, \eqref{eq:omega4-def}). This is in itself not surprising: 
 Hamiltonian systems preserve their canonical volume form. Since the Wainwright-Hsu-equations derive from the (Hamiltonian) Einstein Field equations by an essentially monotonous time-dependent rescaling, we expect to find a volume form $\omega_4$, that is essentially monotonously expanding.
Such an expanding volume form has the useful property that all forward invariant sets must either have infinite or zero $\omega_4$-volume. 

We prove the genericity Theorem \ref{thm:b8-attractor-global-genericity} by noting that the exceptional solutions form an invariant set; using their detailed description, we can show that its $\omega_4$-volume is finite and hence zero. 

The results the formation of particle horizons are also proved in Section \ref{sect:volume-form}.
In our language, the primary question is whether $\int_0^{\infty} \delta(t)\dd t <\infty$. If this integral is finite, then particle horizons form and the singularity is ``local''; this behaviour is both predicted by the BKL picture and required for its consistency. Heuristically, time spent away from $\mathcal T$ helps convergence of the integral (since $\delta$ decays uniformly exponentially in these regions), while time spent near $\mathcal T$ gives large contributions to the integral. Looking at our analysis near $+\taubp$ (Proposition \ref{prop:neartaub:main}), we get a contribution to the integral of order $\frac{\delta}{r^2}$. 
The local attractor Theorem \ref{thm:local-attractor} is insufficient to decide the question of locality, since it can only control the quotient $\frac{\delta}{r}$ (and hence the integral $\int \delta^2\dd t<\infty$). 

However, we can show by elementary calculation that the set $\{\bpt x:\,\delta>r^4 \}$ has finite $\omega_4$-measure. Using some uniformity estimates on the volume expansion, we can infer that the (naturally invariant) set $\textsc{Bad}$ of initial conditions that fail to have $\delta< r^4$ for all sufficiently large times, has finite and hence vanishing $\omega_4$-volume. Therefore, for Lebesgue almost every initial condition, $\delta < r^4$ holds eventually, allowing us to bound the contribution of the entire stay near $+\taubp$ by $\frac{\delta}{r^2}<\sqrt{\delta}$, which decays exponentially in the ``number of episodes near Taub-points'' (also called Kasner-eras).
This gives rise to Theorem \ref{thm:horizon-formation}: Lebesgue almost every initial condition in Bianchi \textsc{VIII} and \textsc{IX} forms particle horizons towards the big bang singularity.

A finer look at the integral even allows us to bound certain $L^p_{\textrm{loc}}(\omega_4)$ integrals in Theorem \ref{thm:horizon-formation-alpha-p}.

{
\paragraph{Acknowledgements}
This work has been partially supported by the Sonderforschungsbereich  ``Raum, Zeit, Materie''. This is a preliminary version of a dissertation thesis. For helpful comments and discussions, I would like to thank, in lexicographic order:
Lars Andersson,
Bernold Fiedler,
Hanne Hardering,
Julliette Hell,
Stefan Liebscher,
Alan Rendall,
Hans Ringstr\"om, and
Claes Uggla.

Comments, especially those pointing out major or minor errors, are particularly welcome.

}

\section{Setting, Notation and the Wainwright-Hsu equations}\label{sect:eqs}
The subject of this work, i.e.,~the behaviour of homogeneous anisotropic vacuum space-times with Bianchi Class A homogeneity under the Einstein field equations of general relativity, can be described by a system of ordinary differential equations, called the Wainwright-Hsu equations \eqref{eq:ode-unpacked}. 
 
In Section \ref{sect:equations}, we will introduce the Wainwright-Hsu ordinary differential equations and various auxiliary quantities and definitions, and provide a rough summary of their dynamics.
Then we transform the Wainwright-Hsu equations into polar coordinates in Section \ref{sect:polar-coords}, which are essential for the analysis in Section \ref{sect:near-taub}.

There are multiple equivalent formulations of the Wainwright-Hsu equations in use by different authors, which differ in sign and scaling conventions, most importantly the direction of time. This work uses reversed time, such that the big bang singularity is at $t=+\infty$. 
The relation of the Wainwright-Hsu equations to the Einstein equations of general relativity will be relegated to Section \ref{sect:append:derive-eq}. The relation between properties of solutions to the Wainwright-Hsu equations and physical properties of the corresponding spacetimes is discussed in Section \ref{sect:append:derive-eq}.
\paragraph{General Notations.}
In this work, we will often use the notation $\bpt x = (x_1,\ldots,x_n)$ in order to emphasize that a variable $\bpt x$ refers to a point and not to a scalar quantity. If we consider a curve $\bpt x(t)$ into a space where different coordinates have names, e.g.~$\bpt x: \RR\to \RR^5=\{(\Sigma_+,\Sigma_-,N_1,N_2,N_3)\}$, then we will in an abuse of notation write $N_1(t)=N_1(\bpt x(t))$ in order to refer to the $N_1$-coordinate of $\bpt x(t)$.

We will use $\pm$ to refer to either $+1$ or $-1$, and different occurrences of $\pm$ are always unrelated, such that e.g.~$(\pm,\pm,\pm)\in\{(+,+,+),(-,+,+), (+,-,+), (-,-,+),\ldots\}$. We will use $*$ to refer to either $+1$,$-1$ or $0$, also such that different occurrences of $*$ are unrelated. 

\subsection{Spatially Homogeneous Spacetimes}\label{sect:equations}
We study the behaviour of homogeneous spacetimes, also called Bianchi-models. These are Lorentz four-manifolds, foliated by space-like hypersurfaces on which a group of isometries acts transitively, subject to the vacuum Einstein Field equations. That is, we assume that we have a frame of four linearly independent vectorfields $e_0=\partial_t, e_1,e_2,e_3$, where $e_1,e_2,e_3$ are Killing fields, with dual co-frame $\dd t, \omega_1,\omega_2,\omega_3$, such that the metric has the form
\[
g = g_{00}(t)\dd t\otimes \dd t + g_{11}(t)\omega_1\otimes \omega_1 + g_{22}(t)\omega_2\otimes \omega_2 + g_{33}(t)\omega_3\otimes \omega_3,
\]
and the commutators (i.e.~the Lie-algebra of the spatial homogeneity) has the form
\[
[e_i,e_j]=\sum_k \gamma_{ij}^k e_k \qquad \gamma_{ij}^k=\signN_k \epsilon_{ijk},
\]
where $\epsilon_{ijk}$ is the usual Levi-Civita symbol ($\epsilon_{ijk}=+1$ if $(ijk)\in \{(123), (231), (312)\}$, $\epsilon_{ijk}=-1$ if $(ijk)\in \{(132), (321), (213)\}$ and $\epsilon_{ijk}=0$ otherwise). The signs $\signN_i\in \{+1,-1,0\}$ determine the Bianchi Type of the cosmological model, according to Table \ref{table:bianchi-types}.

The metric is described by the seven Hubble-normalized variables $H,\widetilde N_i, \Sigma_i$, with $i\in\{1,2,3\}$, according to 
\begin{equation}\label{eq:metric-in-wsh}
g_{00}= - \frac 1 4 H^{-2}\qquad g_{ii}=\frac 1 {48}\frac{H^{-2}}{\widetilde N_j \widetilde N_k},
\end{equation}
where $(i,j,k)$ is always assumed to be a permutation of $\{1,2,3\}$, and subject to the linear and sign constraints
\[\numberthis\label{eq:gauge-constraints}
\Sigma_1+\Sigma_2+\Sigma_3=0,\qquad H < 0, \qquad \widetilde N_i >0 \quad\text{for all }i\in\{1,2,3\}.\]
The variable $H$ corresponds to the Hubble scalar, i.e.~the expansion speed of the cosmological model, i.e.~ the mean curvature of the surfaces $\{t = \textrm{const}\}$ of spatial homogeneity. The ``shears'' $\Sigma_i$ correspond to the trace-free Hubble-normalized principal curvatures (hence, the linear ``trace-free'' constraint). The condition $H<0$ corresponds to our choice of the direction of time: We choose to orient time such that the universe is shrinking, i.e.~the  singularity (big bang) lies in the future; this unphysical choice of time-direction is just for convenience of notation. 

The vacuum Einstein Field equations state that the space-time is Ricci-flat. If we express the normalized trace-free principal curvatures $\Sigma_i$ as time-derivatives of the metric variables $\widetilde N_i$, then the the Einstein Field equations become the Wainwright-Hsu equations \eqref{eq:ode-from-gr}, which are a system of seven ordinary differential equations, subject to one linear constraint equation ($\Sigma_1+\Sigma_2+\Sigma_3=0$) and one algebraic equation, called the Gauss-constraint $G=1$ \eqref{eq:gauge-constraints}): 

\[\numberthis\label{eq:ode-from-gr}\begin{aligned}
H'&=\frac 1 2 (1+2\Sigma^2)H\\
\Sigma_i' &= (1-\Sigma^2)\Sigma_i +\frac 1 2 S_i\\
\widetilde N_i'&=-(\Sigma^2+\Sigma_i)\widetilde N_i\\
 1 &\overset{!}{=} \Sigma^2+N^2 := G,
\end{aligned}\]
where we used the shorthands
\[\begin{aligned}
N_i &:= \hat n_i \widetilde N_i \\
\Sigma^2 &:= \frac{1}{6}(\Sigma_1^2+\Sigma_2^2+\Sigma_3^2) \\
N^2&:= N_1^2+N_2^2+N_3^2 -2(N_1N_2+N_2N_3+N_3N_1)\\
S_i &:= 4\left(N_i(2N_i - N_j -N_k)-(N_j-N_k)^2\right).
\end{aligned}\]
\begin{table}
\fbox{ \[\begin{array}{cclcrclc}
\signN_1 & \signN_2 & \signN_3 \qquad & \text{Bianchi Type} &\qquad
\signN_1 & \signN_2 & \signN_3 \qquad & \text{Bianchi Type}\\
+&+&+\qquad& \text{\textsc{IX}}&\qquad
+&-&+\qquad& \text{\textsc{VIII}}\\
+&+&0&\text{$\textsc{VII}_0$} &
+&-&0& \text{$\textsc{VI}_0$}\\
+&0&0& \text{\textsc{II}} &
0&0&0& \text{\textsc{I}}.
\end{array}\]}
\caption[Bianchi Types]{Bianchi Types of Class A, depending on $\signN_i$ (up to permutation and simultaneous sign-reversal)}
\label{table:bianchi-types}
\end{table}
\subsection{The Wainwright-Hsu equations}
The equation for $H$ is decoupled from the remaining equations. Thus, we can drop the equation for $H$, solve the remaining equations, and afterwards integrate to obtain $H$. Likewise, we can stick with the equations for $N_i$ instead of $\widetilde N_i$, such that $\hat n_i = \mathrm{sign}\, N_i$; for Bianchi-types \textsc{VIII} and \textsc{IX} this already determines the metric, and for the lower Bianchi types we can again integrate afterwards. This yields a standard form of the Wainwright-Hsu equations from \eqref{eq:ode-from-gr}, as used in \cite{heinzle2009mixmaster}, \cite{heinzle2009new}, up to constant factors. The most useful equations are also summarized in Section \ref{sect:eq-cheat-sheet}.

It is useful to solve for the linear constraint $\Sigma_1+\Sigma_2+\Sigma_3=0$, introducing $\bpt \Sigma=(\Sigma_+,\Sigma_-)$ by
\[\numberthis \label{eq:taub-def}\begin{gathered}
\taubp_1 = (-1,0) \qquad
\taubp_2 = \left(\frac 1 2, -\frac 1 2 \sqrt{3}\right) \qquad
\taubp_3 = \left(\frac 1 2, \frac 1 2 \sqrt{3} \right)\\
\Sigma_i = 2\langle \taubp_i, \bpt \Sigma \rangle \qquad \Sigma_+ = -\frac 1 2 \Sigma_1 \qquad \Sigma_- = \frac 1 {2\sqrt 3}(\Sigma_3-\Sigma_2),
\end{gathered}\]
which turns the vacuum Wainwright-Hsu differential equations into a system of five ordinary differential equations on $\RR^5=\{(\Sigma_+,\Sigma_-, N_1,N_2,N_3)\}=\{\bpt{\Sigma}, \bpt{N}\}$, with one algebraic constraint equation \eqref{eq:constraint}. The three points $\taubp_1,\taubp_2,\taubp_3$ are called Taub-points. We will, in an abuse of notation, consider the Taub-points both as points in $\RR^2$, and as points in $\RR^5$ (where all three $N_i$ vanish). The Wainwright-Hsu equations are then given by the differential equations
\begin{subequations}
\label{eq:ode}\begin{align}
N_i' &=  -(\Sigma^2 +2 \langle \taubp_i,\bpt\Sigma\rangle)N_i\label{eq:ode-ni} \\
&=-\left(\left|\bpt\Sigma+\taubp_i\right|^2 -1\right)N_i\label{eq:ode2-ni}\\
\bpt\Sigma' &= N^2 \bpt\Sigma + 2 \left(N_1^2 \taubp_1 + N_2^2\taubp_2 + N_3^2\taubp_3 + N_1N_2\taubp_3 + N_2N_3\taubp_1 + N_3N_1\taubp_2\right)\label{eq:ode2-sigma}\\
&=N^2 \bpt\Sigma + 2 \threemat{\taubp_1}{\taubp_2}{\taubp_3}{\taubp_3}{\taubp_1}{\taubp_2}[\bpt N, \bpt N],
\end{align}\end{subequations}
and the Gauss constraint equation
\begin{equation}\label{eq:constraint}
1\overset{!}{=} \Sigma^2 + N^2 =:G(\bpt x),
\end{equation}
where we used the shorthands
\[
\Sigma^2=\Sigma_+^2+\Sigma_-^2,\qquad N^2 = N_1^2+N_2^2+N_3^2 - 2(N_1N_2+N_2N_3+N_3N_1).
\]
We can unpack these equations with unambiguous notation into 
\begin{subequations}\label{eq:ode-unpacked}\begin{align}
N_1' &=  -(\Sigma^2 -2\Sigma_+)N_1 \\
N_2' &=  -(\Sigma^2 +\Sigma_+ - \sqrt{3}\Sigma_-)N_2 \\
N_3' &=  -(\Sigma^2 +\Sigma_+ + \sqrt{3}\Sigma_-)N_3 \\
\Sigma_+' &= N^2 \Sigma_+ -2 N_1^2  +N_2^2 + N_3^2 + N_1N_2 -2 N_2N_3 + N_1N_3 \\
\Sigma_-' &= N^2\Sigma_- +\sqrt{3}\left(- N_2^2 +N_3^2 + N_1N_2 -N_1N_3  \right),
\end{align}\end{subequations}
which is, up to constant factors, the form of the Wainwright-Hsu equations used in \cite{ringstrom2001bianchi}, \cite{liebscher2011ancient} and \cite{beguin2010aperiodic}.

It is occasionally useful to fully tensorize the Wainwright-Hsu equations, yielding the form
\begin{equation}\label{eq:ode-fulltensor}\begin{aligned}
\bpt N' &= -\langle\bpt \Sigma, \bpt \Sigma\rangle \bpt N - \bpt D[\bpt \Sigma,\bpt N]\\
\bpt \Sigma'&= Q[\bpt N, \bpt N]\bpt \Sigma + \bpt T[\bpt N, \bpt N]\\
G(\bpt \Sigma, \bpt N)&= \langle \bpt \Sigma, \bpt \Sigma\rangle + Q[\bpt N, \bpt N]\overset{!}{=}1,
\end{aligned}\end{equation}
where $Q:\RR^3\times \RR^3\to \RR$ and $\bpt T: \RR^3\times \RR^3\to\RR^2$ and $\bpt D: \RR^2\times \RR^3 \to \RR^3$. 
We write $Q$ as a $3\times 3$-matrix with entries in $\RR$ such that $Q[\bpt N,\bpt M]=\bpt N^T Q\bpt N=N^2$; we write $\bpt T$ as a similar $3\times 3$-matrix with entries in $\RR^2$. We write $\bpt D$ as a $3\times 3$-matrix with entries in $\RR^2$ such that $\bpt D[\bpt \Sigma, \bpt N] = (\bpt D\bpt N)\cdot \bpt \Sigma$, where $\bpt D\bpt N$ is the usual matrix product (with entries in $\RR^2$) and the dot-product is evaluated component wise. Then the tensors $Q,\bpt T, \bpt D$ can be written as

\begin{equation}\label{eq:ode-fulltensor-tensors}\begin{aligned}
Q= \sixmat{1 & -2 & -2\\  &1 & -2 \\ & & 1}\qquad
\bpt T = \sixmat{2\taubp_1 & 2\taubp_3 & 2\taubp_2\\ & 2\taubp_2 & 2\taubp_1\\  &  & 2\taubp_3} \qquad
\bpt D = \sixmat{2\taubp_1 &  & \\  & 2\taubp_2 & \\  &  & 2\taubp_3}.
\end{aligned}\end{equation}

\paragraph{Permutation Equivariance.}
The equations \eqref{eq:ode-from-gr} are equivariant under permutations $\sigma:\{1,2,3\}\to\{1,2,3\}$ of the three indices. This permutation invariance also applies to \eqref{eq:ode} as
\[
(\bpt \Sigma, N_1,N_2,N_3)\to (A_\sigma\bpt\Sigma, N_{\sigma(1)}, N_{\sigma(2)}, N_{\sigma(3)}),
\]
where $A_\sigma:\RR^2\to\RR^2$ is the linear isometry with  $A_\sigma^{T}\taubp_i = \taubp_{\sigma(i)}$. The equations are also equivariant under $(\bpt\Sigma, \bpt N)\to (\bpt \Sigma, -\bpt N)$, as can be seen directly from \eqref{eq:ode-fulltensor}.

\paragraph{Invariance of the Constraint.}
The signs of the $N_i$ are preserved under the flow, because $N_i'$ is a multiple of $N_i$ and therefore $N_i=0$ implies $N_i'=0$. The quantity $G$ is preserved under the flow. This can best be seen from \eqref{eq:ode-fulltensor}:
\[\begin{aligned}
\DD_t G &= 2\langle \bpt \Sigma, \bpt \Sigma'\rangle + Q[\bpt N, \bpt N'] + Q[\bpt N', \bpt N]\\
&= 2Q[\bpt N, \bpt N]\langle \bpt \Sigma, \bpt \Sigma\rangle + 2\bpt \Sigma \cdot \bpt T[\bpt N, \bpt N] - 2 \langle \bpt \Sigma, \bpt \Sigma\rangle Q[\bpt N, \bpt N] \\
&\qquad - Q[\bpt D[\bpt \Sigma, \bpt N], \bpt N] - Q[\bpt N, \bpt D[\bpt \Sigma,\bpt N]]\\
&= \Sigma\cdot \bpt N^T \left[2 \bpt T - \bpt D^T Q - Q \bpt D\right]\bpt N.
\end{aligned}\]
Using $\taubp_1+\taubp_2+\taubp_3=0$ and \eqref{eq:ode-fulltensor-tensors}, it is a simple matter of matrix multiplication to verify that $2\bpt T-\bpt D^TQ-Q\bpt D=0$ and hence $\DD_t G=0$. Therefore, sets of the form $\{\bpt x\in\RR^5: G(\bpt x) = c\}$ are invariant for any $c\in\RR$ and especially for the physical $c=1$. 

The set $\mathcal M=\{\bpt x \in \RR^5: G(\bpt x)=1\}$ is a smooth embedded submanifold; this is apparent from the implicit function theorem, since (if $\bpt x\neq 0$)
\[\begin{multlined}
\frac{1}{2}\dd G = \Sigma_+\dd\Sigma_+ + \Sigma_-\dd \Sigma_- \\+ (N_1-N_2-N_3)\dd N_1 + (N_2-N_3-N_1)\dd N_2 + (N_3-N_1-N_2)\dd N_1\neq 0.
\end{multlined}\]

\paragraph{Named invariant sets.}
There are several recurring important sets, which require names and are listed in Table \ref{table:inv-sets}.
\begin{table}
\fbox{\small \[\begin{aligned}
\mathcal M &= \{\bpt x\in \RR^5: \; G(\bpt x)=1\}&&\qquad\text{the physically relevant Phase-space}\\
\mathcal M_{\signN} &= \{\bpt x\in \mathcal M: \; \mathrm{sign}\,N_i = \signN_i\}&&\qquad\text{a specific octant of the Phase-space}\\
\mathcal K &= \mathcal M_{000}=\{\bpt x\in\mathcal M: \; \bpt N=0\} &&\qquad\text{the Kasner circle}\\
\mathcal A &= \{\bpt x\in\mathcal M: \; \text{at most one }N_i\neq 0\} && \qquad\text{the Mixmaster attractor}\\
\mathcal A_{\signN} &= \overline{\mathcal M_{\signN}}\cap \mathcal A && \qquad\text{a specific octant of $\mathcal A$}\\
\mathcal T_i &=\{\bpt x\in\mathcal M: N_j = N_k,\, \langle \taubp_j,\bpt \Sigma\rangle = \langle \taubp_k,\bpt \Sigma\rangle\}&&\qquad\text{a Taub-space}\\
\mathcal T &=\mathcal T_1 \cup \mathcal T_2 \cup\mathcal T_3 && \qquad\text{all three Taub-spaces}\\
\mathcal{TL}_i &=\{\bpt x\in\mathcal M: N_j = N_k,\,N_i=0,\,\bpt \Sigma = \taubp_i\}\subseteq \mathcal T_i&& \qquad\text{a Taub-line}\\
\mathcal{T}^G_i&=\{\bpt x\in\mathcal M: |N_j| = |N_k|,\, \langle \taubp_j,\bpt \Sigma\rangle = \langle \taubp_k,\bpt \Sigma\rangle\}&&\qquad\text{a generalized Taub-space; }\\
&&&\qquad\text{  only invariant if $\mathrm{sign}\,N_j=\mathrm{sign}\,N_k$}
\end{aligned}\]}
\caption[Named Subsets]{Named subsets. Here $(i,j,k)$ stands for a permutation of $\{1,2,3\}$ and $\signN\in\{+,0,-\}^3$.
All of these sets, except for $\mathcal{T}^G_i$, are invariant.}
\label{table:inv-sets}
\end{table}

The set $\mathcal M$ is invariant because $G$ is a constant of motion. The Taub-space $\mathcal T_i$ is invariant because of the equivariance under exchange of the two other indices $j$ and $k$. The invariance of the Taub-lines $\mathcal{TL}_i$ can be seen by considering \eqref{eq:ode-unpacked} for $i=1$ and applying the permutation invariance for $\mathcal{TL}_2$ and $\mathcal{TL}_3$. The generalized Taub-spaces $\mathcal T^G_i$ are not invariant if $\signN_j\neq\signN_k$. The other sets are invariant because the signs $\signN_i=\mathrm{sign}\,N_i$ are fixed.

Recall that the signs of the $N_i$ correspond to the Bianchi Type of the Lie-algebra associated to the homogeneity of the cosmological model and are given in Table \ref{table:bianchi-types} (up to index permutations).

\paragraph{Auxiliary Quantities.}
The following quantities turn out to be useful later on (where $(i,j,k)$ is a permutation of $(1,2,3)$):
\begin{subequations}\label{eq:delta-r-def}
\begin{align} 
\delta_i &= 2\sqrt{|N_jN_k|} \label{eq:delta-r-def-delta}\\
r_i &= \sqrt{(|N_j|-|N_k|)^2 + \frac{1}{3}\langle\taubp_j-\taubp_k, \bpt\Sigma\rangle^2}\label{eq:delta-r-def-r}\\
\nonumber \psi_i &\quad\text{such that:} \\
r_1\cos\psi_1&=\frac{1}{\sqrt{3}}\langle\taubp_{3}-\taubp_{2}, \bpt \Sigma\rangle=\Sigma_- 
&r_1\sin\psi_1 &= |N_2|-|N_3| \\
r_2\cos\psi_2&=\frac{1}{\sqrt{3}}\langle\taubp_{1}-\taubp_{3}, \bpt \Sigma\rangle=-\frac{\sqrt 3}{2}\Sigma_+ -\frac{1}{2}\Sigma_-
&r_2\sin\psi_2 &= |N_3|-|N_1| \\
r_3\cos\psi_3&=\frac{1}{\sqrt{3}}\langle\taubp_{2}-\taubp_{1}, \bpt \Sigma\rangle=\frac{\sqrt 3}{2}\Sigma_+ -\frac{1}{2}\Sigma_-
&r_3\sin\psi_3 &= |N_1|-|N_2|. 
\end{align}\end{subequations}
The auxiliary products $\delta_i^2$ can be used to measure the distance from the Mixmaster attractor $\mathcal A = \{\bpt x:\, \max_i \delta_i(\bpt x)=0\}$. The $r_i$ and can be used to measure the distance from the generalized Taub-space $\mathcal T_i^G = \{\bpt x:\,r_i(\bpt x)=0\}$, and the $(r_i,\psi_i)$-pairs form polar coordinates around the generalized Taub-spaces.

The products $\delta_i$ obey an especially geometric differential equation, similar to \eqref{eq:ode2-ni}:
\begin{equation}\label{eq:ode2-delta}
\delta_i' =  -\left(\left|\bpt\Sigma-\frac{\taubp_i}{2}\right|^2-\frac{1}{4}\right)\delta_{i}.
\end{equation}

%
%
\subsection{The Wainwright-Hsu equations in polar coordinates}\label{sect:polar-coords}
Near the generalized Taub-spaces $\mathcal T^G_i$, it is possible to use polar coordinates \eqref{eq:delta-r-def}. Without loss of generality we will only transform \eqref{eq:ode} into these coordinates around the Taub-space $\mathcal T^G_1$ (the other ones can be obtained by permuting the indices and rotating or reflecting $\bpt \Sigma$).

The use of polar coordinates near the Taub-spaces $\mathcal T_i$ for Bianchi \textsc{IX}, i.e. $\mathcal M_{+++}$, is by no means novel (c.f.~e.g.~\cite{ringstrom2001bianchi}, \cite{heinzle2009mixmaster}). However, to the best of our knowledge, polar coordinates around the generalized Taub-spaces $\mathcal T^G_i$ have not been used previously in the case where the generalized Taub-space fails to be invariant.

We only use use polar coordinates on $\mathcal M=\{\bpt x:\,\,G(\bpt x)=1\}$.
%
%
\paragraph{Polar Coordinates around the invariant Taub-spaces.}
Consider the case $\mathcal M_{*++}$, where $N_2,N_3>0$ are positive and we are interested in a neighborhood of $\mathcal T_1 = \{\bpt x:\, \Sigma_-=0,\,N_2-N_3=0\}$. The sign of $N_1$ does not significantly matter.

We use the additional shorthands
\[ N_- = N_2-N_3\qquad N_+=N_2+N_3,
\]
such that (with \eqref{eq:delta-r-def}):
\[\begin{aligned}
r_1\ge 0 &:& r_1^2 &= \Sigma_-^2 + N_-^2\\
\psi &:& N_-&= r_1\sin\psi&\Sigma_-&=r_1\cos\psi\\
&&N_+^2 &= N_-^2+\delta_1^2 &N^2&= N_-^2+N_1(N_1-2N_+).
\end{aligned}\]


\noindent This gives us the differential equations (using $\Sigma^2+N^2=1$):
\[\begin{aligned}
N_-' &= (N^2 -1 - \Sigma_+)N_- +\sqrt{3}\Sigma_-N_+\\
N_+'&= (N^2 -1 -\Sigma_+)N_+ +\sqrt{3}\Sigma_-N_-\\
\Sigma_-' &= N^2\Sigma_- - \sqrt{3}N_-\left(N_+ - N_1\right),
\end{aligned}\]
allowing us to further compute
\begin{subequations}\label{eq:neartaub-b9-q}
\begin{align}
\frac{r_1'}{r_1}&= \frac{\Sigma_-\Sigma_-' + N_-N_-'}{r_1^2} 
= N^2 - (\Sigma_++1) \frac{N_-^2}{r_1^2} +\sqrt{3} N_1 \frac{\Sigma_- N_-}{r_1^2}\label{eq:neartaub-b9-q:r}\\
\psi' &=\frac{\Sigma_-N_-' - N_-\Sigma_-'}{r_1^2}
=\sqrt{3}N_+ - (\Sigma_++1)\frac{N_-\Sigma_-}{r_1^2} -\sqrt{3}N_1\frac{N_-^2}{r_1^2}\label{eq:neartaub-b9-q:psi}\\
\frac{\delta_1'}{\delta_1} &= N^2-(\Sigma_++1)\label{eq:neartaub-b9-q:delta}\\
\partial_t\log \frac{\delta_1}{r_1}&= -(\Sigma_+ + 1)\frac{\Sigma_-^2}{r_1^2} -\sqrt{3}N_1\frac{\Sigma_- N_-}{r_1^2}.\label{eq:neartaub-b9-q:delta-r}
\end{align}\end{subequations}
%
%
Near $\taubp_1$, i.e. for $\Sigma_+\approx -1$, we can use $1+\Sigma_+=\frac{\Sigma_-^2+N^2}{1-\Sigma_+}$ in order to rearrange some terms in \eqref{eq:neartaub-b9-q}:
\begin{subequations}\label{eq:neartaub-b9-t}
\begin{align}
\frac{r_1'}{r_1} &= r_1^2\sin^2\psi \frac{-\Sigma_+}{1-\Sigma_+} + N_1\, h_r\label{eq:neartaub-b9-t:r}\\
\frac{\delta_1'}{\delta_1}&= \frac{-1}{1-\Sigma_+}r_1^2\cos^2\psi + \frac{-\Sigma_+}{1-\Sigma_+}r_1^2\sin^2\psi + N_1\,h_\delta \label{eq:neartaub-b9-t:delta}\\
\partial_t \log \frac{\delta_1}{r_1}&=\frac{-1}{1-\Sigma_+}r_1^2\cos^2\psi +N_1(h_\delta-h_r) \label{eq:neartaub-b9-t:delta-r}\\
\psi'&= \sqrt{3}r_1 \sqrt{\sin^2\psi + \frac{\delta_1^2}{r_1^2}} - \frac{r_1^2}{1-\Sigma_+} \sin\psi \cos\psi + N_1\sin\psi\, h_{\psi}, \label{eq:neartaub-b9-t:psi}
\end{align}\end{subequations}
where
\begin{subequations}\label{eq:neartaub-b9-t-h}\begin{align}
h_r &=+\sqrt{3}\frac{\Sigma_-N_-}{r_1^2} + (N_1-2N_+)\left(1+\frac{N_-^2}{(1-\Sigma_+)r_1^2}\right)\\
h_\delta &= (N_1-2N_+)\frac{-\Sigma_+}{1-\Sigma_+}\\
h_\psi &= -\sqrt{3}\sin\psi - \cos\psi \frac{N_1-2N_+}{1-\Sigma_+}.
\end{align}\end{subequations}
Let us point out that $|h_r|,|h_\delta|, |h_\psi| \le 5$ if $|N_1|, |N_+|, |N_-|\le \frac 1 2$ and $\Sigma_+<0$.
%
%
\paragraph{Polar Coordinates around the non-invariant generalized Taub-spaces.} 
Without loss of generality, assume that we are in $\mathcal M_{*+-}$, i.e.~$N_2>0>N_3$, and that we are interested in a neighborhood of $\mathcal T^G_1 = \{\bpt x:\, \Sigma_-=0,\,|N_2|-|N_3|=0\}$.

The assumption $N_2>0>N_3$ is incompatible with $N_3=N_2$; hence, $\mathcal M_{*+-}\cap \mathcal T_1=\emptyset$, which is why we have to work with $\mathcal T^G_1$. Since $\mathcal T^G_1$ is not invariant, we expect the corresponding equations for $r_1'$ and $\psi'$ to become singular near $r_1=0$.

There is a possible symmetry-based motivation for expecting usable polar coordinates around $\mathcal T^G_1$ that will be given at the end of this section. Regardless of this motivation, the usability of such polar coordinates is proven by our use of them.

We will proceed analogous to the case of $\mathcal T_1$, using the same names for quantities which fulfill the same function in this work, such that the definitions of e.g.~$N_+,N_-$ will depend on the signs $\signN_2,\signN_3$. Hence, we introduce shorthands
\[ N_- = |N_2|-|N_3|=N_2+N_3\qquad N_+=|N_2|+|N_3|=N_2-N_3,\]
such that (with \eqref{eq:delta-r-def}):
\[\begin{aligned}
r_1\ge 0 &:& r_1^2 &= \Sigma_-^2 + N_-^2\\
\psi &:& N_-&= r_1\sin\psi&\Sigma_-&=r_1\cos\psi\\
&&N_+^2 &= N_-^2+\delta_1^2 &N^2&= N_+^2+N_1(N_1-2N_-).
\end{aligned}\]

This gives us the differential equations (using $\Sigma^2+N^2=1$):
\[\begin{aligned}
N_-' &= (N^2 -1 - \Sigma_+)N_- +\sqrt{3}\Sigma_-N_+\\
N_+'&= (N^2 -1 -\Sigma_+)N_+ +\sqrt{3}\Sigma_-N_-\\
\Sigma_-' &= N^2\Sigma_- - \sqrt{3}\left(N_-N_+ - N_+N_1\right),
\end{aligned}\]
allowing us to further compute
\begin{subequations}\label{eq:neartaub-b8-q}
\begin{align}
\frac{r_1'}{r_1}&= \frac{\Sigma_-\Sigma_-' + N_-N_-'}{r_1^2} 
= N^2 - (\Sigma_++1) \frac{N_-^2}{r_1^2} +\sqrt{3} N_1 \frac{\Sigma_- N_+}{r_1^2} \label{eq:neartaub-b8-q:r}\\
\psi' &=\frac{\Sigma_-N_-' - N_-\Sigma_-'}{r_1^2}
=\sqrt{3}N_+ - (\Sigma_++1)\frac{N_-\Sigma_-}{r_1^2} -\sqrt{3}N_1\frac{N_-N_+}{r_1^2} \label{eq:neartaub-b8-q:psi}\\
\frac{\delta_1'}{\delta_1} &= N^2-(\Sigma_++1) \label{eq:neartaub-b8-q:delta}\\
\partial_t\log \frac{\delta_1}{r_1}&= -(\Sigma_+ + 1)\frac{\Sigma_-^2}{r_1^2} -\sqrt{3}N_1\frac{\Sigma_- N_+}{r_1^2}.\label{eq:neartaub-b8-q:delta-r}
\end{align}\end{subequations}
\noindent
Near $\taubp_1$, i.e. for $\Sigma_+\approx -1$, we can use $1+\Sigma_+=\frac{\Sigma_-^2+N^2}{1-\Sigma_+}$ in order to rearrange some terms in \eqref{eq:neartaub-b8-q}:

\begin{subequations}\label{eq:neartaub-b8-t}
\begin{align}
\frac{r_1'}{r_1} &= \frac{-\Sigma_+}{1-\Sigma_+}r_1^2\sin^2\psi +\delta_1^2\frac{\cos^2\psi-\Sigma_+}{1-\Sigma_+} + N_1\, h_r \label{eq:neartaub-b8-t:r}\\
\frac{\delta_1'}{\delta_1}&= \frac{-1}{1-\Sigma_+}r_1^2\cos^2\psi + \frac{-\Sigma_+}{1-\Sigma_+}r_1^2\sin^2\psi +\frac{-\Sigma_+}{1-\Sigma_+}\delta_1^2+ N_1\,h_\delta \label{eq:neartaub-b8-t:delta}\\
\partial_t \log \frac{\delta_1}{r_1}&=\frac{-1}{1-\Sigma_+}r_1^2\cos^2\psi - \delta_1^2\frac{\cos^2\psi}{1-\Sigma_+} + N_1(h_\delta-h_r) \label{eq:neartaub-b8-t:delta-r}\\
\psi'&= \sqrt{3}r_1 \sqrt{\cos^2\psi + \frac{\delta_1^2}{r_1^2}} - \frac{r_1^2+\delta_1^2}{1-\Sigma_+} \cos\psi \sin\psi + N_1\sin\psi\, h_{\psi},\label{eq:neartaub-b8-t:psi}
\end{align}\end{subequations}
where
\begin{subequations}\label{eq:neartaub-b8-t-h}\begin{align}
h_r &=-\sqrt{3}\frac{\Sigma_-N_+}{r_1^2} + (N_1-2N_-)\left(1-\frac{N_-^2}{(1-\Sigma_+)r_1^2}\right)\\
%
h_\delta &= (N_1-2N_-)\frac{-\Sigma_+}{1-\Sigma_+}\\
h_\psi &= -\sqrt{3}\sqrt{\sin^2\psi + \frac{\delta_1^2}{r_1^2}} - \sin\psi \frac{N_1-2N_-}{1-\Sigma_+}.
\end{align}\end{subequations}
Let us point out that, if $|N_1|, |N_+|, |N_-|\le \frac 1 2$ and $\Sigma_+<0$ and $\frac{\delta_1}{r_1}\le 1$, then $\frac{N_+}{r_1}\le \sqrt{2}$ and hence $|h_r|,|h_\delta|, |h_\psi| \le 5$.

\paragraph{Motivation for polar coordinates around $\mathcal T_1^G$.}
One possible motivation for a priori expecting useful equations from this approach is by a symmetry argument: The Taub-space is invariant since it is the fixed point space of the reflection $\sigma:(\Sigma_+,\Sigma_-, N_1,N_2,N_3)\to (\Sigma_+,-\Sigma_-, N_1, N_3, N_2)$, and \eqref{eq:ode} is equivariant under $\sigma$. This transformation $\sigma:\mathcal M_{\pm +-}\to \mathcal M_{\pm-+}$ does not map the quadrant $N_2>0>N_3$ into itself. Instead, we have $\mathcal T^G_1$ as the fixed point space of $\widetilde\sigma:(\Sigma_+,\Sigma_-, N_1,N_2,N_3)\to (\Sigma_+,-\Sigma_-, N_1, -N_3, -N_2)$.
The Wainwright-Hsu equations are not equivariant under this reflection $\widetilde\sigma$, and we therefore have no reason to expect the fixed-point space $\mathcal T^G_1$ of $\widetilde \sigma$ to be invariant. However, considering \eqref{eq:ode}, equivariance is only spoiled by terms of the form $N_2N_1$ and $N_3N_1$ changing their signs; hence, we expect $\mathcal T^G_1=\{\bpt x:\, \Sigma_-=0,\,|N_2|-|N_3|=0\}$ to be invariant up to terms of order $|N_2N_1|$ and $|N_3N_1|$. Such terms can be well controlled, as it will turn out in Section \ref{sect:near-taub}.
\section{Description of the Dynamics} \label{sect:farfromA}
We will now give an overview of the behaviour of trajectories of \eqref{eq:ode}. This overview will contain most of the classic results about Bianchi cosmological models. 

Our overview will be organized by first describing the simplest subsets named in Table \ref{table:inv-sets} and then progressing to the higher dimensional subsets, finally describing Bianchi Type \textsc{IX} ($\mathcal M_{+++}$) and Bianchi Type \textsc{VIII} ($\mathcal M_{+-+}$) solutions. Our approach in this section is very similar to \cite{ringstrom2001bianchi} and \cite{heinzle2009new}; unless explicitly otherwise stated, all observations in this section can be found therein.

\paragraph{A very short summary of relevant dynamics.}
The Kasner circle $\mathcal K$ is actually a circle and consists entirely of equilibria. The so-called Mixmaster attractor $\mathcal A$ consists of three $2$-spheres $\{\bpt x:\, \Sigma_+^2+\Sigma_-^2+N_i^2=1\,N_j=N_k=0\}$, which intersect in $\mathcal K$. Only half of these spheres are accessible for any trajectory, since the $\mathrm{sign}\,N_i$ are fixed. For this reason, these half-spheres are also called ``Kasner-caps'', i.e.~$\mathcal A_{+00}$ is the $N_1>0$-cap.

The dynamics on the Kasner-caps will be discussed in Section \ref{sect:lower-b-types}; each orbit in a Kasner-cap is a heteroclinic orbit connecting two equilibria on the Kasner-circle $\mathcal K$.

The long-time behaviour of the lower dimensional Bianchi-types (at least one $N_i=0$) is well-understood: All such solutions converge to an equilibrium $\bpt p \in\mathcal K$ as $t\to\infty$. The behavior in the highest-dimensional Bianchi Types \textsc{IX} and \textsc{VIII} is not yet fully understood, and is therefore of most interest in this work.

It is known (c.f.~\cite{ringstrom2001bianchi}) that Bianchi Type \textsc{IX} solutions that do not lie in a Taub-space converge to the Mixmaster attractor as $t\to+\infty$, i.e.~towards the big bang singularity. 
It has been conjectured that generic Bianchi Type \textsc{VIII} solutions share this behaviour; this will be proven in this work (Theorem \ref{thm:b8-attractor-global} and \ref{thm:b8-attractor-global-genericity}).

The question of particle horizons was already mentioned in the introduction and is further discussed from a physical viewpoint in Section \ref{sect:gr-phys-interpret}. In terms of the Wainwright-Hsu equations, the question can be formulated as (see Section \ref{thm:b8-attractor-global-genericity}, or c.f. e.g. \cite{heinzle2009future}):

\[
\text{Is }\quad
I(\bpt x) =   \max_i \int_{0}^\infty \delta_i(\phi(\bpt x, t))\dd t = 2\max_i\int_0^\infty \sqrt{|N_jN_k|}(t)\dd t<\infty\,?
\]
Here $\phi$ is the flow to \eqref{eq:ode}. The space-time associated to the solution $\phi(\bpt x, \cdot)$ forms a particle horizon if and only if $I<\infty$.

It is known that there exist solutions in Bianchi \textsc{IX} and \textsc{VIII}, where $I<\infty$ (c.f.~\cite{liebscher2011ancient}). It is not known, whether there exist any nontrivial solutions with $I=\infty$ (of course solutions starting in $\mathcal T$ that do not converge to $\mathcal A$ must have $I=\infty$). We prove that, in both Bianchi \textsc{IX} and \textsc{VIII} and for Lebesgue almost every initial condition, particle horizons develop ($I<\infty$) (Theorem \ref{thm:horizon-formation}).

\subsection{Lower Bianchi Types}\label{sect:lower-b-types}
\paragraph{Bianchi Type \textsc{I}: The Kasner circle.}
The smallest, i.e.~lowest dimensional, Bianchi-type is Type \textsc{I}, $\mathcal M_{000}=\mathcal K$, where all three $N_i$ vanish (see Table \ref{table:inv-sets}). By the constraint $N^2+\Sigma^2=1$, we can see that $\mathcal K$ is the unit circle in the $(\Sigma_+,\Sigma_-)$-plane, and consists entirely of equilibria.

The linear stability of these equilibria is given by the following
\begin{lemma}\label{farfromA:lemma:kasnermap-stability}
Let $\bpt p=(\Sigma_+,\Sigma_-,0,0,0)\in \mathcal K$; first consider the case $\bpt p\neq \taubp_i$ and without loss of generality $|\bpt p+\taubp_1|<1$. Then the vectorfield has one central direction given by $\partial_{\mathcal K}=(-\Sigma_-, \Sigma_+,0,0,0)$, one unstable direction given by $\partial_{N_1}=(0,0,1,0,0)$, and two stable directions given by $\partial_{N_2}=(0,0,0,1,0)$ and $\partial_{N_3}=(0,0,0,0,1)$.

The three Taub-points $\taubp_i$ have each one stable direction given by $\partial_{N_i}$ and three center directions given by $\partial_{\mathcal K}$, $\partial_{N_j}$ and $\partial_{N_k}$.
\end{lemma}
\begin{proof}
We first note that the four vectors $\partial_{\mathcal K}$ and $\partial_{N_i}$ form a basis of the tangent space  $T_{\bpt p}\mathcal M = \mathrm{ker}\,\dd G$. 

The stability of an equilibrium is determined by the eigenvalues and eigenspaces of the Jacobian of the vector field; a generalized eigenspace is \emph{central}, if its eigenvalue has vanishing real part, it is \emph{stable} if its eigenvalue has negative real part, and it is \emph{unstable} if its eigenvalue has positive real part.
The Jacobian $\DD f$ of the vector field $f$ given by \eqref{eq:ode} at $\bpt p\in\mathcal K$ is diagonal with three entries of the form $\DD f = \lambda_1 \partial_{N_1} \otimes \dd N_1 + \lambda_{N_2} \partial_2 \otimes \dd N_2 + \lambda_3 \partial_{N_3}\otimes \dd N_3$; we can read off the stability from \eqref{eq:ode2-ni} and Figure \ref{fig:n-discs}.
\end{proof}

\begin{figure}[hbpt]
 \centering
        \begin{subfigure}[b]{0.45\textwidth}
                \centering
                \includegraphics[width=\textwidth]{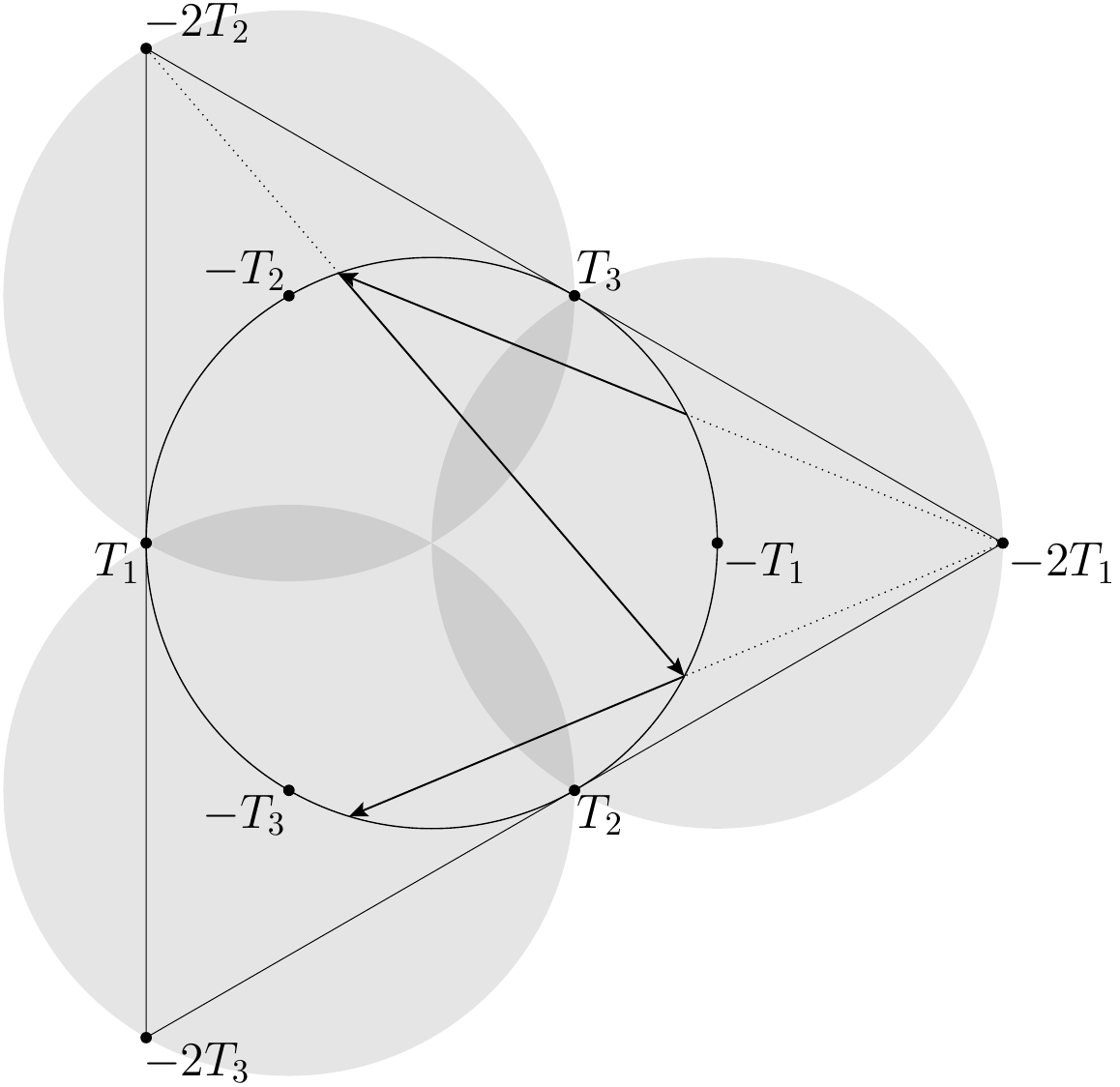}
                \caption{The three discs, where $N_i'/N_i>0$ and a short heteroclinic chain.}\label{fig:short-het}\label{fig:n-discs}
        \end{subfigure}~~%
        \begin{subfigure}[b]{0.45\textwidth}
        \centering
        \includegraphics[width=\textwidth]{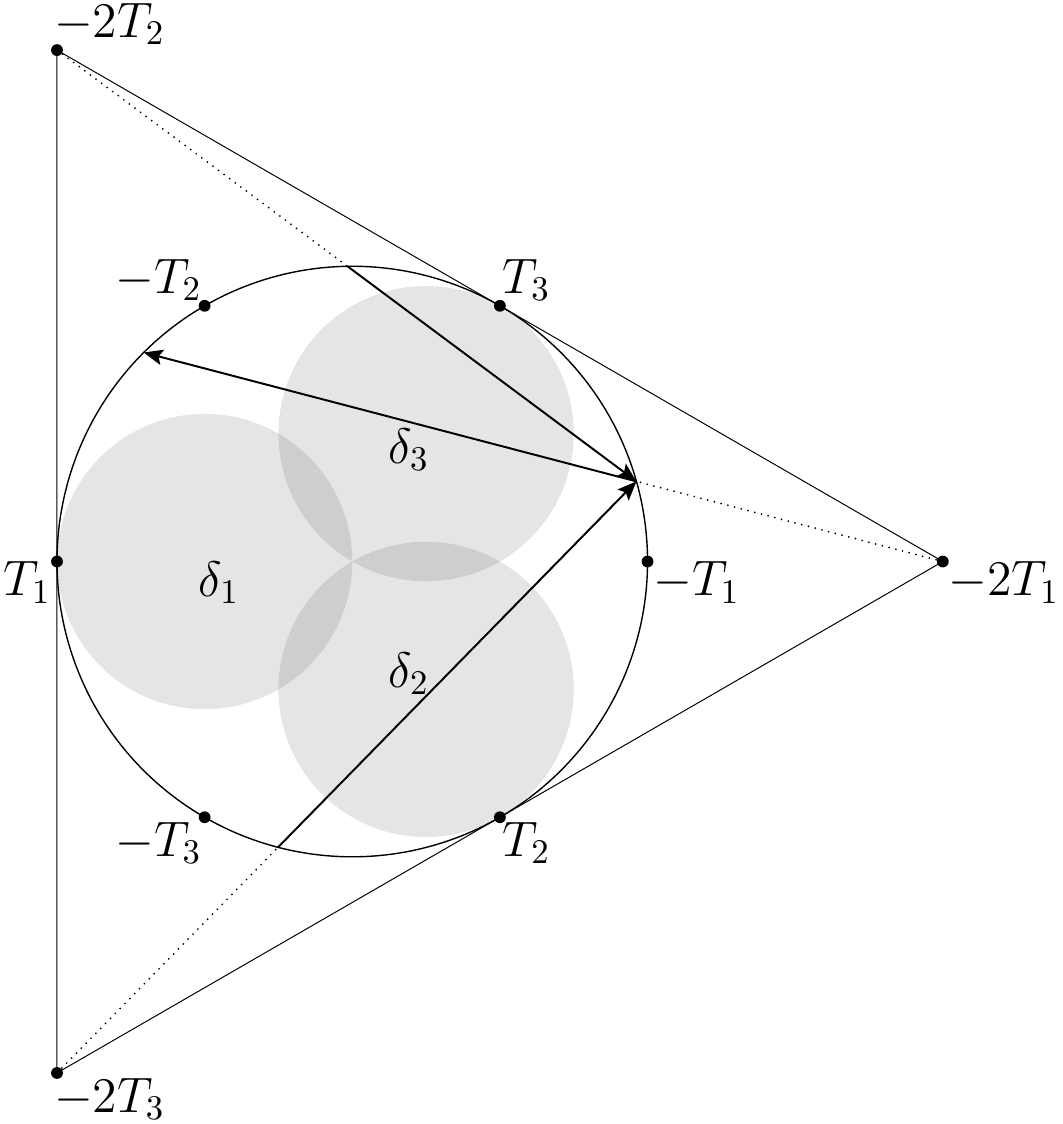}
        \caption{The three discs, where $\delta_i'/\delta_i >0$. Also, the Kasner-map $K$ is a double cover.}\label{fig:double-cover}\label{fig:delta-discs}
        \end{subfigure}
        \caption{Stability properties and the Kasner map. All figures are in $(\Sigma_+, \Sigma_-)$-projection.}\label{fig:kasner-segments}        
\end{figure}

\paragraph{The Taub-line.}
There exists another structure of equilibria, given by $\mathcal{TL}_i:= \{\bpt x\in\mathcal M: N_j = N_k,\,N_i=0,\,\bpt \Sigma = \taubp_i\}$. Up to index permutations and $\bpt N\to -\bpt N$, this set has the form $\mathcal{TL}_1\cup\mathcal M_{++0}=\{\bpt p:\, \bpt p =(-1,0,0,n,n),\, n>0\}$. We observe that this is a line of equilibria. Each such equilibrium has one stable direction (corresponding to $N_1$) and three center directions. 

\paragraph{Bianchi Type \textsc{II}: The Kasner caps.}
Consider without loss of generality the set $\mathcal M_{+00}$. The constraint $G=1$ then reads
\[1=N_1^2 + \Sigma_+^2+\Sigma_-^2,\]
i.e.~the so-called ``Kasner-cap'' $\mathcal M_{+00}$ forms a half-sphere with $\mathcal K$ as its boundary. Considering \eqref{eq:ode}, we can see that $\bpt \Sigma' = (\bpt \Sigma+2\taubp_1)N_1^2$ is a scalar multiple of $\bpt \Sigma+2\taubp_1$; hence, the $\bpt\Sigma$-projection of the trajectory stays on the same line through $-2\taubp_1$. Since $N_1^2\ge 0$, this trajectory is heteroclinic and converges in forward and backward times to the two intersections of this line with the Kasner circle $\mathcal K$, where the $\alpha$-limit is closer to $-2\taubp_1$. Such a trajectory is depicted in Figure \ref{fig:short-het}.

\paragraph{The Mixmaster-Attractor $\mathcal A$.}
The most relevant set for the long-time behavior of the Bianchi system is the Mixmaster attractor $\mathcal A$.
It consists of the union of all six Bianchi $\textsc{II}$ pieces and the Kasner circle. Since the signs of the $N_i$ stay constant along trajectories, it is useful to only study a piece of $\mathcal A$, given without loss of generality by:

\[
\mathcal A_{+++}=\mathcal M_{+00}\cup \mathcal M_{0+0}\cup \mathcal M_{00+}\cup \mathcal K.
\]
The set $\mathcal A_{+++}$ is given by the union of three perpendicular half-spheres (the three Kasner-caps), which intersect in the Kasner-circle. The dynamics in the Mixmaster-Attractor now consists of the Kasner circle $\mathcal K$ of equilibria and the three caps, consisting of heteroclinic orbits to $\mathcal K$. It is described in detail by the so-called Kasner-map.


\paragraph{The Kasner-map $K:\mathcal K\to\mathcal K$.}
We wish to describe which equilibria in $\mathcal K$ are connected by heteroclinic orbits. We can collect this in a relation $K\subseteq \mathcal K\times \mathcal K$, i.e.~we write $\bpt p_- K\bpt p_+$ if either $\bpt p_-=\bpt p_+=\taubp_i$ or there exists a heteroclinic orbit $\gamma:\RR\to \mathcal A$ such that $\bpt p_- = \lim_{t\to-\infty}\gamma(t)$ and $\bpt p_+=\lim_{t\to+\infty}\gamma(t)$.

We can see by Figure \ref{fig:n-discs} (or Lemma \ref{farfromA:lemma:kasnermap-stability}) that each non-Taub point $\bpt p_-$ has a one-dimensional unstable manifold, i.e. one trajectory in $\mathcal A_{\signN}$, which converges to $\bpt p_-$ in backwards time. Therefore, the relation $K$ can be considered as a (single-valued, everywhere defined) map.

This map is depicted in Figure \ref{fig:double-cover}, and has a simple geometric description in the $(\Sigma_+,\Sigma_-)$-projection: 
Given some $\bpt p_-\in\mathcal K$, we draw a straight line through $\bpt p_-$ and the nearest of the three points $-2\taubp_i$. This line has typically two intersections $\bpt p_-$ and $\bpt p_+\in\mathcal K$ with the Kasner-circle, one of which is nearer to $-2\taubp_i$, which is $\bpt p_-$, and one which is further away, which is $\bpt p_+$. At a $\taubp_i$, there are two possible choices of nearest $-2\taubp_j$ and $-2\taubp_k$ and the lines through these points are tangent to $\mathcal K$; we just set $K(\taubp_i)=\taubp_i$. We see from Figure \ref{fig:double-cover} that this map $K:\mathcal K\to\mathcal K$ is continuous and a double cover (i.e.~each point $\bpt p_+$ has two preimages $\bpt p_{-}^1$ and $\bpt p_-^2$, which both depend continuously on $\bpt p_+$).

Looking at Figure \ref{fig:short-het}, we can also see that the Kasner-map is expanding. Hence it is $C^0$-conjugate to either $[z]_{\ZZ}\to [2z]_{\ZZ}$ or $[z]_{\ZZ}\to [-2z]_{\ZZ}$; since it has three fixed points the latter case must apply. Hence we have

\begin{proposition}\label{prop:kasnermap-homeomorphism-class}
There exists a homeomorphism $\psi:\mathcal K \to \RR/3\ZZ$, such that $\psi(\taubp_i)=\left[i\right]_{3\ZZ}$ and
\[
\psi(K(\bpt p))=[-2 \psi(\bpt p)]_{3\ZZ}\qquad\forall \bpt p\in\mathcal K.
\]
\end{proposition}
A formal proof of Proposition \ref{prop:kasnermap-homeomorphism-class} is a digression; for this reason, it is deferred until Section \ref{sect:appendix-kasner-map}, where we give a more detailed description of the Kasner map.

\paragraph{Basic Heuristics near the Mixmaster Attractor.}
Heuristically, the Kasner-map determines the behavior of solutions near $\mathcal A$: 
Consider an initial condition $\bpt x_0\in \mathcal M_{\pm\pm\pm}$ near $\mathcal A$, i.e.~an initial condition where none of the $N_i$ vanish. Then the trajectory $\bpt x(t)$ will closely follow the heteroclinic solution $\gamma_1$ passing near $\bpt x_0$. Let $\bpt p_1$ be the end-point of this heteroclinic; $\bpt x(t)$ will follow $\gamma_1$ and stay for some time near $\bpt p_1$ (since it is an equilibrium). However, if $\bpt p_1\neq \taubp_i$, then one of the $N_i$ directions is unstable; therefore, $\bpt x(t)$ will leave the neighborhood of $\bpt p_1$ along the unique heteroclinic emanating from $\bpt p_1$, and follow it until it is near $\bpt p_2=K(\bpt p_1)$. This should continue until $\bpt x(t)$ leaves the vicinity of $\mathcal A$, which should not happen at all (at least if the name ``Mixmaster Attractor'' is well deserved).

The expansion of the Kasner-map is the source of the (so far heuristic) chaoticity of the dynamics of the Bianchi system: The expansion along the Kasner-circle supplies the sensitive dependence on initial conditions, while the remaining two directions are contracting. Since we study a flow in four-dimensions, Lyapunov exponents need only account for three dimensions, the last one corresponding to the time-evolution.

\paragraph{Bianchi-Types $\textsc{VII}_0$ and $\textsc{VI}_0$.}
There are two Bianchi-types, where exactly one of the three $N_i$ vanishes: Types $\textsc{VII}_0$ and $\textsc{VI}_0$.
 In these Bianchi-Types, we have monotone functions (Lyapunov functions), which suffice to almost completely determine the long-time behaviour of trajectories. Without loss of generality, we focus on the case where $N_1=0$. Then we can write
\[
\Sigma_+' = (1-\Sigma^2)(\Sigma_++1)\qquad 1= \Sigma^2 + (N_2-N_3)^2.
\]
We can immediately see that $\Sigma_+$ is non-decreasing along trajectories; indeed, we must have $\lim_{t\to \pm \infty}\Sigma^2(t) = 1$ for all trajectories. Considering \eqref{eq:ode2-ni} and Figure \ref{fig:n-discs}, we can see that, for $t\to +\infty$ we must either have $\Sigma_+ \ge \frac 1 2$, since both $N_2$ and $N_3$ must be stable, or $\Sigma_+=-1$ all along; then $N_2=N_3$ and the trajectory lies in the Taub-line $\mathcal {TL}_1$. In backward time, we must have $\Sigma_+\to -1$: All other points on the Kasner-circle have one of the two $N_2,N_3$ unstable. These statements can be formalized as
\begin{lemma}\label{lemma:farfromA:b6}
Consider an initial condition $\bpt x_0$ with $\bpt x_0\in\mathcal M_{0-+}$. Then, in forward time, the trajectory $\bpt x(t)=\phi(\bpt x_0, t)$ converges to a point on the Kasner-circle $\lim_{t\to \infty}\bpt x(t)=\bpt p_+\in\mathcal K$ with $\Sigma_+(\bpt p_+)\ge \frac 1 2$. In backwards time, the trajectory converges to the Taub-point $\taubp_1=(-1,0,0,0,0)=\lim_{t\to -\infty}\bpt x(t)$.
\end{lemma}
\begin{lemma}\label{lemma:farfromA:b7}
Consider an initial condition $\bpt x_0$ with $\bpt x_0\in\mathcal M_{0++}\setminus \mathcal{TL}_1$. Then, in forward time, the trajectory $\bpt x(t)=\phi(\bpt x_0, t)$ converges to a point on the Kasner-circle $\lim_{t\to \infty}\bpt x(t)=\bpt p_+\in\mathcal K$ with $\Sigma_+(\bpt p_+)\ge \frac 1 2$. In backwards time, the $\bpt\Sigma$-projection of the trajectory converges to the Taub-point $\taubp_1=(-1,0)=\lim_{t\to -\infty}\bpt\Sigma(\bpt x(t))$. No claim about the dynamics of the $N_i(t)$ for $t\to-\infty$ is made.
\end{lemma}
\begin{proof}
Contained in the preceding paragraph.
\end{proof}
It can be shown that, in the case of Bianchi $\textsc{VII}_0$, i.e.~$\bpt x_0\in\mathcal M_{0++}\setminus \mathcal{TL}_1$, the limit $\lim_{t\to -\infty}\bpt x(t)=\bpt p\in \mathcal{TL}_1\setminus\{\taubp_1\}$ exists and does not lie on the Kasner circle. This claim follows directly from Lemma \ref{lemma:farfromtaub:no-delta-increase-near-taub}; however, the proof of Lemma \ref{lemma:farfromtaub:no-delta-increase-near-taub} is rather lengthy and not required for our main results.

\subsection{Bianchi-Types \textsc{VIII} and \textsc{IX} for large $N$}
As we have seen, the lower Bianchi types do not support recurrent dynamics. This is different in the two top-dimensional Bianchi-types \textsc{VIII} and \textsc{IX}. This section is devoted to describing the behaviour far from $\mathcal A$. 

\begin{lemma}[Long-Time Existence]
Every solution $\bpt x:[0,T)\to \mathcal M$ of $\eqref{eq:ode}$ has bounded $\Sigma^2(t)<C(\bpt x_0)$  for all $t\in[0,T)$ and has unbounded forward existence time (i.e.~no finite-time blow-up occurs towards the future, i.e.~towards the big bang singularity).

The product $|N_1N_2N_3|$ is non-increasing along solutions $\bpt x(t)$, since (from \eqref{eq:ode-unpacked}):
\[ \frac{\dd}{\dd t} \log|N_1N_2N_3| = -3\Sigma^2 \le 0. \]
\end{lemma}

\begin{proof}
The monotonicity of $|N_1N_2N_3|$ is already proven in the statement.

The only way that long-term existence can fail is finite-time blow-up, i.e.~$\lim_{t\to T_{\max}}|\bpt x(t)|=\infty$ for some $0<T_{\max}<\infty$. We cannot exclude this possibility a priori, since the vectorfield given by \eqref{eq:ode} is polynomial. However, it suffices to estimate $|\bpt x'(t)|\le C + C|\bpt x|$, with constants independent of of $t\ge 0$ (but possibly depending on $\bpt x_0$).

We first consider the case of Bianchi Type \textsc{VIII}, without loss of generality $\mathcal M_{-++}$.
Consider a maximal solution $\bpt x:[0,T_{\max})\to \mathcal M_{-++}$. Since $N^2=(N_1-N_2+N_3)^2 -4 N_1N_3 >0$, we can see that $N^2>0$ and (from $1=\Sigma^2+N^2$) that $\Sigma^2<1$ for all times $t<T_{\max}$. From this, we can easily estimate $|\bpt x'|\le C+C|\bpt x|$; therefore, solutions exist for all positive times (i.e.~$T_{\max}=\infty$).
 
Next, we consider the case of Bianchi Type \textsc{IX}, without loss of generality $\bpt x:[0,T_{\max})\to \mathcal M_{+++}$. If $N_2\ge \max(N_1,N_3)$, then we can see
\[N^2=(N_1-N_2+N_3)^2 -4 N_1N_3 \ge -4N_1N_3\ge -4 \left(N_1N_2N_3\right)^{\frac{2}{3}}\left(\frac{N_1N_3}{N_2N_2}\right)^{\frac 1 3}\ge-4 \left(N_1N_2N_3\right)^{\frac{2}{3}}. \]
By permutation symmetry, the above inequality holds regardless of which $N_i$ is largest.
Therefore, we have for all times $t>0$:
\[
N^2(t) \ge -4 \left(N_1N_2N_3\right)^{\frac{2}{3}}(0),\qquad \Sigma^2(t) = 1-N^2(t) \le 1+4 \left(N_1N_2N_3\right)^{\frac{2}{3}}(0).
\]
Unbounded time of existence follows as in the case of Bianchi \textsc{VIII}.
\end{proof}
Next, we show that $|N_1N_2N_3|\to 0$ as $t\to \infty$, and that this convergence is essentially uniformly exponential:
\begin{lemma}[Essentially exponential convergence of $|N_1N_2N_3|$]\label{farfrom-A:lemma:uniform}
For every $C_{N}>0$, there exist constants $C_1,C_0>0$ such that for all trajectories $\bpt x:[t_1,t_2]\to \mathcal M$ with $|N_1N_2N_3|(t_1)< C_N$ we have
\[
|N_1N_2N_3|(t_2) \le |N_1N_2N_3|(t_1) \exp\left(C_0 - C_1(t_2-t_1)\right).
\]
\end{lemma}
\begin{proof}
Consider the function $s:\mathcal M\to [0,\infty)$
\[s(\bpt x) = \int_0^1 \Sigma^2(\phi(\bpt x, t))\dd t,\]
where $\phi:\mathcal M\times \RR\to \mathcal M$ is the flow associated to \eqref{eq:ode}.
Fix some $C_N>0$. We will show that we find some constant $\widetilde C>0$ such that $s(\bpt x)> \widetilde C$ whenever $|N_1N_2N_3|(\bpt x)\le C_N$. From this, we can conclude
\[
\log\frac{|N_1N_2N_3|(t_2)}{|N_1N_2N_3|(t_1)} = -3\int_{t_1}^{t_2}\Sigma^2(t)\dd t \le -3 (s(t_1)+ s(t_1+1) +\ldots) \le -3 (t_2-t_1-1)\widetilde C.
\]
We will prove the estimate $s(\bpt x)> \widetilde C(C_N)$ by contradiction. Assume that we had a sequence $\{\bpt x_n\}$, such that $s(\bpt x_n)\to 0$ and $|N_1N_2N_3|(\bpt x_n)$ is bounded as $n\to\infty$.

At first, we show that this cannot happen in bounded regions of phase-space: We can see from \eqref{eq:ode} that there do not exist any invariant sets with $\Sigma^2=0$ (because then $N^2=1$ and $\Sigma'\neq 0$). Therefore $s(\bpt x)>0$ for all $\bpt x\in\mathcal M$. Since $s$ is continuous, the sequence $\bpt x_n$ cannot converge, and hence cannot be bounded (otherwise, there would be some convergent subsequence).

We can assume without loss of generality that $\Sigma^2(\bpt x_n)\to 0$ (since we assumed $\int_0^1 \Sigma^2(\phi(\bpt x_n, t))\dd t\to 0$). In order to avoid convergent subsequences, we must have $\max_{i}|N_i|(\bpt x_n)\to \infty$ as $n\to\infty$.

Consider first the case of Bianchi \textsc{VIII} with $\bpt x\in\mathcal M_{-++}$. Then 
\[
N^2 = (N_2-N_3)^2+N_1^2 -2 N_1 (N_2+N_3) = 1-\Sigma^2.
\]
All three terms in the middle are non-negative and hence bounded; therefore, we must have $|N_2|,|N_3|\to\infty$ and $|N_1|\to 0$. 
We can write
\[\DD_t\Sigma_+ = (1-\Sigma^2)(\Sigma_++1) + 3N_1(N_2+N_3-N_1)\]
We can estimate the terms involving $N_i$ as 
\[|N_1|(N_2+N_3+|N_1|) \le \frac{C_N}{\max(N_2,N_3)}\to 0\qquad\text{as $n\to\infty$}.\]
Since $\DD_t\log{|N_i|} \le |\bpt \Sigma|(2+|\bpt \Sigma|)$ is bounded, the estimate $N_2,N_3 \gg 1 \gg |N_1|$ and hence $|N_1|(N_2+N_3+|N_1|)\ll 1$ stay valid for at least one unit of time. 
We therefore cannot have $\lim_{n\to\infty}s(\bpt x_n)=0$: $s(\bpt x_n)\approx 0$ is only possible if $\max_{t\in[0,1]}\Sigma^2(\phi(\bpt x_n, t))\approx 0$. This is, however, impossible since then $\Sigma_+'\approx 1$. 

Consider now the case of Bianchi \textsc{IX}, i.e.~$\bpt x\in\mathcal M_{+++}$. Assume without loss of generality that $N_3\ge N_2\ge N_1$. Then we can write
\[
1-\Sigma^2=N^2 = (N_1+N_2-N_3)^2 -4 N_1N_2 \ge  (N_1+N_2-N_3)^2 -4 C_N^{\frac 2 3}.
\]
Therefore, we must have $N_2,N_3\to \infty$ and $N_1\to 0$. Apart from this, the same arguments as for Bianchi \textsc{VIII} apply.
\end{proof}
This result, i.e.~Lemma \ref{farfrom-A:lemma:uniform}, is not as explicitly stated in the previous works \cite{ringstrom2001bianchi, heinzle2009new}, and certainly not as extensively used, but is not a novel insight either. It directly proves that metric coefficients stay bounded, see Section \ref{sect:gr-phys-interpret}.

Using Lemma \ref{farfrom-A:lemma:uniform}, we can quickly see the following:
%
\begin{lemma}[Existence of $\omega$-limits]\label{lemma:farfromtaub:omega-existence}
For any initial condition $\bpt x_0 \in \mathcal M$, the $\omega$-limit set is nonempty, $\omega(\bpt x_0)\neq \emptyset$, i.e.~there exists a sequence of times $(t_n)_{n\in\NN}$ with $\lim_{n\to\infty}t_n=\infty$ such that the limit $\lim_{n\to\infty}\bpt x(t_n)$ exists.
\end{lemma}
\begin{proof}
We begin again by considering the Bianchi \textsc{VIII} case of $\bpt x_0\in\mathcal M_{-++}$. The only way of avoiding the existence of an $\omega$-limit is to have $\lim_{t\to \infty}|\bpt x(t)|=\infty$. As in the proof of Lemma \ref{farfrom-A:lemma:uniform}, this is only possible via $N_2,N_3\to \infty$ and $N_1\to 0$. Then
\[\begin{aligned}
\Sigma_+'(t) &= (1-\Sigma^2)(\Sigma_++1) + 3N_1(N_2+N_3-N_1) \\
& \ge  (1-\Sigma^2)(\Sigma_++1) -9 |N_1N_2N_3|\\
&\ge (1-\Sigma^2)(\Sigma_++1) - C_1 e^{-C_2 t}\\
\frac{\delta_1'}{\delta_1}(t) &= -(\Sigma^2 +\Sigma_+),
\end{aligned}\]
where $\delta_1=2\sqrt{|N_2N_3|}$ as in \eqref{eq:delta-r-def} and \eqref{eq:ode2-delta}.
There are basically two possibilities with regards to the dynamics of $\Sigma_+$: If $\Sigma_+\to -1$ as $t\to+\infty$, then this convergence must happen exponentially, since $(1-\Sigma^2)(\Sigma_++1)\ge 0$. That is, if $\Sigma_+\to -1$, then we must have $|\Sigma_++1|(t) \le \frac{C_1}{C_2}e^{-C_2 t}$ for all sufficiently large times $t>t_0>0$. Then 
\[
\int_{t_0}^{\infty}\partial_t\log \delta_1(t)\dd t = \int_0^\infty -\Sigma_-^2 - \Sigma_+(1+\Sigma_+)\dd t\le 2\frac{C_1}{C_2^2}e^{-C_2 t_0} < \infty,
\]
and hence $\lim_{t\to\infty}\delta_1(t)<\infty$, contradicting our assumption. 

The other option is to have $1+\Sigma_+\not\to-1$ as $t\to\infty$. Then we must have
$1+\Sigma_+(t) > \epsilon$ for some $\epsilon>0$ for all sufficiently large times. Informally, we can see from Figure \ref{fig:delta-discs} that this contradicts $\delta_1\to\infty$. 
Formally, we can say: Since $\Sigma_+$ is bounded, $\int (1-\Sigma^2)\epsilon\dd t<\int \Sigma_+'\dd t + C < \infty$, and hence $\int(1-\Sigma^2)\dd t <\infty$. On the other hand, $\DD_t\log(\delta_1) = -\Sigma^2 -\Sigma_+ \le 1-\Sigma^2 -\epsilon$ and integration shows $\lim_{t\to\infty}\delta_1(t) = 0$, contradicting our assumption $N_2,N_3\to\infty$.

Next, we consider the Bianchi \textsc{IX} case of $\bpt x_0\in\mathcal M_{+++}$. Again, the only way of avoiding the existence of an $\omega$-limit is to have $\lim_{t\to \infty}|\bpt x(t)|=\infty$. As in the proof of Lemma \ref{farfrom-A:lemma:uniform}, this is only possible via $N_2,N_3\to \infty$ and $N_1\to 0$ for some permutation of indices. Using $1-\Sigma^2 \ge |1-\Sigma^2|-8|N_1N_2N_3|^{\frac 2 3}$ and replacing $1-\Sigma^2$ by $|1-\Sigma^2|$, we can use the same arguments as in the Bianchi \textsc{VIII} case.
\end{proof}
\subsection{The Bianchi \textsc{IX} Attractor Theorem}\label{sect:farfromA:attract}
The Mixmaster Attractor was named in the 60s. However, the first proof that $\mathcal A$ actually is an attractor was given in \cite[Theorem $19.2$, page 65]{ringstrom2001bianchi}, and simplified in \cite{heinzle2009new}. We shall state this important result:
\begin{thm}[Classical Bianchi \textsc{IX} Attractor Theorem]\label{farfromA:thm:b9-attract}
Let $\bpt x_0\in \mathcal M_{+++}\setminus \mathcal T$. Then \[\lim_{t\to\infty} \mathrm{dist}(\bpt x(t), \mathcal A)=0.\] Also, the $\omega$-limit set $\omega(\bpt x_0)$ does not consist of a single point.
\end{thm}
The proofs of Theorem \ref{farfromA:thm:b9-attract} given in \cite{ringstrom2001bianchi} and \cite{heinzle2009new} require some subtle averaging arguments (summarized as Lemma \ref{lemma:farfromtaub:no-delta-increase-near-taub}), which are lengthy and fail to generalize to the case of Bianchi \textsc{VIII} initial data. We will now give the first steps leading to the proof of Theorem \ref{farfromA:thm:b9-attract}, up to the missing averaging estimates for Bianchi \textsc{IX} solutions. Then, we will state the missing estimates and show how they prove Theorem \ref{farfromA:thm:b9-attract}. Afterwards, we will give a high-level overview of how we replace Lemma \ref{lemma:farfromtaub:no-delta-increase-near-taub} in this work. Nevertheless, for the sake of completeness, we provide a proof of Lemma \ref{lemma:farfromtaub:no-delta-increase-near-taub} in Section \ref{sect:neartaub:forbidden-cones}.
\paragraph{Rigorous steps leading to Theorem \ref{farfromA:thm:b9-attract}.}
We first show that solutions cannot converge to the Taub-line $\mathcal{TL}_i$, if they do not start in the Taub-space $\mathcal T_i$:
\begin{lemma}[Taub Space Instability]\label{lemma:farfromA:taub-instability}
Let $\epsilon>0$ small enough. Then there exists a constant $C_{r,\epsilon}\in(0,1)$, such that the following holds:

Recall the definition of $r_1$, which measures the distance to $\mathcal T_1$ (see \eqref{eq:delta-r-def}). For any piece of trajectory $\bpt x:[t_1,t_2]\to \{\bpt x\in\mathcal M_{*++}: |\bpt \Sigma(\bpt x)-\taubp_1|\le \epsilon\}$, the following estimate holds:
\[\begin{aligned}
 r_1(\gamma(t_2))&\ge C_{r,\epsilon}h(\gamma(t_1)) r_1(\gamma(t_1)),\qquad \text{ where}\\
 h(\bpt x) &= |N_1| + |N_1|^2 + |N_1N_2N_3|.
 \end{aligned}\]
\end{lemma}
\begin{proof}
This Lemma uses the polar coordinates introduced in Section \ref{sect:polar-coords}. We use \eqref{eq:neartaub-b9-t} to see
\[\begin{aligned}
\DD_t \log r_1&\ge r_1^2\sin^2\psi \frac{-\Sigma_+}{1-\Sigma_+} - C|N_1| - C |N_1^2| - C|N_1N_2N_3|,\\
\DD_t \log |N_1|&\approx -3 < -1,\\
\DD_t \log|N_1N_2N_3|&\approx -3 < -1.
\end{aligned}\]
The desired estimate follows by integration.
\end{proof}

Together with Lemma \ref{lemma:farfromtaub:omega-existence}, this allows us to see that there exist $\omega$-limit points on $\mathcal K$:
\begin{lemma}\label{farfromA:lemma:limitpoint-on-A}
Let $\bpt x_0\in\mathcal M_{+++}\setminus \mathcal T$. Then there exists at least one $\omega$-limit point 
$\bpt p\in \left(\mathcal K\setminus\{\taubp_1, \taubp_2,\taubp_3\}\right)\cap \omega(\bpt x_0)$.

Let $\bpt x_0\in\mathcal M_{+-+}\setminus \mathcal T_2$. Then there exists at least one $\omega$-limit point 
$\bpt p\in \left(\mathcal K\setminus\{\taubp_2\}\right)\cap \omega(\bpt x_0)$.
\end{lemma}
\begin{proof}
We already know that there exists an $\omega$-limit point $\bpt y\in\omega(\bpt x_0)$; this point must have $|N_1N_2N_3|=0$ and hence be of a lower Bianchi type. In view of Lemma \ref{lemma:farfromA:b6} and Lemma \ref{lemma:farfromA:b7} and the fact that both $\alpha(\bpt y)\subseteq \omega(\bpt x_0)$ and $\omega(\bpt y)\subseteq \omega(\bpt x_0)$, it suffices to exclude the case where $\omega(\bpt x_0)\subseteq \mathcal T\mathcal L_i\setminus \mathcal K$ for some $i$. This possibility is excluded by Lemma \ref{lemma:farfromA:taub-instability}.
\end{proof}
Therefore, we know that $\liminf_{t\to\infty}\max_i\delta_i(t)=0$ and hence $\liminf_{t\to\infty}\mathrm{dist}(\bpt x(t), \mathcal A)=0$ for initial conditions in $\mathcal M_{\pm\pm\pm}\setminus \mathcal T$. While we presently lack the necessary estimates to prove the missing part of the attractor theorem, $\limsup_{t\to\infty}\max_i\delta_i(t)=0$, we can at least describe how this may fail: Each $\delta_i$ can only grow by a meaningful factor in the vicinity of a Taub-point $\taubp_i$:
\begin{lemma}\label{lemma:farfromtaub:delta-increase-only-near-taub}
There exists a constant $C>0$ such that, given $\epsilon_+>0$, we find $\epsilon_{123}=\epsilon_{123}(\epsilon_+)>0$, the following holds:

Suppose we have a piece of trajectory $\bpt x:[t_1,t_2]\to \mathcal M_{\pm\pm\pm}$, such that:
\begin{enumerate}
\item We have the product bound $|N_1N_2N_3|(t)<\epsilon_{123}$ for all $t\in[t_1,t_2]$
\item The first point of the partial trajectory is bounded away from the Taub-line $\mathcal {TL}_1$, i.e.~$1+\Sigma_+(t_1)>\epsilon_+$
\item The piece of trajectory has comparatively large $\delta_1$, in the sense $1\ge\delta_1^4(t)\ge|N_1N_2N_3|(t)$ for all $t\in[t_1,t_2]$, i.e.~$|N_1|\le 4\delta_1^2=16|N_2N_3|$ for all $t\in[t_1,t_2]$.
\end{enumerate}
Then $\delta_1$ can only increase by a bounded factor along this piece of trajectory, i.e.
\[\delta_1(t_2)\le \exp \left(\frac{C}{\epsilon_+}\right)\delta_1(t_1).\]
\end{lemma}
\begin{proof}
Recall the proof of Lemma \ref{lemma:farfromtaub:omega-existence}. From $\delta_1^4\ge|N_1N_2N_3|$, we know that $|N_1|\le \delta_1^2$ and therefore 
\[\begin{aligned}
\Sigma_+' &= (1-\Sigma^2)(\Sigma_++1) + 3N_1(N_2+N_3-N_1) \\
& \ge  |1-\Sigma^2|(\Sigma_++1) -C \sqrt{|N_1N_2N_3|}.
\end{aligned}\]
If $\epsilon_{123}>0$ is small enough, this allows us to see that $1+\Sigma_+(t)>\frac 1 2 \epsilon_+$ for all $t\in [t_0,t_1]$. 
Since $|\Sigma_+|\le 2$ is bounded, we see that $\int|1-\Sigma^2|\epsilon_+\dd t < C$  and hence $\int |1-\Sigma^2|\dd t < \frac{C}{\epsilon_+}$. On the other hand, $\DD_t \log \delta_1 < (1-\Sigma^2) - \frac 1 2\epsilon_+$, which yields the claim upon integration.
\end{proof}
\paragraph{Sketch of classic proofs of Theorem \ref{farfromA:thm:b9-attract}.}
The previous proofs of Theorem \ref{farfromA:thm:b9-attract}, both in \cite{ringstrom2001bianchi} and \cite{heinzle2009new}, rely on the following estimate (Lemma $3.1$ in \cite{heinzle2009new}, Section $15$ in \cite{ringstrom2001bianchi}):
\begin{hlemma}{\ref{lemma:farfromtaub:no-delta-increase-near-taub}}\label{hlemma:no-delta-increase-near-taub}
We consider without loss of generality the neighborhood of $\mathcal T_1$. Let $\epsilon>0$ small enough. Then there exists a constant $C_{\delta,\epsilon}\in(1,\infty)$, such that, for any piece of trajectory $\gamma:[t_1,t_2]\to \{\bpt x\in\mathcal M_{*++}: |\bpt \Sigma(\bpt x)-\taubp_1|\le \epsilon,\,|N_1|\le 10, |\delta_1|\le 10\}$, the following estimate holds:
\[ \delta_1(\gamma(t_2))\le C_{\delta,\epsilon}\delta_1(\gamma(t_1)). \]
\end{hlemma}
The proof of this Lemma \ref{lemma:farfromtaub:no-delta-increase-near-taub} requires some lengthy averaging arguments and will be deferred until Section \ref{sect:averaging-unneeded}, page \pageref{lemma:farfromtaub:no-delta-increase-near-taub}. We stress that Lemma \ref{lemma:farfromtaub:no-delta-increase-near-taub} is not actually needed for any of the results in this work, and is proven only for the sake of completeness of the literature review.
\begin{remark}
The above formulation of Lemma \ref{lemma:farfromtaub:no-delta-increase-near-taub} includes the Bianchi \textsc{VIII} case of $N_2,N_3>0>N_1$. The versions stated in \cite{heinzle2009new} and \cite{ringstrom2001bianchi} only consider the Bianchi \textsc{IX} case $N_1,N_2,N_3>0$; however, their proofs extend to this case virtually unchanged.
\end{remark}
\begin{proof}[Proof of Theorem \ref{farfromA:thm:b9-attract} using Lemma \ref{lemma:farfromtaub:no-delta-increase-near-taub}]
We begin by showing $\lim_{t\to\infty}\delta_1(t)=0$.
By Lemma \ref{farfromA:lemma:limitpoint-on-A}, we have $\liminf_{t\to\infty}\delta_1(t)=0$. Suppose $\limsup_{t\to\infty}\delta_1(t)>0$.
Then $\delta_1$ must increase from arbitrarily small values up to some finite nonzero infinitely often, and hence we find arbitrarily late subintervals $[T^L,T^R]\subseteq [0,\infty)$ such that $\delta_1^4>|N_1N_2N_3|$ and $\delta_1$ increases by an arbitrarily large factor. This contradicts Lemma \ref{lemma:farfromtaub:delta-increase-only-near-taub} and Lemma \ref{lemma:farfromtaub:no-delta-increase-near-taub}, which basically say that large increases of $\delta_1$ can neither happen with $1+\Sigma_+>\epsilon_+$ nor with $1+\Sigma_+<\epsilon_+$.

The same applies for $\delta_2$ and $\delta_3$. 

Suppose there was only a single $\omega$-limit point. This point cannot lie in $\mathcal K\setminus \{\taubp_i\}$, since at each of these points, at least one of the $N_i$ is unstable. The only remaining possibility is $\omega(\bpt x_0)=\taubp_i$ for some $i\in\{1,2,3\}$, which is excluded by Lemma \ref{lemma:farfromA:taub-instability}.
\end{proof}

\begin{remark}
The above proof also shows that $N_2N_3\to 0$ in the Bianchi \textsc{VIII} case $\bpt x_0\in\mathcal M_{-++}\setminus \mathcal T_1$. This generalization is directly possible while keeping \cite{heinzle2009new} virtually unchanged, even though it has not been explicitly noted therein.

The above proof also shows that in Bianchi $\textsc{VII}_0$, i.e.~for any $\bpt x_0\in\mathcal M_{0++}$,
 we must have $\lim_{t\to-\infty}\bpt x(t)=\bpt p_-$ with $\bpt p_-=(-1,0,0,N,N)$ for some $N>0$. This is false in the case of Bianchi $\textsc{VI}_0$: There we have for any $\bpt x_0\in\mathcal M_{0+-}$ that $\lim_{t\to-\infty}\bpt x(t)=(-1,0,0,0,0)$. Hence, $\delta_1$ can grow by an arbitrarily large factor near $\taubp_1$ in $\mathcal M_{*+-}$, and
no analogue of Lemma \ref{lemma:farfromtaub:no-delta-increase-near-taub} can hold in the Bianchi \textsc{VIII} models $\mathcal M_{*+-}$ and $\mathcal M_{*-+}$.

This difficulty is partially responsible for the fact that, for $\bpt x_0\in\mathcal M_{+-+}$, it was previously unknown whether $\lim_{t\to\infty}N_2N_3(t)\overset{?}{=}0$ and $\lim_{t\to\infty}N_1N_3(t)\overset{?}{=}0$. 
\end{remark}
\paragraph{Sketch of our replacement for Lemma \ref{lemma:farfromtaub:no-delta-increase-near-taub}.}\label{paragraph:farfromA:sketch}
In this work, we will replace the rather subtle averaging estimates from Lemma \ref{lemma:farfromtaub:no-delta-increase-near-taub} by the program outlined in this paragraph. Let us first repeat the reasons, why we want to avoid Lemma \ref{lemma:farfromtaub:no-delta-increase-near-taub}:
\begin{enumerate}
\item The analogue statement of Lemma \ref{lemma:farfromtaub:no-delta-increase-near-taub} in Bianchi \textsc{VIII} is wrong. Lemma \ref{lemma:farfromA:b6} shows that counterexamples to such a generalization can be found by taking any sequence $\{\bpt x_n\}$ of initial data converging to any point in $\mathcal M_{0-+}$. Therefore, any argument relying on Lemma \ref{lemma:farfromtaub:no-delta-increase-near-taub} has no chance of carrying over to the Bianchi \textsc{VIII} case.
\item The proof of Lemma \ref{lemma:farfromtaub:no-delta-increase-near-taub} is lengthy and requires subtle averaging arguments.
\item The complexity of the proof of Lemma \ref{lemma:farfromtaub:no-delta-increase-near-taub} in not unavoidable: Most of the effort is spent trying to understand asymptotic regimes \emph{that do not occur anyway}.
\end{enumerate}
Our replacement is described by the following program:
\begin{enumerate}
\item At first, we study pieces of trajectories $\bpt x:[0,T]\to\mathcal M_{\pm\pm\pm}$, which start near $\mathcal A$ and stay bounded away from the generalized Taub-spaces $\mathcal T^G_i$, i.e.~have all $r_i>\epsilon$. Along such partial solutions, all $\delta_i$ decay essentially exponentially (Proposition \ref{prop:farfromtaub-main}).
\item Next, consider how solutions near $\mathcal A$ can enter the neighborhood of the generalized Taub-spaces. This can only happen near some $-\taubp_i$ (Proposition \ref{prop:farfromtaub-main2}).
\item For such solutions entering the vicinity of $-\taubp_i$, the quotient $\frac{\delta_i}{r_i}$ is initially small, and stays small near $-\taubp_i$ and along the heteroclinic leading to $+\taubp_i$ (Proposition \ref{prop:near-taub:qc}).
\item Next, we study solutions near $\taubp_i$ for which $\frac{\delta_i}{r_i}$ is initially small. Then, $\frac{\delta_i}{r_i}$ stays small. This additional condition ($\delta_i \ll r_i$) allows us to describe solutions with easier averaging arguments and stronger conclusions than Lemma \ref{lemma:farfromtaub:no-delta-increase-near-taub}. Bianchi \textsc{VIII} solutions can be analyzed same way. This is done in Section \ref{sect:neartaub:neartaub}, leading to the conclusion that $\delta_i$ decays essentially exponentially, with nonuniform rate (Proposition \ref{prop:neartaub:main}).
\item Finally, we combine the previous steps in Section \ref{sect:global-attract} in order to prove Theorems \ref{thm:local-attractor}, \ref{thm:b9-attractor-global} and \ref{thm:b8-attractor-global}. These extend Theorem \ref{farfromA:thm:b9-attract} with somewhat finer control over solutions and provide an analogue in Bianchi \textsc{VIII}.
\end{enumerate}
%
%
%

%
\section{Dynamics near the Mixmaster-Attractor $\mathcal A$}\label{sect:near-A}\label{sect:far-from-taub}
Our previous arguments in Section \ref{sect:farfromA} about the dynamics of \eqref{eq:ode} have been of a rather qualitative and global character. We have established that there exist $\omega$-limit points on the Mixmaster-attractor $\mathcal A$.

We have also sketched the classical proof that trajectories converge to $\mathcal A$ in the case of Bianchi Type \textsc{IX} (Theorem \ref{farfromA:thm:b9-attract}) (where we deferred the proof of the crucial estimate Lemma \ref{lemma:farfromtaub:no-delta-increase-near-taub} to a later point). 

In this section, we will give a more precise description of the behaviour near $\mathcal A$. The goal of this section is to show that pieces of trajectories $\gamma:[0,T]\to \mathcal M$ near $\mathcal A$ converge to $\mathcal A$ essentially exponentially, at least as long as they stay bounded away from the Taub-points $\taubp_i$.

\begin{figure}[hbpt]
 \centering
        \begin{subfigure}[b]{0.45\textwidth}
                \centering
                \includegraphics[width=\textwidth]{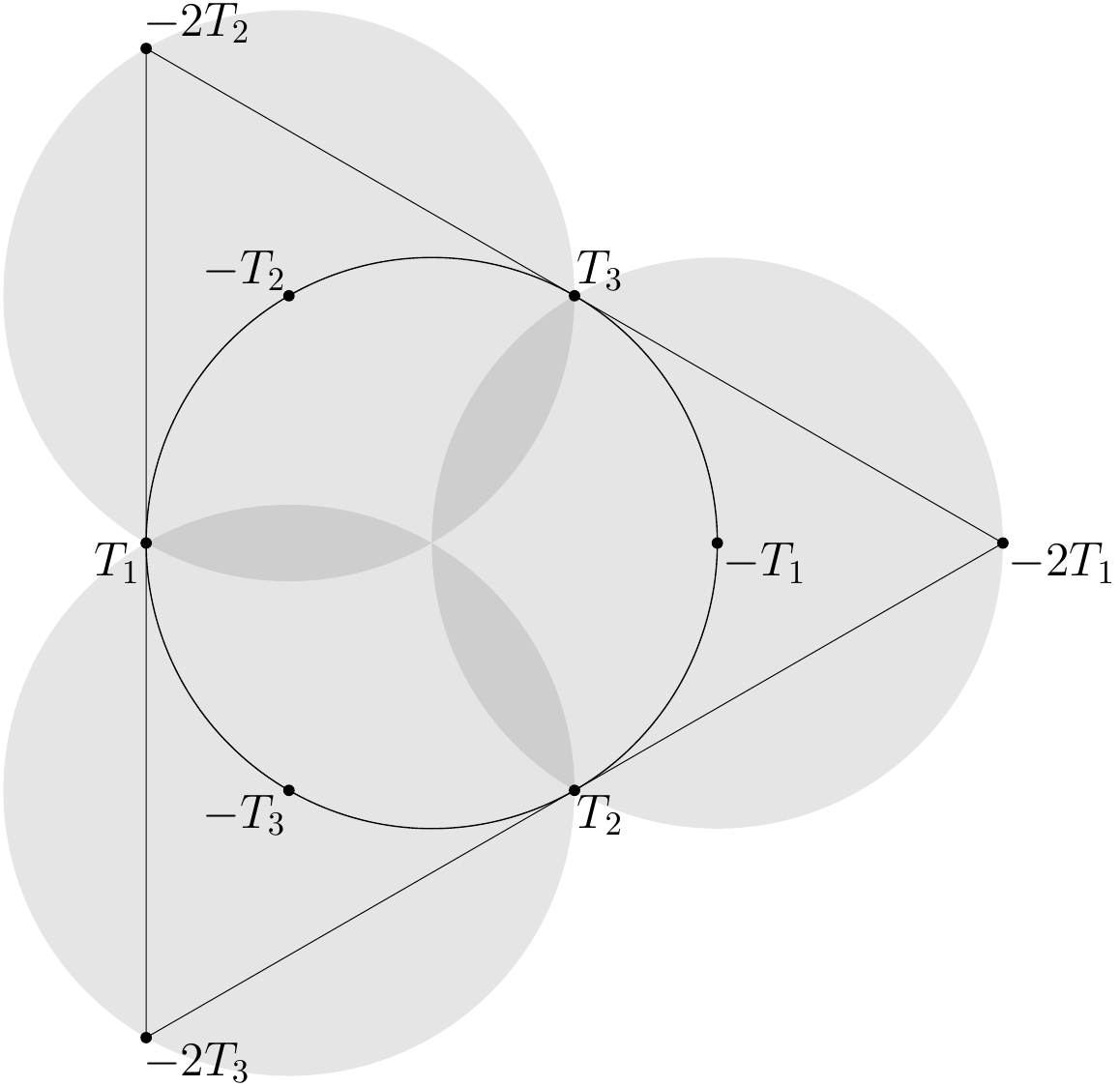}
                \caption{The stability defining discs for the $N_i$}\label{fig:n-discs2}
        \end{subfigure}~~%
        \begin{subfigure}[b]{0.45\textwidth}
                \centering
                \includegraphics[width=\textwidth]{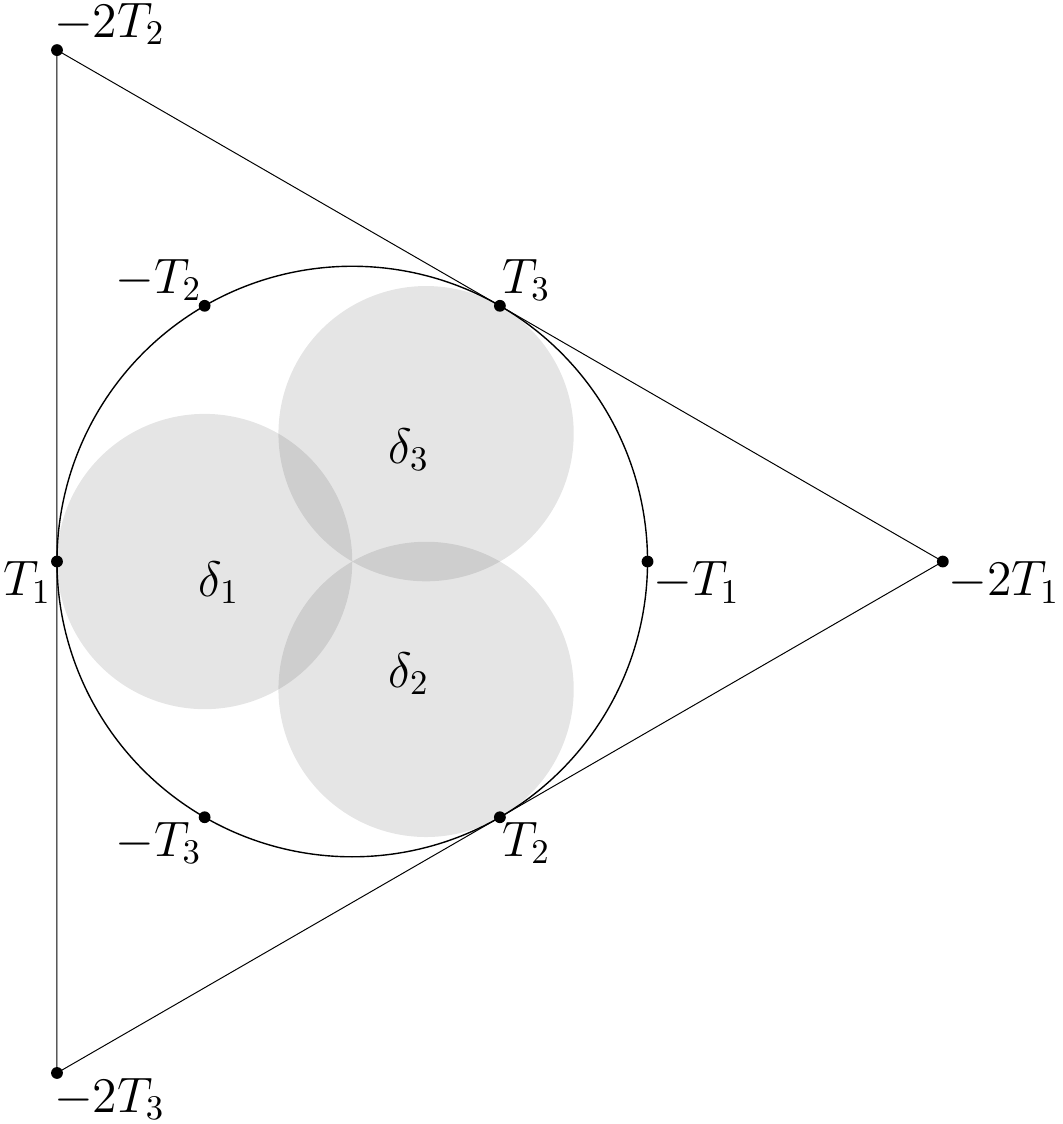}
                \caption{The stability defining discs for the $\delta_i$}\label{fig:delta-discs2}
        \end{subfigure}\\%
        \begin{subfigure}[b]{0.45\textwidth}
        \centering
        \includegraphics[width=\textwidth]{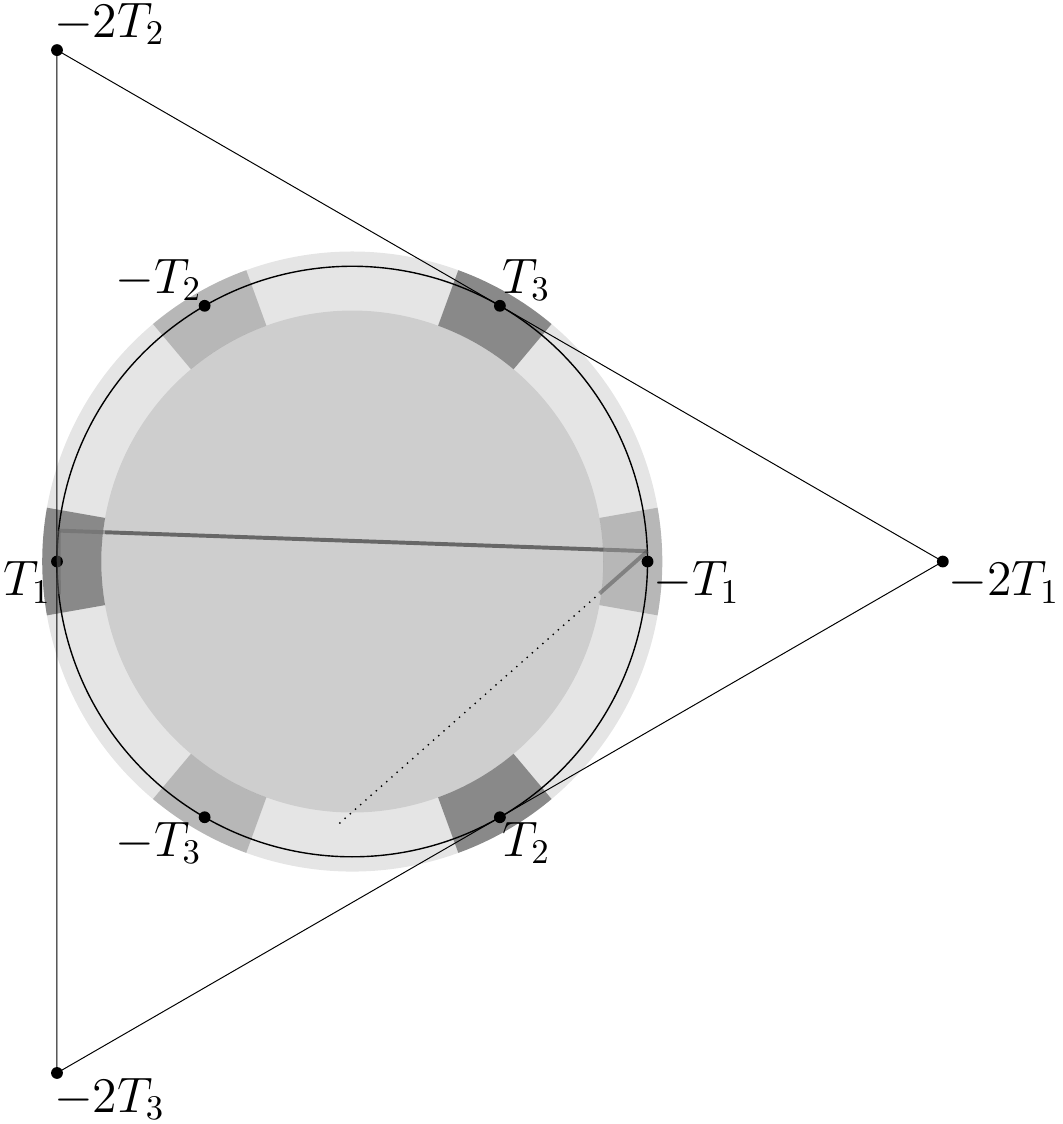}
        \caption{A trajectory leaving the controlled regions}\label{fig:evil-traj}\label{fig:nuclear-evil}
        \end{subfigure}~~%
        \begin{subfigure}[b]{0.4\textwidth}
        \centering
        \includegraphics[width=\textwidth]{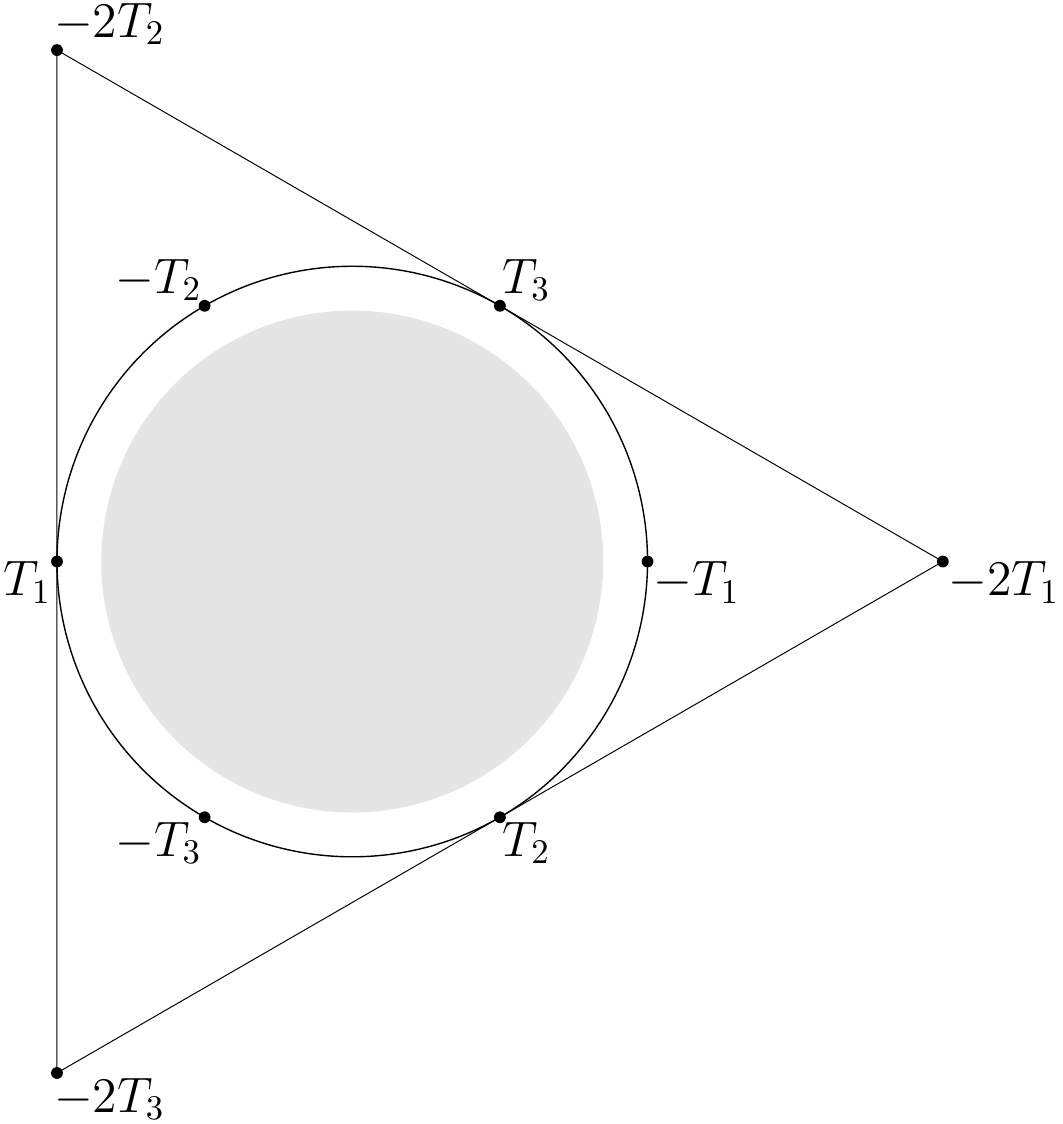}
        \caption{The region \textsc{Cap}}\label{fig:nuclear-cap}
        \end{subfigure}\\%
        \begin{subfigure}[b]{0.4\textwidth}
        \centering
        \includegraphics[width=\textwidth]{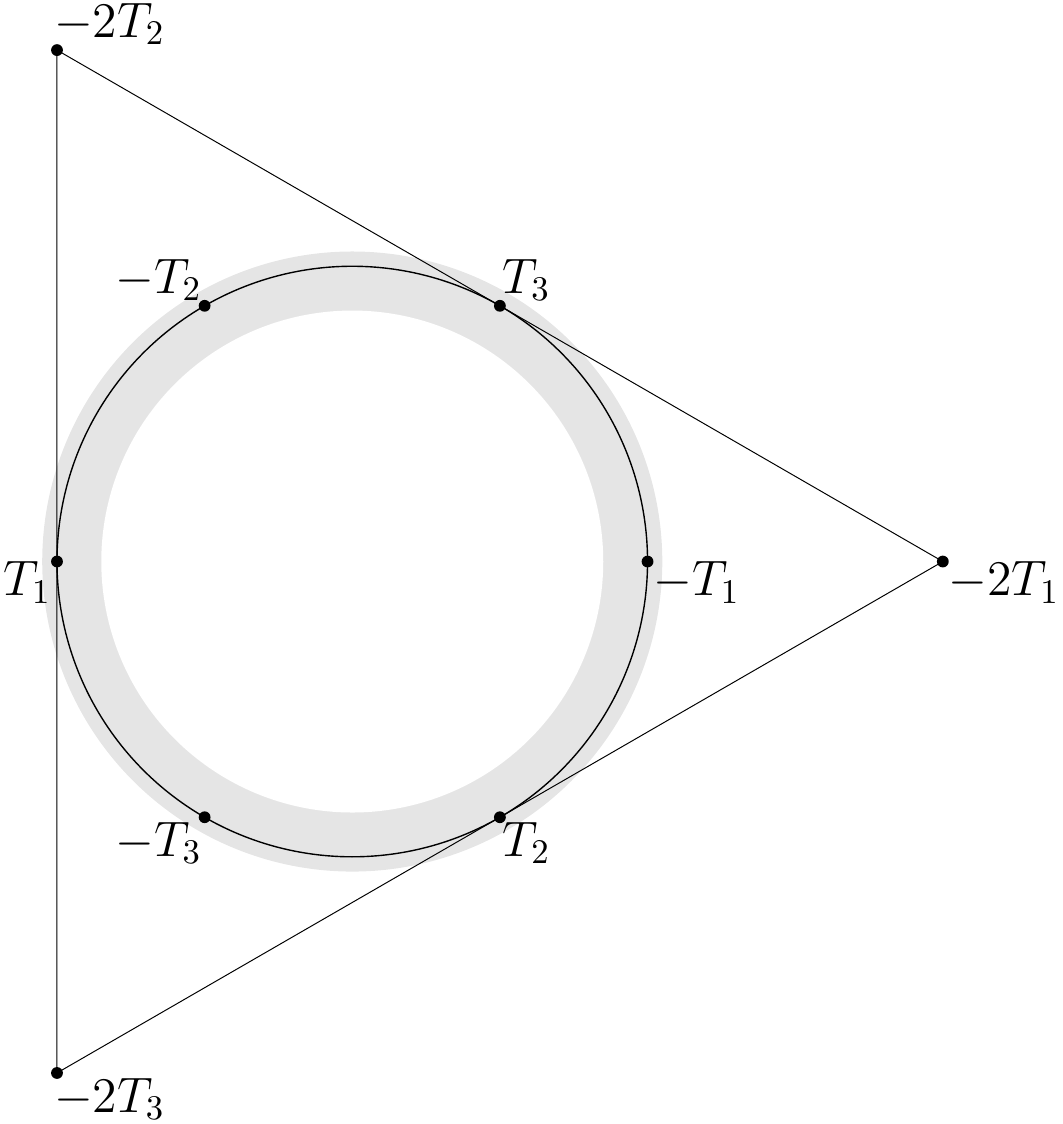}
        \caption{The region \textsc{Circle}}\label{fig:nuclear-circle}        
        \end{subfigure}~~
        \begin{subfigure}[b]{0.4\textwidth}
                \centering
                \includegraphics[width=\textwidth]{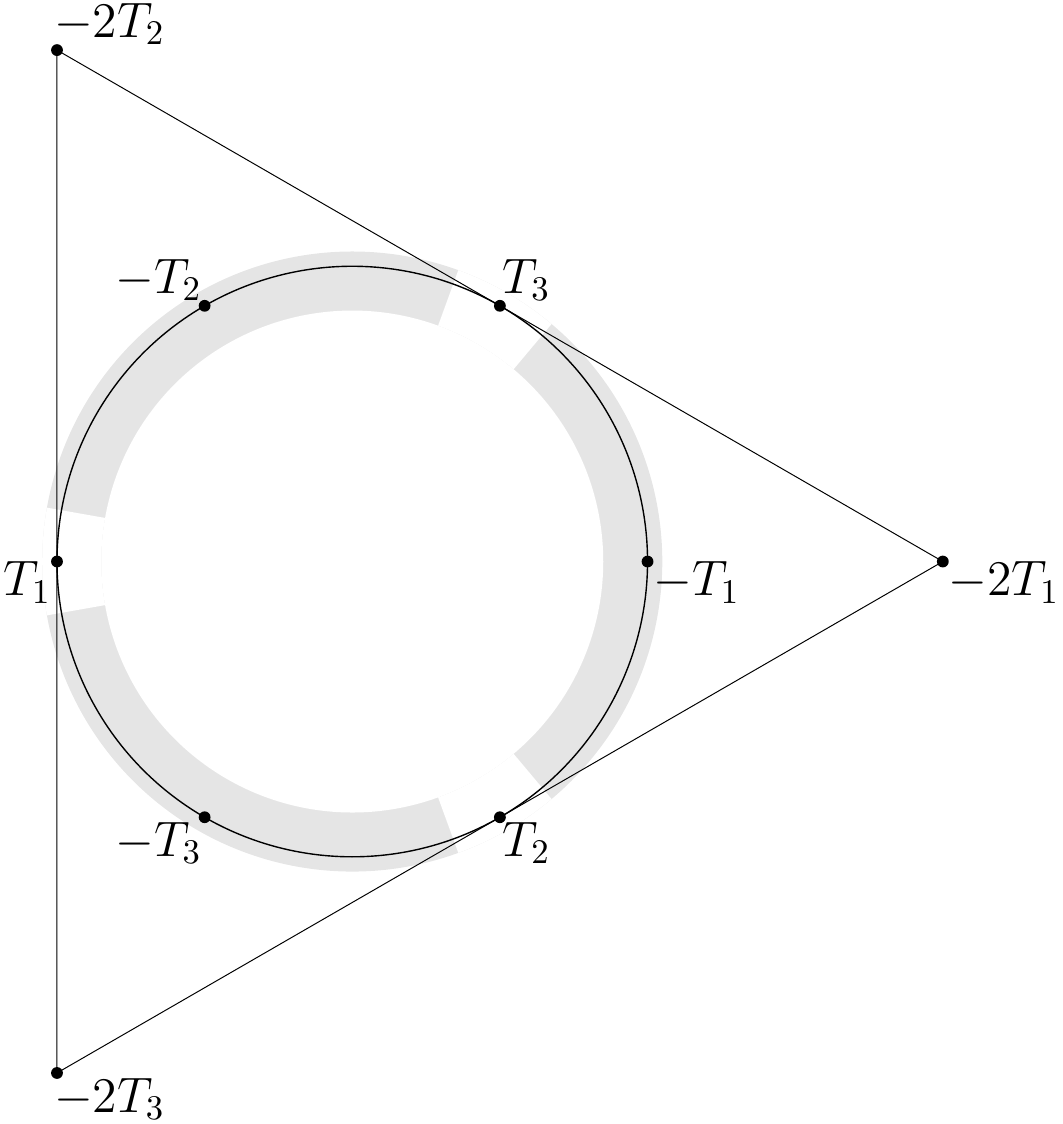}
                \caption{The region \textsc{Hyp}}
                \label{fig:nuclear-hyp}
        \end{subfigure}%
        \caption{The relevant regions are colored in gray and are not up to scale (except for Fig.\ref{fig:n-discs} and Fig.\ref{fig:delta-discs}). See also Fig.\ref{fig:kasner-segments}.}\label{fig:nuclear}        
\end{figure}

The goal of this section is to prove the following two propositions \ref{prop:farfromtaub-main} and \ref{prop:farfromtaub-main2}:
\begin{proposition}[Essentially uniform exponential convergence to $\mathcal A$ away from the Taub points]\label{prop:farfromtaub-main}
For any $\varepsilon_T>0$ small enough, there exist constants $\epsilon_d, \epsilon_N, \epsilon_s, C_0>0$ (depending on $\varepsilon_T$) and $C_T=5$ such that the following holds:

\begin{subequations}
Consider a trajectory $\bpt x:[0,T^*)\to\mathcal M_{\pm\pm\pm}$, such that, for all $t\in [0,T^*)$ the following inequalities hold:
\begin{align}
\max_i\delta_i(t)&< \epsilon_d\label{eq:nearA:capdist}\\
\min_i d(\bpt x(t), \taubp_i)&> \varepsilon_T\label{eq:nearA:taubdist}.
\end{align}
Assume further that for the initial condition $\bpt x_0=\bpt x(0)$, the following stronger estimate holds:
\begin{equation}
\max_i\delta_i(\bpt x_0)< C_0^{-1}\epsilon_d 
\end{equation}
\end{subequations}\begin{subequations}
Then, each $\delta_i$ is essentially uniformly exponentially decreasing in $[0,T^*)$, i.e.
\begin{equation}
\delta_i(t_2) \le C_0 e^{-\epsilon_s(t_2-t_1)}\delta_i(t_1)\qquad\forall\,0\le t_1\le t_2<T^*,\quad i\in\{1,2,3\}.\label{eq:nearA:exponential}
\end{equation}
Hence, if $T^*<\infty$ and one of the inequalities \eqref{eq:nearA:taubdist}, \eqref{eq:nearA:capdist} is violated at time $T^*$, it must be \eqref{eq:nearA:taubdist}, and \eqref{eq:nearA:capdist} must still hold at $T^*$.
\end{subequations}
\end{proposition}
Informally, this proposition states that trajectories near $\mathcal A$ converge exponentially to $\mathcal A$, as long as they stay bounded away from the Taub points.
\begin{proposition}\label{prop:farfromtaub-main2}
Assume the setting of Proposition \ref{prop:farfromtaub-main}. There are constants $\epsilon_N=\epsilon_N(\varepsilon_T)>0$ and $C_T=5$ such that additionally the following holds:

\begin{subequations}
Assume that $T^*<\infty$, and that initially 
\begin{equation}\max_i d(\bpt x_0, \mathcal T_i) > C_T \varepsilon_T.\label{eq:nearA:init-far-from-taub}\end{equation}
\end{subequations}\begin{subequations}
 Then the final part of the trajectory preceding $T^*$ must have the form depicted in Figure \ref{fig:evil-traj}, i.e.~there is $\ell\in\{1,2,3\}$ and there are times $0<T_1<T_2<T_3 \le T^*$ (typically $T_3=T^*$) such that
\begin{align}
d(\bpt x(t), -\taubp_\ell)&\le C_T \varepsilon_T &\forall t\in [T_1,T_2]\label{eq:farfromtaub-thm-12}\\
\max_i |N_i(t)|&\ge \epsilon_N &\forall t\in [T_2,T_3]\label{eq:farfromtaub-thm-23}\\
d(\bpt x(t), \taubp_\ell)&\le C_T\varepsilon_T &\forall t\in [T_3, T^*]\label{eq:farfromtaub-thm-3star}.
\end{align}
\end{subequations}
\end{proposition}
Informally, this proposition states that the only way for trajectories near $\mathcal A$ to reach the vicinity of the Taub-points is via the heteroclinic connection $-\taubp_i\to\taubp_i$.

We first give an informal outline of the proofs:
\begin{proof}[Informal proof of Proposition \ref{prop:farfromtaub-main}]
We split the trajectory into time intervals where it is either near $\mathcal K$ (i.e.~$\max_i |N_i|\le \epsilon_N$) or away from $\mathcal K$ (i.e.~$\max_i |N_i|\ge \epsilon_N$). 

Near $\mathcal K$, we can see from \eqref{eq:ode2-ni} and \eqref{eq:ode2-delta} that each  $\DD_t \log |N_i|$ and $\DD_t \log \delta_i$ depends only on the $\bpt \Sigma$-coordinates and is positive only on some disc in $\RR^2$. These six discs are plotted in Figure \ref{fig:n-discs}. We can observe that these discs only touch and intersect $\mathcal K$ at the three Taub-points and that near each point $\bpt p\in \mathcal K\setminus \{\taubp_i\}$, exactly one of the $\log|N_i|$ is increasing and the two remaining $\log |N_j|$ and all three $\log \delta_i$ are decreasing. Under our assumption $\min_i d(\gamma(t), \taubp_i)>\varepsilon_T$, this increase and decrease is uniform when $\max_i |N_i|\le \epsilon_N$ and if $\epsilon_N=\epsilon_N(\varepsilon_\taubp)>0$ is small enough. Hence, for any small piece of trajectory $\bpt x:[t_1,t_2]\to \{\bpt x\in\mathcal M:\,\min_i d(\bpt x,\taubp_i)\ge\varepsilon_T,\,\max_i|N_i|\le \epsilon_N\}$, one of the $|N_i|$ is uniformly exponentially increasing, while all three $\delta_i$ and the remaining two $|N_j|,|N_k|$ are uniformly exponentially decreasing, with some rate $2\epsilon_s=2\epsilon_s(\varepsilon_T, \epsilon_N)>0$.

Eventually the trajectory will leave the neighborhood of $\mathcal K$; since we assumed that we are near $\mathcal A$, i.e.~$\max_i \delta_i < \epsilon_d$, this can only happen near one of the Kasner caps (Bianchi type \textsc{II}). By continuity of the flow, the trajectory will follow a heteroclinic orbit until it is near $\mathcal K$ again, and will spend only bounded amount of time $T < C(\epsilon_N)$ for this transit. Hence, all $|N_i|$ and $\delta_i$ can only change by a bounded factor during such a heteroclinic transit.

The time spent near $\mathcal K$ between two heteroclinic transits is bounded below by $\mathcal O(|\log \epsilon_d|)$ (for fixed $\epsilon_N$): Consider an interval $[t_1,t_2]$ spent near $\mathcal K$, where $t_1 >0$. Suppose without loss of generality that initially $|N_1|(t_1)=\epsilon_N$ and that $|N_2|$ is uniformly exponentially increasing, such that $|N_2|(t_2)=\epsilon_N$. Then, we must have $|N_2|(t_1)= \frac{4|N_1|}{\delta_3^2}(t_1)\le C \epsilon_N \epsilon_d^{-2}$, and we must have $t_2-t_1 \ge C |\log \epsilon_d|$. Hence, if $\epsilon_d \ll \epsilon_N$ is small enough, the exponential decrease of the three $\delta_i$ will dominate all contributions from the heteroclinic transits and we obtain an estimate of the form \eqref{eq:nearA:exponential}. 
\end{proof}
\begin{proof}[Informal proof of Proposition \ref{prop:farfromtaub-main2}]
We again use continuity of the flow: Each small heteroclinic ``bounce'' near one of the $\taubp_i$ must increase the distance from $\taubp_i$ by at least some $C_u=C_u(\varepsilon_T)$. By continuity of the flow, each episode with $\max_i |N_i|\ge\epsilon_N$ therefore must increase the distance from $\taubp_i$ by $C_u/2$; near $\mathcal K$, the trajectory is almost constant and $d(\bpt x_0, \taubp_i)$ cannot shrink by more that $C_u/3$. Hence the only way to reach the vicinity of a Taub point is by following the heteroclinic $-\taubp_i\to\taubp_i$.
\end{proof}

The remainder of this section is devoted to making these informal proofs rigorous, i.e.~filling all the gaps and replacing the hand-wavy arguments by formal ones. We begin by naming the regions of the phase-space, where the various estimates hold:

\begin{definition}\label{def:farfromtaub-base-sets}
Given $\varepsilon_\taubp, \epsilon_N, \epsilon_d>0$ (later chosen in this order) we define:
\[\begin{aligned}
{\textsc{Cap}}[\epsilon_N, \epsilon_d]&=\{\bpt x\in \mathcal M: \max|N_i|\ge \epsilon_N, \max_i\delta_i \le \epsilon_d\}\\
{\textsc{Circle}}[\epsilon_N, \epsilon_d]&= \{\bpt x\in \mathcal M: \max |N_i|\le\epsilon_N,\, \max_i\delta_i \le \epsilon_d\}\\
{\textsc{Hyp}}[\varepsilon_\taubp, \epsilon_N, \epsilon_d]&= {\textsc{Circle}}[\epsilon_N, \epsilon_d] \setminus \left[B_{\varepsilon_\taubp}(\taubp_1)\cup B_{\varepsilon_\taubp}(\taubp_2)\cup B_{\varepsilon_\taubp}(\taubp_3) \right].
\end{aligned}\numberthis
\]
\end{definition}
These sets are sketched in Figure \ref{fig:nuclear} (not up to scale). They are constructed such that for appropriate parameter choices:
\begin{enumerate}
\item The union $\textsc{Cap}\cup\textsc{Circle}$ contains an entire neighborhood of $\mathcal A$ (by construction).
\item The region $\textsc{Circle}$ is a small neighborhood of the Kasner circle. This is because of the constraint $1=\Sigma^2+N^2$ and $|N^2|< C \epsilon_N^2$ (see Figure \ref{fig:nuclear-circle}).
\item The region $\textsc{Cap}$ has three connected components, where one of the three $|N_i|\gg 0$, because by $\max_i |N_i|\ge \epsilon_N$ at least one $N_i$ must be bounded away from zero and by $\max_{i}2\sqrt{|N_jN_k|} \le \epsilon_d$ at most one $N_i$ can be bounded away from zero (this only works if $\epsilon_d$ is small enough, depending on $\epsilon_N$).

The region $\textsc{Cap}$ is bounded away from the Kasner circle (see Figure \ref{fig:nuclear-cap}). 
By continuity of the flow, the dynamics in $\textsc{Cap}$ can be approximated by pieces of heteroclinic orbits in $\mathcal A$, up to uniformly small errors (Lemma \ref{lemma:farfromtaub-cap}).
\item The region $\textsc{Hyp}$ has three connected components. In each connected component, one of the $|N_i|$ is uniformly exponentially increasing and the remaining two $|N_j|$, $|N_k|$ are uniformly exponentially decreasing. All three products $\delta_i$ are uniformly exponentially decreasing in $\textsc{Hyp}$ (Lemma \ref{lemma:farfromtaub-uniform-hyp}; this only works if $\epsilon_N$ is small enough, depending on $\varepsilon_\taubp$).
\item The remaining part of the neighborhood of $\mathcal A$, i.e.~$\textsc{Circle}\setminus\textsc{Hyp}$, consists of the neighborhoods of the three Taub points. The analysis of the dynamics in these neighborhoods is deferred until Section \ref{sect:near-taub}.
\end{enumerate}
\begin{lemma}[Uniform Hyperbolicity Estimates]\label{lemma:farfromtaub-uniform-hyp}
Given any $\varepsilon_T>0$ small enough, we find $\epsilon_N>0$ and $\epsilon_s>0$ small enough such that, for any $\bpt x\in \textsc{Hyp}[\varepsilon_{\taubp}, \epsilon_N,\infty]$, we find one $i\in\{1,2,3\}$ such that $\DD_t \log |N_i|>2\epsilon_s$, and the remaining two $\DD_t\log|N_j|< -2\epsilon_s$ and all three $\DD_t \log \delta_j <-2\epsilon_s$.

Let $\varepsilon_\taubp, \epsilon_N, \epsilon_s>0$ as above.
For any piece of trajectory $\bpt x:(t_1,t_2)\to \textsc{Hyp}[\epsilon_N,\infty]$, we can conclude
\[\numberthis\begin{aligned}
\int_{t_1}^{t_2}|N_i|\dd t &< \frac{\epsilon_N}{2\epsilon_s}\\
\diam \gamma \le \int_{t_1}^{t_2}|\gamma'(t)|\dd t &\le C_\Sigma\int_{t_1}^{t_2}\max_i|N_i|(t)\dd t\le C_\Sigma\frac{\epsilon_N}{2\epsilon_{s}},
\end{aligned}\]
where we can choose $C_\Sigma= 2$ if $\epsilon_N<0.1$ (from \eqref{eq:ode}).
\end{lemma}
\begin{proof}
The first part of the lemma consists of choosing $\epsilon_s$ and $\epsilon_N$ dependent on $\varepsilon_T$. 
From Equations \eqref{eq:ode2-ni} and \eqref{eq:ode2-delta}, we see that each  $\DD_t \log |N_i|$ and $\DD_t \log \delta_i$ depends only on the $\bpt \Sigma$-coordinates and is positive only on some disc in $\RR^2$. These six discs are plotted in Figure \ref{fig:n-discs}, and only touch or intersect $\mathcal K$ near the three Taub-points, a neighborhood of which is excluded. By the constraint $1-\Sigma^2 = N^2$ and $|N^2|\le 9 \epsilon_N^2$, the set \textsc{Hyp}, depicted in Figure \ref{fig:nuclear}, is near $\mathcal K$, and the desired uniformity estimates hold.

The second part follows from the uniform hyperbolicity in \textsc{Hyp}: In each component of \textsc{Hyp}, exactly one $N_i$ is unstable (see Figure \ref{fig:nuclear} and Figure \ref{fig:n-discs}). Suppose without loss of generality that $N_1$ is the unstable direction; then we can estimate for $t\in (t_1,t_2)$:
\[
|N_1(t)| \le e^{-2\epsilon_s (t_2-t)}|N_1(t_2)|,\qquad |N_2(t)|\le e^{-2\epsilon_s(t-t_1)}|N_2(t_1)|.
\]
For $|N_3|$, the analogous estimate as for $N_2$ holds.
Using $|N_1(t_2)|\le \epsilon_N$ and $|N_2(t_1)|\le \epsilon_N$ and integrating yields the claim about $\int |N|\dd t$. From \eqref{eq:ode}, we see $|\bpt x'|\le C_\Sigma \max_i|N_i|$ (if $\max_i|N_i|\le 1$).
\end{proof} 

Continuity of the flow allows us to approximate solutions in \textsc{Cap} by heteroclinic solutions in $\mathcal A$, up to any desired precision $\varepsilon_c$, if we only chose the distance from $\mathcal A$ (i.e.~$\epsilon_d>0$) small enough. More precisely:

\begin{lemma}[Continuity of the flow near \textsc{Cap}]\label{lemma:farfromtaub-cap}
Let $\epsilon_N>0$ and $\varepsilon_c>0$. Then there exists $\epsilon_d=\epsilon_d(\epsilon_N, \varepsilon_c)$ small enough and $\hat C_0=\hat C_0(\epsilon_N)>0$ large enough, such that the following holds:

Let $\bpt x:(t_1,t_2)\to \textsc{Cap}[\epsilon_N, \epsilon_d]$ be a piece of a trajectory. 
Then $t_2-t_1 < \hat C_0$ and there exists $y\in \mathcal A$ such that 
\[\numberthis\label{eq:farfromtaum-cap-continuity} d(\bpt x(t), \phi(y, t-t_1))<\varepsilon_c\qquad \text{for all}\quad t\in(t_1,t_2),\]
where $\phi:\mathcal M\times \RR\to\mathcal M$ is the flow corresponding to \eqref{eq:ode}.
\end{lemma}
\begin{proof}Follows from continuity of the flow and the fact that all trajectories in $\mathcal A$ are heteroclinic and must leave \textsc{Hyp} at some time.
\end{proof}

We now have collected all the ingredients to formally prove the two main results from this section. At first, we 
combine Lemma \ref{lemma:farfromtaub-uniform-hyp} and Lemma \ref{lemma:farfromtaub-cap} in order to show 
that each $\delta_i$ is uniformly essentially exponentially decreasing in $\textsc{Cap}\cup \textsc{Hyp}$:
\begin{proof}[Formal proof of Proposition \ref{prop:farfromtaub-main}]
Given $\varepsilon_T>0$, find $\epsilon_N, \epsilon_d, \epsilon_s, \hat C_0>0 $ such that Lemma \ref{lemma:farfromtaub-uniform-hyp} and Lemma \ref{lemma:farfromtaub-cap} hold (for arbitrary $\varepsilon_c$).

Set $\mu_i = \delta_i'/\delta_i + \epsilon_s$; it suffices to show that $\int_{t_1}^{t_2}\mu(t) \dd t < \log C_0$ is bounded above, independently of $\gamma$, $i\in\{1,2,3\}$ and $t_1$,$t_2$. Fix $\gamma$ and $t_1<t_2$.

Decompose $[t_1,t_2]$ into intervals $S_k < T_k <S_{k+1}$ corresponding to the preimages of the regions \textsc{Cap} and \textsc{Hyp}, i.e. such that
\[
\bpt x([S_k,T_k])\subseteq {\textsc{Hyp}} \qquad\text{and}\qquad \bpt x([T_k, S_{k+1}])\subseteq {\textsc{Cap}}.
\]

We begin by considering the contribution from \textsc{Cap}, i.e.,~an interval $[T_k, S_{k+1}]$.
By Lemma \ref{lemma:farfromtaub-cap} we have $S_{k+1}-T_k< \hat C_0$; since $\mu$ is bounded, we get $\int_{T_k}^{S_{k+1}}\mu(t)\dd t < \log C_0$ for some $C_0$.

Next, we consider the contributions from \textsc{Hyp}. In this region, $\mu < -\epsilon_s$. Take an interval $[S_k, T_k]$, which is not the initial or final interval, i.e.~$t_1 < S_k < T_k < t_2$.
Assume without loss of generality that $|N_1(S_k)|=\epsilon_N$ and $|N_2(T_k)|=\epsilon_N$. 
Then $|N_2(S_k)|= 0.25 \delta_3^2/|N_1|(S_k)< \epsilon_d^2/\epsilon_N$. Since $|N_2'/N_2|<3$, we obtain $T_k - S_k >  - \frac{2}{3}\log \frac{\epsilon_d}{\epsilon_N}$.
Adjust $\epsilon_d>0$ to be so small, that $C_s(T_k-S_k) > 2\log C_0$. Then such an interval gives us a contribution of $\int_{S_k}^{T_k}\mu(t)\dd t < -\log C_0$. 

 For the complete interval $(t_1,t_2)$, sum over $k$; two disjoint intervals in the \textsc{Cap}-region must always enclose an interval in the \textsc{Hyp}-section, which cancels the contribution of its preceding \textsc{Cap}-region. Therefore, at most the last \textsc{Cap}-region stays unmatched and we obtain
 \(
 \int_{t_1}^{t_2}\mu(t)\dd t < \log C_0.
 \)
\end{proof}
Next, we adjust the constants from Lemma \ref{lemma:farfromtaub-cap} in order to show that trajectories near $\mathcal A$ can only enter the vicinity of Taub-points via the heteroclinic $-\taubp_\ell\to\taubp_\ell$:
\begin{proof}[Formal proof of Proposition \ref{prop:farfromtaub-main2}]
We find some $C_u>0$ such that $d(K(p), T_i)> d(p,T_i)+C_u$ for every $p\in\mathcal K$ with $d(p,\taubp_i)\in (\varepsilon_T, 0.5]$. It is evident from Figure \ref{fig:short-het} (or, formally, Proposition \ref{prop:kasnermap-homeomorphism-class}) that this is possible.

By Lemma \ref{lemma:farfromtaub-uniform-hyp}, we can make $\epsilon_N$ small enough that $\diam \gamma < C_u/8$ for pieces of trajectories $\gamma:(t_1,t_2)\to{\textsc{Hyp}}$. 
Using the continuity of the flow, i.e. Lemma \ref{lemma:farfromtaub-cap}, we can make $\epsilon_d$ small enough such that pieces $\bpt x:(t_1,t_2)\to {\textsc{Cap}}$ are approximated by heteroclinic orbits up to distance $C_u/4$. 

Now suppose $T^*<\infty$ and $\min_i d(\bpt x_0, \mathcal T_i)>5\varepsilon_T$.
We cannot have $\max_i\delta_i(\phi(x_0,T^*))\ge \epsilon_d$; hence, $d(\bpt x(T^*), \taubp_\ell)=\varepsilon_T$ for some $\ell\in\{1,2,3\}$. Set
\[T^3 = \sup\left\{t<T^*: \bpt x(t)\not\in \overline{B_{1.5 \varepsilon_T}(\taubp_\ell)}\cap \textsc{Circle}[\epsilon_n, \epsilon_d]\right\}.\]
By the assumption \eqref{eq:nearA:init-far-from-taub}, we have $T^3>0$.
We cannot have $d(\taubp_\ell, \bpt x(T^3))=1.5\varepsilon_\taubp$, since we already know $\diam \bpt x([T^3,T^*])\le C_u/8 < 0.5 \varepsilon_T$ (since, by construction, $\bpt x([T^3, T^*])\subseteq \textsc{Hyp}(\varepsilon_\taubp, \epsilon_N, \epsilon_d)$). This proves \eqref{eq:farfromtaub-thm-3star}, as well as 
\[\bpt x(T^3)\in\partial \textsc{Cap}[\epsilon_N, \epsilon_d]\cap \partial \textsc{Circle}[\epsilon_N,\epsilon_d]\cap B_{1.5\varepsilon_T}(\taubp_\ell).\]
Next, we set 
\[T^2 = \sup\left\{t\in [0,T^3): \bpt x([t,T_3))\subseteq\textsc{Cap}[\epsilon_n, \epsilon_d]\right\}.\]
In the interval $t\in [T_2,T_3]$, the trajectory is in one of the three \textsc{Cap} regions; this must be the $|N_\ell|\ge \epsilon_N$ cap, since otherwise $d(\bpt x(t), \taubp_\ell)$ would be decreasing (see Figure \ref{fig:short-het}). 
We set 
\[T^1 = \sup\left\{t\in [0,T^2): \bpt x(t)\not\in \textsc{Circle}[\epsilon_n, \epsilon_d]\right\}.\]
Similar arguments yield the remaining claim \eqref{eq:farfromtaub-thm-12}.
\end{proof}

\begin{remark}
The constants generated in this section are sub-optimal (at least doubly exponentially so). If one cared at all about their numerical values, then one would need to replace Lemma \ref{lemma:farfromtaub-cap} and Proposition \ref{prop:kasnermap-homeomorphism-class} by explicit estimates.
\end{remark}

\section{Analysis near the generalized Taub-spaces $\mathcal T_i^G$}\label{sect:near-taub}
In this section, we will study the dynamics in the vicinity of the generalized Taub-spaces, without loss of generality $\mathcal T_1^G$, using the polar coordinates from Section \ref{sect:polar-coords}. This section is structured in the following way:

In section \ref{sect:neartaub:motivation}, we will give a highly informal motivation for the general form of our estimates. This part can be safely skipped by readers who are uncomfortable with its hand-wavy nature. In section \ref{sect:neartaub:heteroclinic}, we will study the behaviour of trajectories near the heteroclinic orbit $-\taubp_1\to\taubp_1$, which come from either the $|N_2|\gg 0$ or the $|N_3|\gg 0$ cap. In section \ref{sect:neartaub:neartaub}, we will study the further behaviour near $\taubp_1$ of such trajectories. In section \ref{sect:neartaub:forbidden-cones}, we will study the behaviour of trajectories near $\taubp_1$ which do \emph{not} necessarily have the prehistory described in section \ref{sect:neartaub:heteroclinic}, and especially provide the deferred proof of Lemma \ref{lemma:farfromtaub:no-delta-increase-near-taub}. This section is mostly optional for our main results: Any trajectory which ever leaves the region where Section \ref{sect:neartaub:forbidden-cones} is necessary will never revisit this region, a fact which is proven without referring to any results from Section \ref{sect:neartaub:forbidden-cones}.

\subsection{Informal Motivation}\label{sect:neartaub:motivation}
We already alluded to the motivation for the estimates in this section in the introductory Section \ref{sect:intro} , page \pageref{paragraph:intro:strategy:quotients}: 
From Proposition \ref{prop:farfromtaub-main}, by varying $\varepsilon_\taubp$, we can control the behaviour of trajectories near $\taubp_1$ if $\delta_1 \ll r_1$ and obtain estimates of the following type for partial trajectories $\gamma:(t_1,t_2)\to B_{\epsilon}(\taubp_1)$ and continuous monotonous functions $\rho:(0,1]\to (0,1]$:
\begin{quote}
Suppose $\delta_1(t_1)<\rho_0(r_1(t_1))$;\\ then $\delta_1(t_2) < C_0 e^{-\rho_1(r_1(t_1))}\delta_1(t_1)$ and  $r_1(t_2)\ge \rho_2(r_1(t_1))$.
\end{quote}
The bounds will take the specific form $\rho_0(r)=Cr$ and $\rho_2(r) = C r$ and $\rho_1(r)= \frac{C}{r^2}$ (Proposition \ref{prop:neartaub:main}).

We know some prehistory of trajectories entering the vicinity of $\taubp_1$ (by Proposition \ref{prop:farfromtaub-main2}), which allows us to track backwards the condition $\delta_1<\rho_0(r_1)$. Further tracking back this condition, 
it is clear that trajectories entering the vicinity of $-\taubp_1$ must have $\delta_1 \le \epsilon_d \ll r_1\sim \epsilon_N$. 

At least in the Bianchi \textsc{IX}-like case, where $\mathrm{sign}\,N_2=\mathrm{sign}\,N_3$, the set $\{\bpt x:\,r_1=0\}$ is invariant, allowing us to get some $\rho_3:(0,1]\to (0,1]$ such that $\delta_1 < \rho_3(r_1)$ for \emph{any} trajectory entering the vicinity of $\taubp_1$ via the route in Proposition \ref{prop:farfromtaub-main2}, i.e.~via $-\taubp_1$ and some cap before.

These estimates combine well if we can make $\rho_3 < \rho_0$.
The estimates will take the specific form $\rho_3(r)= \epsilon r$, for arbitrarily small $\epsilon>0$ (Proposition \ref{prop:near-taub:qc}), which is precisely the required estimate at $\taubp_1$. 
In Bianchi \textsc{VIII}, we have no qualitative a-priori reason to expect bounds of the same form. Nevertheless, we will prove that they hold, which allows us to control any solution entering $\taubp_1$ as in Proposition \ref{prop:farfromtaub-main2}.

For the sake of brevity of arguments, we will present our analysis in the reverse order: We chronologically follow a trajectory from $-\taubp_1$ to $+\taubp_1$ and then until it leaves the vicinity of $+\taubp_1$, instead of tracking estimates backwards.

\subsection{Analysis near $-\taubp_1$ and near the heteroclinic $-\taubp_1\to\taubp_1$}\label{sect:neartaub:heteroclinic}
The behaviour of trajectories away from $\taubp_1$ is already partially described by Proposition \ref{prop:farfromtaub-main}; we only need to additionally estimate the quotient $\frac{\delta_1}{r_1}$ in this region. The necessary estimates can be summarized in the following:
\begin{proposition}\label{prop:near-taub:qc}
Let $\varepsilon_{T}\in (0,0.1)$. We can chose $\epsilon_N, \epsilon_d>0$ small enough, such that Propositions \ref{prop:farfromtaub-main} and \ref{prop:farfromtaub-main2} hold, as well as choose constants $C_1,\ldots,C_5>0$ large enough, such that the following holds:

Let $0< T_1 \le T_2$ and $\bpt x:[0,T_2)\to M_{\pm,\pm,\pm}$ be piece of trajectory, such that
\begin{subequations}
\begin{align}
\bpt x(0)&\in \partial \textsc{Circle}[\epsilon_N, \epsilon_d],\label{eq:near-taub-main-ass-initial} \\
\nonumber &\quad\text{i.e.}\quad \max_i\delta_i(0)<\epsilon_d,\ \max(|N_2|,|N_3|)(0)=\epsilon_N  \\
\bpt x([0,T_1))&\subseteq B_{2\varepsilon_{T}}(-\taubp_1)\cap \textsc{Circle}[\epsilon_N, \epsilon_d],\label{eq:near-taub-main-ass-q}\\
\nonumber &\quad\text{i.e.}\quad \max_i\delta_i(t)<\epsilon_d,\ \max_i|N_i|(t)\le\epsilon_N,\ d(\gamma(t), \taubp_1)\le 2\varepsilon_T\qquad\forall\,t\in [0,T_1]\\
\bpt x([T_1,T_2))&\subseteq\textsc{Cap}[\epsilon_N, \epsilon_d]\label{eq:near-taub-main-ass-c},\\
\nonumber &\quad\text{i.e.}\quad \max_i\delta_i(t)<\epsilon_d,\ |N_1|(t)\ge \epsilon_N \quad \forall\,t\in [T_1,T_2),
\end{align}\end{subequations}
i.e.~we are in the situation of the conclusion of Proposition \ref{prop:farfromtaub-main2}. 

Then the following estimates hold:
\begin{subequations}\begin{align}
\frac{\delta_1}{r_1}(0) &\le  C_1\epsilon_d&&\label{eq:near-taub-main-conc-initial}\\
\frac{\delta_1}{r_1}(t_2) &\le C_2 \frac{\delta_1}{r_1}(t_1)\qquad&\forall& 0 \le t_1 \le t_2 < T_2\label{eq:near-taub-main-conc-delta-r-qc}\\
\delta_1(t_2) &\le C_3 e^{- C_4^{-1}(t_2-t_1)}\delta_1(t_1)\qquad&\forall&0 \le t_1 \le t_2 < T_2\label{eq:near-taub-main-conc-delta-qc}\\
r_1(t_2) &\ge C_5^{-1} r_1(t_1)\qquad&\forall& T_1 \le t_1 \le t_2 < T_2\label{eq:near-taub-main-conc-r-c}.
\end{align}\end{subequations}
Alternatively to the assumption \eqref{eq:near-taub-main-ass-initial}, we can assume \eqref{eq:near-taub-main-conc-initial}; then the case $0=T_1$ is valid as well.
\end{proposition}
\begin{proof}
The claim \eqref{eq:near-taub-main-conc-delta-qc} is already proven in Proposition \ref{prop:farfromtaub-main}.
Assume without loss of generality that $|N_3|(0)=\epsilon_N$. Assuming that $\epsilon_d$ is small enough compared to $\epsilon_N$, we have $|N_2|(0)\le \delta_1^2(0)/\epsilon_N \le 0.5 \epsilon_N$ and hence $r_1(0) \ge |N_3|-|N_2|\ge 0.5\epsilon_N$. Therefore $\frac{\delta_1}{r_1}(0)\le C\delta_1(0)\le C\epsilon_d$ and claim \eqref{eq:near-taub-main-conc-initial} holds.
%
First, consider the Bianchi \textsc{IX}-like case $\signN_2=\signN_3$ in polar coordinates, i.e.~equations \eqref{eq:neartaub-b9-q}. We can immediately estimate $\partial_t\log\frac{\delta_1}{r_1}\le C|N_1|$. We already established in Section \ref{sect:far-from-taub} that $\int |N_1|\dd t$ is bounded for $t_1,t_2\in [0,T_1]$ (Lemma \ref{lemma:farfromtaub-uniform-hyp}) and that $T_2-T_1$ is bounded (Lemma \ref{lemma:farfromtaub-cap}), yielding \eqref{eq:near-taub-main-conc-delta-r-qc}.

Next, consider the Bianchi \textsc{VIII}-like case $\signN_2=+1$ and $\signN_3=-1$, i.e.~equations \eqref{eq:neartaub-b8-q:delta-r}.
Assume that we have for all $t\in [0,T_2]$,
\begin{equation}\label{eq:delta-r-one}\frac{\delta_1}{r_1}(t) < 1.\end{equation}
Under the assumption \eqref{eq:delta-r-one}, we can estimate \[\frac{N_+}{r_1}=\frac{\sqrt{N_-^2+\delta_1^2}}{r_1}=\sqrt{\frac{N_-^2}{N_-^2+\Sigma_-^2}+\frac{\delta_1^2}{r_1^2}}\le \sqrt{2},\] and hence $\partial_t\log\frac{\delta_1}{r_1}\le C|N_1|$, yielding \eqref{eq:near-taub-main-conc-delta-r-qc}.
When we adjust $\varepsilon_d$ such that $C_1C_2\varepsilon_d<1$, then this argument bootstraps to prove \eqref{eq:near-taub-main-conc-delta-r-qc}, without assuming a priori \eqref{eq:delta-r-one} (proof: Assume there was a time $T\in (0,T_2)$ such that \eqref{eq:delta-r-one} was violated; then $\frac{\delta_1}{r_1}(T)\le C_{2}\frac{\delta_1}{r_1}(0)\le C_2C_1\varepsilon_d<1$).

Next, we prove the remaining claim \eqref{eq:near-taub-main-conc-r-c}.
Since we know that $\frac{\delta_1}{r_1}(t) < 1$, we can estimate $\DD_t \log r < C$ for for all $t\in [T_1,T_2)$, both in Bianchi \textsc{VIII} and \textsc{IX}. Since $T_2-T_1$ is bounded, we obtain \eqref{eq:near-taub-main-conc-r-c}.

Considering the above proof, it is obvious that we can alternatively replace the assumption \eqref{eq:near-taub-main-ass-initial}  by \eqref{eq:near-taub-main-conc-initial}, and then also allow $0=T_1$.
\end{proof}

\begin{remark}
We excluded the set $\forbidden[\epsilon_v]$ from our analysis, given by
\[
\forbidden[\epsilon_v] = \{\bpt x \in \RR^5: \delta_1 \ge \epsilon_v r_1 \}, 
\]
i.e.~we described the dynamics outside of $\forbidden$ and showed that the set $\forbidden$ cannot reached by initial conditions described by Proposition \ref{prop:farfromtaub-main2}. 

Ignoring the constraint $G=1$, the set $\forbidden$ looks like a linear cone times $\RR^2$, since both $r_1$ and $\delta_1$ are homogeneous of first order in $N_2,N_3,\Sigma_-$ and independent of $N_1$ and $\Sigma_+$. 

Even though Bianchi \textsc{VIII} lacks an explicit invariant Taub-space, the $\forbidden$-cone around the generalized Taub-space $\mathcal T_1^G$ is a suitable ``morally backwards invariant'' replacement.
\end{remark}

\subsection{Analysis near $\taubp_1$}\label{sect:neartaub:neartaub}
Our analysis of the neighborhood of $\taubp_1$ can be summarized in the following
\begin{proposition}\label{prop:neartaub:main}
For any $\varepsilon_T\in (0,0.1)$, there exist constants $\epsilon_v>0$, $C_1,C_r,C_{\delta,r}>0$ and $C_e=10$ such that the following holds:

Let $\gamma:[0,T^*)\to B_{2\varepsilon_T}(\taubp_1)$ be a partial trajectory with
\begin{equation}
\frac{\delta_1}{r_1}(0)< \epsilon_v.
\end{equation}
Then, for all $0\le t_1\le t_2 <T^*$:
\begin{subequations}
\begin{align}
(|N_1|, \delta_2,\delta_3) (t_2) &\le C_1 e^{-C_e^{-1} (t_2-t_1)}\,(|N_1|, \delta_2,\delta_3)(t_1) \label{eq:neartaub:main:others}\\
\frac{\delta_1}{r_1}(t_2)&\le C_{\delta,r}\frac{\delta_1}{r_1}(t_1) \label{eq:neartaub:main:q}\\
r_1(t_2)&\ge C_r^{-1}r_1(t_1)\label{eq:neartaub:main:r}\\
\delta_1(t_2)&\le C_0\exp\left(-\frac{C_e^{-1}}{r_1^{2}(t_1)}(t_2-t_1)\right)\,\delta_1(t_1)\label{eq:neartaub:main:delta}\\
T^*&<\infty\label{eq:neartaub:main:T}.
\end{align}
\end{subequations}
\end{proposition}
We begin by proving the first three of the claims, in a way analogous to the proof of Proposition \ref{prop:near-taub:qc}:
\begin{proof}[Proof of Proposition \ref{prop:neartaub:main}, conclusions \eqref{eq:neartaub:main:others}, \eqref{eq:neartaub:main:q}, \eqref{eq:neartaub:main:r}]
The exponential decay \eqref{eq:neartaub:main:others} follows trivially from \eqref{eq:ode2-ni}, see e.g.~Figures \ref{fig:n-discs2} and \ref{fig:n-discs}. 

In order to see \eqref{eq:neartaub:main:q}, we again need to bootstrap: First consider $T_1=\sup\{t\in [0,T^*): \frac{\delta_1}{r_1}(t') < 1\,\forall t'\in [0,t)\}$.
We can estimate $\partial_t\log \frac{\delta_1}{r_1} < C|N_1|$ for all $t\in [0,T_1)$, using \eqref{eq:neartaub-b9-t:delta-r} and \eqref{eq:neartaub-b8-t:delta-r} and the fact that the higher order terms $|h_r|,|h_\delta|, |h_\psi|$ are bounded (by $C=5$). By the exponential decay of $N_1$, this yields \eqref{eq:neartaub:main:q} upon integration, for all $t_2\le T_1$. By adjusting $\epsilon_v<C_{\delta,r}^{-1}$ we can then conclude $T_1=T^*$.

Using $\frac{\delta_1}{r_1}<1$, we can estimate $\partial_t r_1 > -C|N_1|$, which upon integration yields the claim \eqref{eq:neartaub:main:r}.
\end{proof}
The next estimate \eqref{eq:neartaub:main:delta} requires a slightly more involved averaging-style argument, similar to the proof of Proposition \ref{prop:farfromtaub-main}:
\begin{proof}[Proof of Theorem \ref{prop:neartaub:main}, conclusion \eqref{eq:neartaub:main:delta}]
We set
\[\mu = \DD_t\log\delta_1 + 0.1 r_1^2.\]
It suffices to prove that $\int_{t_1}^{t_2}\mu(t)\dd t\le \log C_0$ for some $C_0>0$.

\paragraph{Strategy.}
We will first consider times where $|N_1|\not\ll r_1$; the integral $\int \mu \dd t$ over these times will be bounded by $\int |N_1|\dd t$. Next we will split $\mu$ into a nonpositive and a nonnegative part; the nonnegative (bad) part will have a contribution for every $\psi$ rotation, which is bounded by $C r_1$, while the nonpositive (good) part will have a negative contribution for every $\psi$-rotation which scales with $r_1\log\frac{\delta_1}{r_1}$. Adjusting $\epsilon_v$ will then yield the desired estimate (after summing over $\psi$-rotations).

\paragraph{Estimates for large $|N_1|$.}
Choose $\widetilde T$ (possibly $\widetilde T= 0$) such that $|N_1(t)|\ge C r_1^2(t)$ (with $C=0.05$) for $t\in (0, \widetilde T]$ and $|N_1(t)|\le 0.1 r_1^2(t)$ for $t \in [\widetilde T, T]$. This is possible, since $\DD_t \log |N_1| < 2\DD_t \log r_1$. 
Then $|\mu(t)| \le C \sqrt{|N_1(t)|}$ for all $t\in [0,\widetilde T]$ and hence $\int_{0}^{\widetilde T}|\mu(t)|\dd t< C$.

\paragraph{Averaging Estimates.}
Consider without loss of generality $t_1,t_2\ge \widetilde T$. Using $\Sigma_+\approx -1$ and $\delta_1 \le 0.1 r_1$, we can estimate
\[\mu \le -0.4 r_1^2 \cos^2\psi + 0.6 r_1^2\sin^2\psi + 0.01 r_1^2 + 0.1 r_1^2 + C|N_1| \le -0.25 r_1^2 + r_1^2\sin^2\psi.\]
We can also estimate $\psi'$:
\[r_1 |\sin\psi| \le \psi' \le 2r_1\sqrt{\sin^2\psi + \frac{\delta_1^2}{r_1^2}}.\]
Let $\mu_+ = r_1^2\sin^2\psi$ be the positive (bad) part of $\mu$; take times $t_1<t_L<t_R<t_2$ with $|\psi(t_R)-\psi(t_L)|\le 2 \pi$. Then
\[\int_{t_L}^{t_R}\mu_+(t)\dd t 
\le \int_{\psi(t_L)}^{\psi(t_R)}\frac{\mu_+(t)}{\psi'(t)}\dd t
\le C_{r}^{-1}r_1(t_R)\int_{0}^{2\pi}|\sin\psi|\dd \psi \le C_+ r_1(t_R),\]
for some $C_+>0$.
On the other hand, let $\mu_-=-0.25 r_1^2$ be the negative (good) part of $\mu$. Take times $t_1<t_L<t_R<t_2$ with $\psi(t_L)=k\pi-0.1$ and $\psi(t_R)=k\pi +0.1$ for some $k\in\ZZ$; then we can estimate
\[\begin{aligned}
\int_{t_L}^{t_R}\mu_-(t)\dd t &\le \int_{-0.1}^{+0.1}\frac{-0.25\, r_1^2}{\psi'}\dd \psi\le \int_{-0.1}^{+0.1}\frac{-0.125\, r_1}{\sqrt{\sin^2\psi + \frac{\delta_1^2}{r_1^2}}}\dd \psi\\
&\le -0.2\, C_{r}^{-1}r_1(t_L)\int_{0}^{0.1}\frac{1}{\sqrt{x^2 + C_{\delta,r}^2 \epsilon_v^2}}\dd x
\end{aligned}\]
We can immediately see that for any $C_+>0$, we find $\epsilon_v>0$ such that $\frac{\delta_1}{r_1}(t_1)<\epsilon_v$ implies 
$\int_{t_L}^{t_R}\mu_-(t)\dd t < -C_+ r_1(t_L)$. Hence, by summing over $\psi$-rotations (and adjusting $\epsilon_v$), we can conclude the assertion \eqref{eq:neartaub:main:delta}.
\end{proof}
\begin{proof}[Proof of Proposition \ref{prop:neartaub:main}, conclusion \eqref{eq:neartaub:main:T}]
We need to show that solutions with small quotient $\frac{\delta_1} {r_1}$ cannot stay near $\taubp_1$ forever.

Assuming without loss of generality $|N_1|\ll r_1^2$ we have $\DD_\psi \log r_1 > C r_1 |\sin\psi|$. This shows that the only way never leaving the vicinity of $\taubp_1$ is for the angle $\psi\in \RR$ to stay bounded, i.e.~$\lim_{t\to\infty}\psi(t)=\psi^{**}$ and $\lim_{t\to\infty}\delta_1(t)=0$ (since otherwise $r_1$ increases by a too large amount during each rotation). This is impossible, since the possible limit-points lie on the Kasner-circle $\mathcal K\setminus \{\taubp_1\}$ and are not $\taubp_1$; hence, either $N_2$ or $N_3$ is unstable and since initially $N_2\neq 0\neq N_3$, the trajectory cannot converge to such a point.
\end{proof}
\subsection{Analysis in the \forbidden-cones}\label{sect:neartaub:forbidden-cones}\label{sect:averaging-unneeded}
Our whole approach aims at avoiding the much more tricky analysis of the dynamics in the \forbidden-cones, where possibly $\delta_1\ge r_1$: Since trajectories starting outside of these cones never enter them, it is unnecessary to know what happens in the \forbidden-cones. However, for various global questions, it is useful to collect at least some results inside of these cones. 

We already know that solutions in Bianchi \textsc{IX} cannot converge to the Taub-points; the same holds in Bianchi \textsc{VIII}, even for solutions in \forbidden:
\begin{lemma}\label{lemma:no-b8-taub-convergence}
For an initial condition $\bpt x_0\in \mathcal M_{+-+}$, it is impossible to have $\lim_{t\to\infty}\bpt x(t)=\taubp_1$.
\end{lemma}
\begin{proof}
Suppose we have such a solution.
Using equation \eqref{eq:neartaub-b8-q}, we can write near $\taubp_1$:
\[
\DD_t \log |N_1|\frac{\delta_1}{r_1} \le \sqrt{3}\frac{|\Sigma_-|}{r_1} |N_1|\frac{\sqrt{N_-^2 + \delta_1^2} }{r_1} - 2.5.
\]
Hence, if ever $|N_1|\frac{\delta_1}{r_1}<1$, this inequality is preserved and $|N_1|\frac{\delta_1}{r_1}$ decays exponentially. Then we can estimate $\DD_t \log r_1 \ge -C |N_1| -C |N_1|\frac{\delta_1}{r_1}$; all the terms on the right hand side have bounded integral and $r_1\to 0$ is impossible.

On the other hand, if $r_1 < |N_1|\delta_1$ for all sufficiently large times, we can estimate $\DD_t\log \delta_1\ge -C r_1^2 - C |N_1| \ge -C |N_1|$, which has bounded integral and thus contradicts $\delta_1\to 0$.
\end{proof}

Unfortunately, this is all we can presently say in the \forbidden{} cone in Bianchi \textsc{VIII}. 

\noindent 
In the case of Bianchi \textsc{IX}, we can still average over $\psi$-rotations in order to show that $\frac{\delta_1}{r_1}$ decays, even in the \forbidden{} region:
\begin{lemma}\label{neartaub:lemma:joint-increase-in-forbidden}
Let $h>0$. There exists constants $\epsilon, C>0$ such that the following holds:

Let $\gamma:[0,T]\to \{\bpt x\in\mathcal M_{*++}:\, |N_1|\le r_1^5,\,r_1<\epsilon,\,r_1\le h\delta_1,\, d(\bpt x, \taubp_1)<0.1\}$ be a partial trajectory near $\taubp_1$. Then 
\[
\log \frac{\delta_1}{r_1}(0) -\log \frac{\delta_1}{r_1}(T) \le C (\log r_1(T)-\log r_1(0)),
\]
i.e.~the increase of $r_1$ and the decrease of $\frac{\delta_1}{r_1}$ have comparable rates.
\end{lemma}
\begin{proof}
We can estimate
\[\begin{aligned}
\DD_t \log r_1 &= r_1^2\sin^2\psi \frac{-\Sigma_+}{1-\Sigma_+} + \mathcal O(|N_1|)\\
\DD_t \log \frac{\delta_1}{r_1}&=\frac{-1}{1-\Sigma_+}r_1^2\cos^2\psi +\mathcal O(|N_1|)\\
\psi'&= \sqrt{3}r_1 \sqrt{\sin^2\psi + \frac{\delta_1^2}{r_1^2}} - \frac{r_1^2}{1-\Sigma_+} \cos\psi \sin\psi + \mathcal O(|N_1\sin\psi|).
\end{aligned}\]
It suffices to show an estimate of the form
\[
\int_0^{2\psi}r_1^2\frac{\sin^2\psi}{\psi'}\dd \psi >C  \int_0^{2\psi}\frac{r_1^2\cos^2\psi}{\psi'}\dd \psi.
\]
Using $\delta_1 \ge h r_1$ and $r_1\le \epsilon$, we can directly estimate $\delta_1 \le \psi'\le 2 \sqrt{r_1^2+\delta_1^2}$ (for $\epsilon>0$ small enough). This allows us to see that $r_1, \delta_1, \frac{\delta_1}{r_1}$ can all change only by a bounded factor during each rotation; hence it suffices to show for $\psi(t_2)-\psi(t_1)\le 2\pi$ that 
\(
\min_{t\in [t_1,t_2]}\psi'(t) > C \max_{t\in [t_1,t_2]}\psi'(t).
\)

\noindent 
However,
\[
\min_{t\in [t_1,t_2]}\psi'(t) \ge \min_t \delta_1(t)\qquad
\max_{t\in [t_1,t_2]}\psi'(t) \ge 2 \max_{t} \sqrt{r_1^2(t) + \delta_1^2(t)}\le \max_t \delta_1(t) \sqrt{1+ h^{-2}},
\]
and we know that $\delta_1$ can only change by a bounded factor during each rotation; hence
the desired estimate follows.
\end{proof}

With a more subtle averaging argument than the previous ones, we can also show the deferred Lemma \ref{lemma:farfromtaub:no-delta-increase-near-taub}. However, this proof is only given for the sake of completeness and is nowhere used in this work, except for completing the literature review in Section \ref{sect:farfromA:attract}.
\begin{lemma}\label{lemma:farfromtaub:no-delta-increase-near-taub}
We consider without loss of generality the neighborhood of $\mathcal T_1$. Let $\epsilon>0$ small enough. Then there exists a constant $C_{\delta,\epsilon}\in(1,\infty)$, such that, for any piece of trajectory $\gamma:[t_1,t_2]\to \{\bpt x\in\mathcal M_{*++}: |\bpt \Sigma(\bpt x)-\taubp_1|\le \epsilon,\,|N_1|\le 10,\,\delta_1\le 10\}$, the following estimate holds:
\[ \delta_1(\gamma(t_2))\le C_{\delta,\epsilon}\delta_1(\gamma(t_1)). \]
\end{lemma}
\begin{proof}
We can assume without loss of generality that $\delta_1 > h r_1$ for all $t\in [t_1,t_2]$ and for some $h=\hat\epsilon>0$; otherwise Proposition \ref{prop:neartaub:main} applies. We can also assume without loss of generality that $|N_1| \le C r_1^2$ for all $t\in [t_1,t_2]$ because of a similar argument as in the proof of Proposition \ref{prop:neartaub:main}.

We will use the letters 
\def\Cbig{C}
\def\Csmall{c}
$\Cbig \gg 1$ and $0<\Csmall\ll 1 $ for unspecified constants. Recall \eqref{eq:neartaub-b9-t}.
We use the auxiliary variable $\zeta=\frac{\delta_1}{r_1}\ge \hat\epsilon$. First, we note that 
\[\begin{array}{rccl}
 \Csmall r_1 \zeta &\le& \psi' &\le \Cbig r_1 \zeta\\
\Csmall r_1^2 \zeta^{-1}\cos^2\psi - \Cbig |N_1|\frac{1}{\psi'}&\le& \DD_\psi r &\le \Cbig r_1^2 \zeta^{-1}\\
\Csmall r_1\cos^2\psi - \Cbig|N_1|\frac{1}{\psi'} &\le&-\DD_\psi \zeta &\le \Csmall r_1.
\end{array}\]
We see that $r_1$ and $\zeta$ can change only by a bounded factor (and also a bounded amount) during a single rotation $[\psi_1, \psi_1+2\pi]$; therefore, we can average the equation for $r_1'$ to see that
\[\numberthis \label{eq:average:rprime-intbound} \int_{\psi(t_1)}^{\psi(t_2)}r_1^2\zeta^{-1}\dd \psi<\Cbig|r_1(t_1)-r_1(t_2)|<\Cbig.\] 
Also, $r_1$ is non-decreasing (except for the terms bounded by $|N_1|$); therefore, its total variation is bounded $\int_{\psi(t_1)}^{\psi(t_2)}r_1'\dd t < \Cbig$.

\noindent 
Our goal is to bound $\int_{t_1}^{t_2}\frac{\delta_1'}{\delta_1} \dd t < \Cbig$. We first split the terms into
\[\begin{aligned}
\DD_t \log \delta_1 &= r_1^2 (\sin^2\psi-\cos^2\psi) \frac{1}{1-\Sigma_+} 
\quad-\quad r_1^2 \cos^2\psi\frac{1+\Sigma_+}{1-\Sigma_+} \quad+\quad  N_1 h_\delta.
\end{aligned}\]
The last term is bounded by $|N_1|$ and hence has bounded integral. 
Since $1+\Sigma_+ \ge -\Cbig |N_1N_2N_3|^{\frac 2 3} \ge -\Cbig |N_1|^{\frac 2 3}$, the second term has its integral bounded above. Therefore, our goal now is to estimate the integral of the first term:
\[\begin{aligned}
&\int_{\psi(t_1)}^{\psi(t_2)} F(\psi, \zeta, r_1, N_1)(\sin^2\psi -\cos^2\psi)\dd \psi\overset{!}{\le}\Cbig,\quad\text{where}\\
F(\psi, \zeta, r_1, N_1)&=\frac{r_1^2}{1-\Sigma_+}\frac{1}{\psi'}\\
&=\frac{r_1}{1-\Sigma_+}\frac{1}{\sqrt{3}\sqrt{\sin^2\psi +\zeta^2} + \frac{r_1}{1-\Sigma_+} \cos\psi \sin\psi + N_1\sin\psi\, h_{\psi}},\quad \text{where}\\
h_\psi &= -\sqrt{3}\sin\psi + \cos\psi \frac{N_1-2r_1\sqrt{\sin^2\psi + \zeta^2}}{1-\Sigma_+}.
\end{aligned}\]
By the constraint $(1-\Sigma_+)(1+\Sigma_+)= r_1^2 + N_1(N_1-2r_1\sqrt{\sin^2\psi + \zeta^2})$, we know $\Sigma_+ = \Sigma_+(r_1, \zeta, \psi, N_1)\approx -1$. The quantity $|F|$ is bounded by $|F| \le \Cbig r_1\zeta^{-1}$. We set
\[
\widetilde \psi = \psi+\frac{\pi}{4},\qquad \left(\widetilde r_1,\widetilde \zeta,\widetilde N_1\right)(\psi)=(r_1,\zeta,N_1)\left(\psi+\frac{\pi}{4}\right).
\]
We will then estimate
\begin{align*}
&\int_{\psi(t_1)}^{\psi(t_2)-\frac{\pi}{4}}F(\psi, \zeta, r_1, N_1)(\sin^2\psi -\cos^2\psi)\dd \psi + 
\int_{\psi(t_1)+\frac{\pi}{4}}^{\psi(t_2)}F(\psi, \zeta, r_1, N_1)(\sin^2\psi -\cos^2\psi)\dd \psi\\
&\quad= \int_{\psi(t_1)}^{\psi(t_2)-\frac{\pi}{4}}\left[F(\psi, \zeta, r_1, N_1)-F(\widetilde\psi, \widetilde\zeta, \widetilde r_1, \widetilde N_1) \right](\cos^2\psi -\sin^2\psi)\dd \psi\\
&\quad=\int_{\psi(t_1)}^{\psi(t_2)-\frac{\pi}{4}}\left[F(\psi, \zeta, r_1, 0)-F(\widetilde\psi, \zeta,  r_1, 0) \right](\sin^2\psi -\cos^2\psi)\dd \psi\numberthis\label{eq:average-negative-term}\\
&\quad+ \mathcal O \left[ \int_{\psi(t_1)}^{\psi(t_2)-\frac{\pi}{4}} |\partial_{\zeta} F|\,|\zeta-\widetilde\zeta| + |\partial_{r_1} F|\,|r_1-\widetilde r_1| +  |\partial_{N_1} F|\,|N_1+\widetilde N_1|\,\dd \psi\right]\quad \overset{!}{\le}\Cbig.
\end{align*}
It is clear that this estimate will suffice for our claim. The crucial estimate is that the term \eqref{eq:average-negative-term} is nonpositive; this can be seen by considering both cases $\sin^2\psi\ge \cos^2\psi$ and $\sin^2\psi \le \cos^2\psi$.

The remaining estimates including derivatives of $F$ are lengthy but straightforward. 
The easiest derivative to estimate is by $N_1$; we can see that $|\partial_{N_1}\Sigma_+| \le \Cbig$ (since $\delta_1=\zeta r_1 \le \Cbig$) and then $|\partial_{N_1} h_\psi |\le \Cbig$; then $|\partial_{N_1}F| \le \Cbig$. This implies that $\int |\partial_{N_1}F||N_1+\widetilde N_1|\dd \psi \le \Cbig$. 

The next derivative to estimate is $\DD_{r_1}F$; we can see that $|\partial_{r_1}\Sigma_+| \le \Cbig \zeta$ and $|\partial_{r_1}h_\psi| \le \Cbig \zeta$; then $|\partial_{r_1}F|\le \Cbig+\Cbig r_1\zeta\le \Cbig$. We already know that $r_1$ can only increase by small amounts; hence
\[
\int_{\psi(t_1)}^{\psi(t_2)-\frac{\pi}{4}}|r_1-\widetilde r_1|\dd \psi \le -\Cbig + \int_{\psi(t_1)}^{\psi(t_2)-\frac{\pi}{4}}(\widetilde r_1-r_1)\dd \psi \le \Cbig.
\]

The last derivative to estimate is by $\zeta$. We can estimate $|\partial_\zeta \Sigma_+|\le \Cbig |N_1|r_1 \zeta^{-2}$ and $|\partial_\zeta h_\psi|\le \Cbig r_1\zeta^{-2}$. This allows us to estimate $|\partial_\zeta F|\le \Cbig r_1 \zeta^{-2}$. Now $|\zeta-\widetilde\zeta|\le \Cbig r_1$ and therefore
\[
\int_{\psi(t_1)}^{\psi(t_2)-\frac{\pi}{4}}|\partial_\zeta F|\,|\zeta-\widetilde\zeta| \dd\psi \le \Cbig \int_{\psi(t_1)}^{\psi(t_2)-\frac{\pi}{4}} r_1^2\zeta^{-2}\dd\psi \le \Cbig \int_{\psi(t_1)}^{\psi(t_2)-\frac{\pi}{4}} r_1^2\zeta^{-1}\dd \psi \le \Cbig,
\]
where the last estimate was due to \eqref{eq:average:rprime-intbound}.
\end{proof}

\section{Attractor Theorems}\label{sect:global-attract}
The goal of this section is to prove that typical initial conditions converge to $\mathcal A$. We have already seen Theorem \ref{farfromA:thm:b9-attract}, which is however somewhat unsatisfactory: It tells nothing about the speed and the details of the convergence; it relies on Lemma \ref{lemma:farfromtaub:no-delta-increase-near-taub}, which has a rather lengthy proof (page \pageref{lemma:farfromtaub:no-delta-increase-near-taub}f) mainly discussing the case $\delta_1 \gg r_1$, \emph{which is not supposed to happen anyway}; lastly, the proof of Theorem \ref{farfromA:thm:b9-attract} has no chance of generalizing to the case of Bianchi \textsc{VIII}.

In this section, we will combine the analysis of the previous Sections \ref{sect:far-from-taub} and \ref{sect:near-taub} in order to prove a \emph{local} attractor result, holding both in Bianchi \textsc{VIII} and \textsc{IX}. Together with the results from Section \ref{sect:farfromA} and some minor calculation, this will yield a ``global attractor theorem'', i.e.~a classification of solutions failing to converge to $\mathcal A$, which happen to be rare; in the case of Bianchi \textsc{IX}, this recovers and extends Theorem \ref{farfromA:thm:b9-attract}, and in the case of Bianchi \textsc{VIII}, this answers a longstanding conjecture.
\paragraph{Statement of the attractor Theorems.}
The local attractor theorem is given by the following:
\begin{thm}[Local Attractor Theorem]\label{thm:local-attractor}
There exist constants $C_1,C_2,C_3>0$ and $\epsilon>0$ such that the following holds:

Let $\bpt x_0\in \mathcal M_{\pm\pm\pm}$ be an initial condition in either Bianchi \textsc{VIII} or \textsc{IX}, with 
\begin{equation}\label{eq:local-attractor-condition}
\frac{\delta_i}{r_i} <\epsilon\quad\text{and}\quad \delta_i <\epsilon\qquad\forall i\in\{1,2,3\}.
\end{equation}

Then, for all $i\in \{1,2,3\}$ and $t_2\ge t_1\ge 0$:
\begin{subequations}
\begin{align}
\delta_i(t_2)&\le C_1 {\delta}(t_1)\label{eq:local-attract:delta}\\
\frac{\delta_i}{r_i}(t_2)&\le C_2\frac{\delta_i}{r_i}(t_1) \label{eq:local-attract:delta-r}.
\end{align}
Furthermore, for all $i\in \{1,2,3\}$,
\begin{align}
\lim_{t\to \infty}\delta_i(t)&=0 \label{eq:local-attract:delta-lim}\\
\lim_{t\to\infty} \frac{\delta_i}{r_i}(t)&=0\label{eq:local-attract:delta-r-lim}\\
\int_0^{\infty}\delta_i^2(t)\dd t &< C_3 \frac{\delta_i^2}{r_i^2}(\bpt x_0),\label{eq:local-attract:integral}
\end{align}
and the $\omega$-limit set $\omega(\bpt x_0)$ must contain at least three points in $\mathcal K$.
\end{subequations}
\end{thm}
The name ``local attractor theorem'' is descriptive: We describe a subset of the basin of attraction, i.e.~an open neighborhood of $\mathcal A\setminus\mathcal T$ which is attracted to $\mathcal A$ and given by
\begin{equation}
\textsc{Basin}_\signN[\epsilon] = \{\bpt x\in \mathcal M_\signN:\, \delta_i\le\epsilon,\,\frac{\delta_i}{r_i}\le \epsilon\quad\forall i\in\{1,2,3\}\}.
\end{equation}
The integral estimate \eqref{eq:local-attract:integral} tells us that the convergence to $\mathcal A$ must be reasonably fast.

We can combine the local attractor Theorem \ref{thm:local-attractor} with the discussion in Section \ref{sect:farfromA} in order to prove a global attractor theorem. Since some trajectories fail to converge to $\mathcal A$, most notably trajectories in the Taub-spaces, a global attractor theorem must necessarily take the form of a classification of all exceptions. In this view, the global result for the case of Bianchi Type \textsc{IX} models is the following:
\begin{thm}[Bianchi \textsc{IX} global attractor Theorem]\label{thm:b9-attractor-global}
Consider $\mathcal M_{+++}$, i.e.~Bianchi \textsc{IX}. Then, for any initial condition $\bpt x_0\in \mathcal M_{+++}$, the long-time behaviour of $\bpt x(t)$ falls into exactly one of the following mutually exclusive classes ($i\in\{1,2,3\}$):
\begin{description}
\item[\textsc{Attract.}] For large enough times, Theorem \ref{thm:local-attractor} applies.
\item[$\textsc{Taub}_i$.] We have $\bpt x_0\in \mathcal T_i$, and hence $\bpt x(t)\in\mathcal T_i$ for all times.
\end{description}
The set of initial conditions for which \textsc{Attract} applies is ``generic'' in the following sense: It is open and dense in $\mathcal M_{+++}$ and its complement has Lebesgue-measure zero (evident from the fact that $\mathcal T_i$ are embedded lower dimensional submanifolds, of both dimension and codimension two).
\end{thm}
The analogous, novel result for the case of Bianchi \textsc{VIII} is the following:
\begin{thm}[Bianchi \textsc{VIII} global attractor Theorem]\label{thm:b8-attractor-global}
Consider $\mathcal M_{+-+}$, i.e.~Bianchi \textsc{VIII}. Then, for any initial condition $\bpt x_0\in \mathcal M_{+-+}$, the long-time behaviour of $\bpt x(t)$ falls into exactly one of the following mutually exclusive classes:
\begin{description}
\item[\textsc{Attract.}] For large enough times, Theorem \ref{thm:local-attractor} applies. 
\item[$\textsc{Taub}_2$.] We have $\bpt x_0\in \mathcal T_2$, and hence $\bpt x(t)\in\mathcal T_2$ for all times. 
\item[$\textsc{Except}_1$.] For large enough times, the trajectory follows the heteroclinic object
\[ -\taubp_1 \to W^u(-\taubp_1) \to \taubp_1 \to W^s(-\taubp_1)\to -\taubp_1,\]
where $W^s(-\taubp_1)$ and $W^u(-\taubp_1)$ are the two-dimensional stable and one-dimensional unstable manifolds of $-\taubp_1$. We have $\lim_{t\to\infty}\max(\delta_2,\delta_3)(t)=0$ and $\limsup_{t\to\infty}\delta_1(t)>0=\liminf_{t\to\infty}\delta_1(t)$.
\item[$\textsc{Except}_3$.] The analogue of $\textsc{Except}_1$ applies, with the indices $1$ and $3$ exchanged.
\end{description}
\end{thm}
This theorem should be read in conjunction with the following, the proof of which will be deferred until Section \ref{sect:volume-form}, page \pageref{proof:thm:b8-attractor-global-genericity}:
\begin{hthm}{\ref{thm:b8-attractor-global-genericity}}[Bianchi \textsc{VIII} global attractor Theorem genericity]
In Theorem \ref{thm:b8-attractor-global}, 
the set of initial conditions $\bpt x_0$ for which \textsc{Attract} applies is ``generic'' in the following sense: It is open and dense in $\mathcal M_{+-+}$ and its complement has Lebesgue-measure zero. 
\end{hthm}
\begin{question}\label{quest:except}
It is currently unknown, whether the case \textsc{Except} in Bianchi \textsc{VIII} is possible at all.

We expect that solutions in $\textsc{Except}_1$ actually converge to the heteroclinic cycle, where $\taubp_1 \to -\taubp_1$ is realized by the unique connection in $\mathcal T_1^G$, i.e.~$\lim_{t\to\infty}r_1(t)=0$, instead of following any other heteroclinic in $W^u(\taubp_1)\cap W^s(-\taubp_1)$.

For topological reasons, we expect that the set of initial conditions, where \textsc{Except} applies, is nonempty, and is of dimension and codimension two.
\end{question}
\paragraph{Proof of the attractor Theorems.}
The remainder of this section is devoted to proving these theorems. We begin with the local attractor result:
\begin{proof}[Proof of the local attractor Theorem \ref{thm:local-attractor}]
We already gave a rough outline of the proof at the end of Section \ref{sect:farfromA:attract}, page \pageref{paragraph:farfromA:sketch}. We now have all ingredients to complete this program.

We apply Propositions \ref{prop:farfromtaub-main2}, \ref{prop:near-taub:qc} and \ref{prop:neartaub:main}. If the constants have been arranged appropriately, then each of these propositions describes a piece of the trajectory, which ends in the domain of the (cyclically) next proposition, and at least one of the three propositions is applicable.
``Appropriate'' means at least $\varepsilon_T^{\ref{prop:neartaub:main}}> C_T^{\ref{prop:farfromtaub-main2}}\varepsilon_T^{\ref{prop:farfromtaub-main2}}$.

Each $\delta_i$ shrinks by an arbitrarily large factor during each time interval, where Proposition \ref{prop:farfromtaub-main2} applies, and can at most grow by a bounded factor in each remaining time-interval, which directly yields \eqref{eq:local-attract:delta}, \eqref{eq:local-attract:delta-lim} (``arbitrary'' if we adjust $\epsilon$ in \eqref{eq:local-attractor-condition}). 

In order to see the estimates \eqref{eq:local-attract:delta-r} and \eqref{eq:local-attract:delta-r-lim}, we note that each $r_i$ can change only by a factor bounded by $\varepsilon_T^{-1}$ in the region, where \ref{prop:farfromtaub-main2} applies; however, the $\delta_i$ must shrink by an arbitrarily large factor, and the quotient $\frac{\delta_i}{r_i}$ is directly controlled in the regions where Propositions \ref{prop:near-taub:qc} and \ref{prop:neartaub:main} apply.

The remaining estimate \eqref{eq:local-attract:integral} follows by integration: The contribution away from the Taub-spaces can be trivially bounded by Proposition \ref{prop:farfromtaub-main}, eq.~\eqref{eq:nearA:exponential}, and Proposition \ref{prop:neartaub:main}, eq.~\eqref{eq:neartaub:main:delta}, bounds the contribution near the Taub-spaces by $\frac{\delta^2}{r^2}$. The long-time behaviour of $\frac{\delta_i}{r_i}$ is controlled by Proposition \ref{prop:near-taub:qc}. 
\end{proof}
Next, we prove the global result for Bianchi \textsc{IX}. Theorem \ref{thm:b9-attractor-global} is trivially equivalent to Theorem \ref{farfromA:thm:b9-attract}, so we could cite \cite{ringstrom2001bianchi} instead; the novel results are contained in the local Theorem \ref{thm:local-attractor}. The following is a novel proof of the global part of the result, which avoids Lemma \ref{lemma:farfromtaub:no-delta-increase-near-taub}:
\begin{proof}[Proof of the global Bianchi \textsc{IX} attractor Theorem \ref{thm:b9-attractor-global}]
Suppose that there is an initial condition $\bpt x_0\in\mathcal M_{+++}$, such that neither \textsc{Attract} nor \textsc{Taub} applies. We know that $\omega(\bpt x_0)\cap \mathcal A \subseteq \mathcal T$ (because otherwise \textsc{Attract} would apply). By Lemma \ref{farfromA:lemma:limitpoint-on-A}, there must exist at least one $\omega$-limit point $\bpt y\in\omega(\bpt x_0)\cap \mathcal K$. The $\omega$-limit set $\omega(\bpt x_0)$ must be connected or unbounded (this holds for general dynamical systems).

Now consider $\omega(\bpt x_0)\cap \mathcal T$. We claim that $\omega(\bpt x_0)\cap \mathcal T\subseteq \mathcal T_i$ for some $i\in\{1,2,3\}$: If it was otherwise, there would need to exist a heteroclinic orbit $\gamma\subseteq \omega(\bpt x_0)$ which connects two Taub-spaces (since $|N_1N_2N_3|\to 0$ we can apply Lemma \ref{lemma:farfromA:b7}). Without loss of generality, such a heteroclinic orbit can only connect $\{\bpt x: \bpt \Sigma=\taubp_1\}$ to $-\taubp_2$ with $\gamma\subseteq \mathcal M_{0++}$. Along $\gamma$, $\delta_2$ and $\delta_3$ and hence also $\frac{\delta_2}{r_2}$ and $\frac{\delta_3}{r_3}$ vanish; near the end of this heteroclinic, $\delta_1$ and $\frac{\delta_1}{r_1}$ must become arbitrarily small. In other words, we have some $\bpt p\in\gamma$ near $-\taubp_2$ such that $\frac{\delta_1}{r_1}(\bpt p)<\widetilde \epsilon/3$.
Hence, for $\epsilon_2=\epsilon_2(\bpt p)>0$ small enough, $B_{\epsilon}(\bpt p)\cap \mathcal M_{+++}\subseteq\textsc{Basin}_{+++}[\widetilde \epsilon]$; this contradicts our assumptions $\bpt p\in\omega(\bpt x_0)$ and $\omega(\bpt x_0)\cap \textsc{Basin}_{+++}[\widetilde \epsilon]=\emptyset$.

We also know by Lemma \ref{lemma:farfromA:taub-instability} that $\omega(\bpt x_0)$ cannot be contained in a Taub-Line $\mathcal{TL}_i=\{\bpt x: \bpt \Sigma(\bpt x)=\taubp_i\}$. Therefore, the only remaining case is that the trajectory $\phi(\bpt x_0, \cdot)$ follows the heteroclinic cycle
\[
-\taubp_i \to W^u(-\taubp_i) \to \taubp_i \to W^s(-\taubp_i)\to -\taubp_i,
\]
where $W^u(-\taubp_i)$ is the one-dimensional unstable manifold of $-\taubp_i$ and $W^s(-\taubp_i)$ is its two-dimensional stable manifold.

We now could use Lemma \ref{lemma:farfromtaub:no-delta-increase-near-taub} to finish our proof, since it prevents $\taubp_i\to W^s(-\taubp_i)$.
However, we prefer to avoid relying on the subtle averaging arguments needed to prove Lemma \ref{lemma:farfromtaub:no-delta-increase-near-taub}, substituting them by the simpler Lemma \ref{neartaub:lemma:joint-increase-in-forbidden}.
Assume without loss of generality that our trajectory follows this dynamics with $i=1$ and hence has $\frac{\delta_1}{r_1}\ge \widetilde \epsilon$ for all sufficiently large times, while $\frac{\delta_2}{r_2}\to 0$ and $\frac{\delta_3}{r_3}\to 0$. We will follow $\frac{\delta_1}{r_1}$ and show that $\frac{\delta_1}{r_1}\to 0$, which contradicts our assumptions. 

We first estimate the variation of $\frac{\delta_1}{r_1}$ away from $\taubp_1$. Using
\[
\DD_t\log \frac{\delta_1}{r} \le C|N_1N_2N_3|^{\frac 2 3} + C |N_1|,
\]
we can conclude that $\int \max(0,\DD_t\log\frac{\delta_1}{r_1})\dd t$ is bounded along every loop. However, by Lemma \ref{neartaub:lemma:joint-increase-in-forbidden}, $\frac{\delta_1}{r_1}$ must decrease by an arbitrarily large factor near $\taubp_1$ along every loop, contradicting the assumptions$\frac{\delta_1}{r_1}>\widetilde \epsilon$. This is because Lemma \ref{neartaub:lemma:joint-increase-in-forbidden} states that $\frac{\delta_1}{r_1}$ must decrease by a comparable factor as the increase of $r_1$, which  must increase by an arbitrarily large factor near $\taubp_1$, since we know that $r_1$ must eventually become larger than, say, $0.1$. 
\end{proof}

\begin{proof}[Proof of the global Bianchi \textsc{VIII} attractor Theorem \ref{thm:b8-attractor-global}]
The proof works like the above proof of Theorem \ref{thm:b9-attractor-global}: We use Lemma \ref{lemma:no-b8-taub-convergence} in order to prevent convergence to $\taubp_1$ and $\taubp_3$ (instead of Lemma \ref{lemma:farfromA:taub-instability}, which is still used near $\mathcal {TL}_2$). 

Preventing convergence to the heteroclinic cycle near $\mathcal T_2$ works as before. This leaves us with the possibility of convergence to one of the two remaining heteroclinic loops described in \textsc{Except}. Since we cannot exclude this case, we allow it in the conclusions of the theorem.
\end{proof}

\section{Phase-Space Volume and Integral Estimates}\label{sect:volume-form}
This section is devoted to proving the last three main Theorems of this work. The first one is the already stated genericity of the attracting case in Theorem \ref{thm:b8-attractor-global}:
\begin{thm}[Bianchi \textsc{VIII} global attractor genericity Theorem]\label{thm:b8-attractor-global-genericity}
In Theorem \ref{thm:b8-attractor-global}, 
the set of initial conditions $\bpt x_0$ for which \textsc{Attract} applies is ``generic'' in the following sense: It is open and dense in $\mathcal M_{+-+}$ and its complement has Lebesgue-measure zero. 
\end{thm}
The second theorem of this section answers affirmatively the ``locality'' part of the longstanding BKL conjecture for spatially homogeneous Bianchi class A vacuum spacetimes, for measure theoretic notions of genericity, and can be considered the main result of this work:
\begin{thm}[Almost sure formation of particle horizons]\label{thm:horizon-formation}
For Lebesgue almost every initial condition in $\mathcal M_{\pm\pm\pm}$ (with respect to the induced measure from $\RR^5=\{(\Sigma_+,\Sigma_-, N_1,N_2,N_3)\}$), the following holds:
\begin{equation}
\int_0^\infty \max_i \delta_i (t)\dd t = 2\int_{0}^\infty \max_{i\neq j}\sqrt{|N_iN_j|}(t)<\infty \qquad \text{for Lebesgue a.e. }\,\,\bpt x_0\in \mathcal M_{\pm\pm\pm}.
\end{equation}
This means that almost every Bianchi \textsc{VIII} and \textsc{IX} vacuum spacetime forms particle horizons towards the big bang singularity. This physical interpretation is described in Section \ref{sect:gr-phys-interpret}, Lemma \ref{lemma:horizon-integral} or \cite{heinzle2009future}.
\end{thm}
This result about Lebesgure a.e.~solutions immediately raises the question for counterexamples:
\begin{question}
It is currently not known, whether there exist any solutions, which are attracted to $\mathcal A$ and have infinite integral $\int_0^\infty\delta_i(t)=\infty$. 

\begin{center}Do such solutions exist? Is it possible to describe an example of such a solution?\end{center}

\noindent We very strongly expect that such solutions \emph{do} exist, for reasons which will be explained in future work.
\end{question}

The third and last theorem of this section strengthens and extends the previous result:
\begin{thm}[$L^p$ estimates for the generalized localization integral]\label{thm:horizon-formation-alpha-p}
Let either $\alpha\in (0,2)$ and $p\in (0,1)$ such that $\alpha p > 2p-1$, i.e.~$p< \frac{1}{2-\alpha}$.
Let $M\subset \mathcal M_{\pm\pm\pm}$ be a compact subset such that Theorem \ref{thm:local-attractor} holds for any initial condition in $\bpt x_0\in M$. Then $I_\alpha \in L^p(M)$, where
\begin{equation}
I_\alpha (\bpt x_0) := \int_0^\infty \max_i \delta_i^\alpha(t)\dd t,
\end{equation}
i.e., using $\phi$ for the flow to \eqref{eq:ode} and $\dd^4 \bpt x$ for the (four dimensional) Lebesgue measure on $\mathcal M_{\pm\pm\pm}$,
\[
\int_{M} \left[\int_0^\infty \delta_i^\alpha(\phi(\bpt x, t))\dd t\right]^p \dd^4 \bpt x<\infty\qquad\forall\,i\in\{1,2,3\}.
\]
If instead $\alpha\ge 2$, we already know from Theorem \ref{thm:local-attractor} that $I_\alpha \in L^\infty(M)$.
\end{thm}
Theorem \ref{thm:horizon-formation-alpha-p} makes a much stronger claim than Theorem \ref{thm:horizon-formation}, even for $\alpha=1$: Local $L^p$-integrability is a sufficient condition for a.e.~finiteness, but very much not necessary. On the other hand, we are aware of no immediate physical interpretation of Theorem \ref{thm:horizon-formation-alpha-p}. 

The proof of Theorem \ref{thm:horizon-formation-alpha-p} won't rely on Theorem \ref{thm:horizon-formation}. Even though the proof of Theorem \ref{thm:horizon-formation} is therefore entirely optional, we nevertheless choose to state and prove Theorem \ref{thm:horizon-formation} separately, because we view it as more important, and can give a more geometric proof than for Theorem \ref{thm:horizon-formation-alpha-p}.
\begin{question}
Unfortunately, the case $\alpha=1=p=1$, i.e.~$I_1 \in L^1_{\textrm{loc}}$, is maddeningly out of reach of Theorem \ref{thm:horizon-formation-alpha-p}, which only provides $I_1\in L^{1-\epsilon}_{\textrm{loc}}$ for any $\epsilon>0$. An extension to $\alpha=p=1$ would imply \emph{finite expectation} of the particle horizon integral.

Is it possible to say something about $\alpha=p=1$? Maybe for special compact subsets $M\subset \mathcal M_{\pm\pm\pm}$?
\end{question}

\paragraph{Outline of the proofs.}
Let us now give a short overview over the remainder of this section. Our primary tool will be a volume-form $\omega_4$, which is \emph{expanded} under the flow $\phi$ of \eqref{eq:ode}. An alternative description would be to say that we construct a density function, such that $\phi$ is volume-expanding. The volume-form will be constructed and discussed in Section \ref{sect:volume-construct}. We will use some very basic facts from the intersection of differential forms, measure theory and dynamical systems theory, which are given in Appendix \ref{sect:general-vol-app} for the convenience of the reader.

We will use this expanding volume-form in order to prove our three Theorems in Section \ref{sect:volume:proofs}.

\subsection{Volume Expansion}\label{sect:volume-construct}
This section studies the evolution of phase-space volumes. Using logarithmic coordinates $\beta_i=-\log|N_i|$, the equations differential equations \eqref{eq:ode} \emph{without} the constraint \eqref{eq:constraint} yield the remarkably simple and controllable formula $\DD_t \omega_5 = 2N^2 \omega_5$ for the evolution of the five-dimensional Lebesgue-measure $\omega_5$ (with respect to $\Sigma_\pm, \beta_i$). 
This formula shows that the flow $\phi$ expands the volume $\omega_5$. Such a volume expansion is impossible for systems living on a manifold with finite volume; it is possible with logarithmic coordinates because these coordinates have pushed the attractor to infinity, and typical solutions escape to infinity in these coordinates. 

\paragraph{Volume expansion for the extended system (without constraint $G=1$).}
Consider coordinates $\beta_i$ given by
\[\begin{aligned}
\beta_i &= -\log |N_i|&                  N_i &= \signN_i e^{-\beta_i}\\
\dd \beta_i &= -\frac{\dd N_i}{N_i}& \dd N_i &= -\signN_i e^{-\beta_i}\dd\beta_i\\
\partial_{\beta_i} &= -N_i \partial_{N_i} & \partial_{N_i}&=-\signN_i e^{\beta_i}\partial_{\beta_i},
\end{aligned}\]
and consider the Lebesgue-measure with respect to the $\beta_i$ coordinates:
\begin{equation}\label{eq:omega5-def}
\omega_5 = |\dd \Sigma_+ \land \dd \Sigma_- \land \dd \beta_1 \land \dd \beta_2 \land \dd \beta_3|,
\end{equation}
which is given in $N_i$ coordinates by
\[\begin{aligned}
\omega_5 &= \left|\frac{-1}{N_1N_2N_3}\dd \Sigma_+ \land \dd \Sigma_- \land \dd N_1 \land \dd N_2 \land \dd N_3\right|.
\end{aligned}\]
Let $\lambda(\bpt x, t)$ denote the volume expansion for $\phi(\bpt x, t)$, i.e.
\[
\phi^*(\bpt x, t)\omega_5 = \lambda(\bpt x, t)\omega_5,
\]
where $\phi^*$ is the pull-back acting on differential forms. Hence, in $(\bpt \Sigma,\bpt \beta)$-coordinates, $\lambda(\bpt x, t)=\det \partial_x \phi(\bpt x, t)$, and, with $f$ denoting the vectorfield corresponding to \eqref{eq:ode}:
\[\begin{aligned}
\DD_t \phi^*(\bpt x, t)\omega_5 &= \left[\mathrm{tr}\,\partial_x \partial_t \phi(\bpt x, t)\right] \phi^*(\bpt x, t)\omega_5 \\
&= \left[\partial_{\Sigma_+}f_{\Sigma_+}+\ldots+\partial_{\beta_3}f_{\beta_3}\right] \lambda(\bpt x, t)\omega_5=2N^2(\phi(\bpt x, t))\lambda(\bpt x, t)\omega_5\\
\lambda(\bpt x, t) &= \exp\left[ 2\int_0^t N^2(\phi(\bpt x, s))\dd s\right].
\end{aligned}\]
The volume is really expanding: In most of the phase space, $N^2>0$, and always $N^2> -4|N_1N_2N_3|^{\frac 2 3}$, which has bounded integral.

\paragraph{Volume expansion on $\mathcal M$.}
We are not really interested in the behaviour of $\phi$ on $\RR^5$, and the measure $\omega_5$. Instead, we are interested in dynamics and measures on the set $\mathcal M=\{\bpt x\in \RR^5:\, G(\bpt x)=1\}$. We can get an induced measure on $\mathcal M$ by choosing a vectorfield $X:\RR^5\to T\RR^5$ such that $\DD_X G =1$ in a neighborhood of $\mathcal M$. Then we set
\begin{equation}\label{eq:omega4-def}
\begin{multlined}
\omega_4=\iota_{X}\omega_5,\\
\text{i.e.}\qquad\omega_4[X_1,\ldots, X_4]=\omega_5[X,X_1,\ldots X_4]\qquad\text{for}\quad X_1,\ldots, X_4\in T\mathcal M.
\end{multlined}\end{equation}
This induced volume is independent of the choice of $X$ (as long as $\DD_X G=1$), and fulfills (see Section \ref{sect:general-vol-app})
\[
\phi^*(\bpt x, t)\omega_4 = \lambda(\bpt x, t)\omega_4 = \frac{\phi^*(\bpt x, t)\omega_5}{\omega_5}\omega_4.
\]

\paragraph{Volume expansion for Poincar\'e-maps.}
We consider the volume-form $\omega_3=\iota_f \omega_4$, where $f$ is the vectorfield \eqref{eq:ode}. By invariance of $f$ under $\phi$, we again have $\phi^*(\bpt x, t)\omega_3=\lambda(\bpt x, t)\omega_3$. 

Let $U\subseteq \mathcal M$ open and $T:U\to \RR$ differentiable. Then the map $\Phi_T:U\to \mathcal M$ with $\Phi_T(\bpt x)=\phi(\bpt x, T(\bpt x))$ has $\Phi_T^*\omega_3=\lambda(\bpt x, T(\bpt x))\omega_3|U$ (also see Section \ref{sect:general-vol-app}).

This especially holds when $S\subset M$ is a Poincar\'e-section and $\Phi_S=\Phi_{T_S}$ is the corresponding Poincar\'e-map, i.e.~when $S\subset \mathcal M$ is a submanifold of codimension one, which is transverse to $f$, and $T=T_S(\bpt x)= \inf\{t>0:\,\phi(\bpt x, t)\in S\}$.

If $S\subseteq \mathcal M$ is a Poincar\'e-section and $K\subseteq S$ is a set with $|K|_{\omega_3}=\int_{S}\mathrm{id}_{K}\omega_3 =0$, then, by Fubini's Theorem, $|\phi(K, \RR)|_{\omega_4}=0$. The measure $\omega_4$ is absolutely bi-continuous with respect to the ordinary Lebesgue measure, i.e.~the notions of sets of measure zero coincide for the ordinary Lebesgue-measure and $\omega_4$.

This especially applies to the boundary $\partial S$ with $|\partial S|=0$. Therefore, if $S_0,S$ are two Poincar\'e-sections and $K\subseteq S_0\subseteq M$ is a set of initial conditions, such that $T_S(\bpt x)<\infty$ for almost every $\bpt x\in K$, then
\[
\int_{\Phi_S(K)\subseteq S} f(\bpt x)\omega_3 = \int_{K\subseteq S_0}f(\Phi_S^{-1}(\bpt x))\lambda(\bpt x, T_S(\bpt x))\omega_3,
\]
for any $L^\infty$ function $f$. This is because the map $\Phi_S$ is, by assumption, sufficiently smooth almost everywhere. 

Our primary source of ``almost everywhere'' statements is the following:
\begin{lemma}\label{lemma:vol-dichotomy}
Let $S\subset \mathcal M$ be a Poincar\'e-section with $N^2>\epsilon>0$, and let $K\subseteq S \subset \mathcal M$ be forward invariant, i.e. $\Phi_S$ is well-defined in $K$ and $\Phi_S(K)\subseteq K$. 

Then either $|K|_{\omega_3}=0$ or $|K|_{\omega_3}=\infty$. If $|K|_{\omega_3}=0$, then $|\phi(K,\RR)|_{\omega_4}=0$.
\end{lemma}
\begin{proof}
Since $\Phi_S(K)\subseteq K$, we have
\[|K|_{\omega_3}\ge |\Phi_S(K)|_{\omega_3} = \int_S \mathrm{id}_{\Phi_S(K)}\omega_3 = \int_S \mathrm{id}_K \phi_S^*\omega_3 = \int_S \mathrm{id}_K \lambda \omega_3 \ge q |K|_{\omega_3},\]
where $q=\inf_{\bpt x \in K}\lambda(\bpt x, T_S(\bpt x))>1$. Therefore, either $|K|_{\omega_3}=0$ or $|K|_{\omega_3}=\infty$.
If $|K|_{\omega_3}=0$ then this implies $|\phi(K,[-h,h])|_{\omega_4}=0$ for small $h>0$ by Fubini's Theorem and hence $|\phi(K,\RR)|_{\omega_4}=0$.
\end{proof}
%
%
%
%
\subsection{Proofs of the main Theorems}\label{sect:volume:proofs}
We will now use the $\omega_3$-expansion between Poincar\'e-sections in order give proofs of the main results.
%

%
\begin{proof}[Proof of Theorem \ref{thm:b8-attractor-global-genericity} (Bianchi \textsc{VIII} global attractor genericity)]\label{proof:thm:b8-attractor-global-genericity}
\textsc{Attract} holds for an open set of initial conditions (by virtue of continuity of the flow). Since \textsc{Taub} can only happen on an embedded submanifold of lower dimension, it suffices to prove that \textsc{Except} happens only for a set of initial conditions with Lebesgue measure zero.

Let \textsc{Except} also denote the set of initial conditions for which the case \textsc{Except} holds. Without loss of generality, we will consider the case $\mathcal M_{+-+}$ and trajectories bouncing between $+\taubp_1$ and $-\taubp_1$.

We choose a small Poincar\'e-section 
\[S\subseteq\{\bpt x\in\mathcal M:\,N_1 = \textrm{const}=h,\, r_1 \le \epsilon,\, \delta_1\le \epsilon\},\]
intersecting the unique heteroclinic $W^u(-\taubp_1)$, which connects $-\taubp_1$ to $\taubp_1$, near $\taubp_1$, such that $S$ is a graph over $(\Sigma_-, N_2, N_3)$. Let $T_S:S\to (0,\infty]$ be the (partially defined) recurrence time and $\Phi:S\to S$ with $\Phi(\bpt x)= \phi(\bpt x, T(\bpt x))$ be the (partially defined) Poincar\'e-map to $S$.

\textsc{Except} is (by definition) invariant under the flow and $\Phi(\textsc{Except}\cap S)\subseteq \textsc{Except}\cap S$. Hence $\textsc{Except}\cap S$ has either vanishing or infinite $\omega_3$-volume, by Lemma \ref{lemma:vol-dichotomy}. Therefore, it suffices to prove $|\textsc{Except}\cap S|_{\omega_3}<\infty$.

Recalling the local attractor Theorem \ref{thm:local-attractor}, we can note $\delta_1 > \widetilde \epsilon r_1$ for every $\bpt x\in \textsc{Except}\cap S$. This allows us to estimate
\begin{subequations}
\label{eq:bad-finite-measure}\begin{align}
|\textsc{Except}\cap S|_{\omega_3} &\le \left|\{\bpt x\in S:\, \delta_1 > \widetilde \epsilon r_1\}\right|_{\omega_3}
\le \left|\{\bpt x\in S:\, \delta_1 > \widetilde \epsilon |\Sigma_-|\}\right|_{\omega_3}\\
&\le \left|\{\bpt x\in S:\,\beta_2+\beta_3 > C+C|\log|\Sigma_-||\}\right|_{\omega_3}\\
&= \int_{\{\bpt x\in S:\, \beta_2+\beta_3 > C+C|\log|\Sigma_-||\}} \omega_3[\ldots] \ |\dd \Sigma_-\land \dd \beta_2 \land\dd \beta_3|\label{eq:bad-finite-measure:last-s}\\
&\le C \int_{\{\bpt x\in S:\, \beta_2+\beta_3 > C+C|\log|\Sigma_-||\}}|\dd \Sigma_-\land \dd \beta_2 \land\dd \beta_3|\label{eq:bad-finite-measure:last-m1}\\
&\le C+ C \int\left|\log |\Sigma_-|\right|^2 \dd \Sigma_- <\infty\label{eq:bad-finite-measure:last},
\end{align}
\end{subequations}
where we integrated $\beta_2,\beta_3$, using $\beta_2,\beta_3\ge 0$ from \eqref{eq:bad-finite-measure:last-m1} to \eqref{eq:bad-finite-measure:last}. In order to go from \eqref{eq:bad-finite-measure:last-s} to \eqref{eq:bad-finite-measure:last-m1}, we used that (in $S$) $N_1=\mathrm{const}=h$ and $\Sigma_+=\Sigma_+(\Sigma_-, \beta_2,\beta_3)$ and 
\[\begin{aligned}
\omega_3 &= |\dd \Sigma_-\land \dd \beta_2 \land \dd \beta_3|\\
&\qquad\cdot\,\left|\omega_5\left[f,
|\partial_{\beta_1} G|^{-1} \partial_{\beta_1},
\partial_{\Sigma_-}+\partial_{\Sigma_-}\Sigma_+\partial_{\Sigma_-}, 
\partial_{\beta_2}+\partial_{\beta_2}\Sigma_+\partial_{\Sigma_+},
\partial_{\beta_3}+\partial_{\beta_3}\Sigma_+\partial_{\Sigma_+}\right]\right|\\
&\le C |\dd \Sigma_-\land \dd \beta_2 \land \dd \beta_3|.
\end{aligned}\]
This can be seen by noting $\partial_{\beta_1} G  = 2 N_1^2 - 2 N_1N_- \ge C>0$ if $\epsilon>0$ is small enough and likewise $f_{\Sigma_+}> C >0$, $|\partial_{\Sigma_-, \beta_2,\beta_3}\Sigma_+| \le \mathcal O(\epsilon)$.
\end{proof}
We will now give the first, more geometric proof of Theorem \ref{thm:horizon-formation}. Readers who prefer the arithmetic variant can skip ahead to the proof of Theorem \ref{thm:horizon-formation-alpha-p}, page \pageref{proof:thm:horizon-formation-alpha-p}. The theorem follows from the following two Lemmas:

\begin{lemma}\label{lemma:badrecurrent-zeroset}
Fix $i\in\{1,2,3\}$, without loss of generality $i=1$. 
Consider a small Poincar\'e-section $S$ as in the proof of Theorem \ref{thm:b8-attractor-global-genericity}, i.e. $S\subseteq \mathcal M_{\signN}$ intersecting the unique heteroclinic connecting $-\taubp_1$ to $\taubp_1$,  near 
$\taubp_1$ such that Proposition \ref{prop:neartaub:main} holds and $S$ is a smooth graph over $\Sigma_-, N_2, N_3$.

Denote
\[\begin{aligned}
\textsc{Bad}&=\{\bpt x\in S:\, \delta_1 > |\Sigma_-|^4\}\\
\textsc{BadRecurrent}&=
\{\bpt x\in S:\, \bpt \Phi_S^k(x)\in\textsc{Bad}\, \text{for infinitely many $k$}\},
\end{aligned}\]
where $\Phi_S$ is the Poincar\'e-map to $S$. Then
\(
\left|\textsc{BadRecurrent}\right|_{\omega_3}=0.
\)
\end{lemma}
\begin{lemma}\label{lemma:badrecurrent-small-int}
Consider the setting of Lemma \ref{lemma:badrecurrent-zeroset}.

Let $\bpt x_0\in S\cap \textsc{Basin}$ such that $\bpt x_0\not\in \textsc{BadRecurrent}$. Then $\int_0^\infty \delta_1(\phi(\bpt x_0, t))\dd t<\infty$.
\end{lemma}
\begin{proof}[Proof of Theorem \ref{thm:horizon-formation} (almost sure formation of particle horizons)]
Follows trivially from Lemma \ref{lemma:badrecurrent-small-int} and Lemma \ref{lemma:badrecurrent-zeroset}.
\end{proof}
\begin{proof}[Proof of Lemma \ref{lemma:badrecurrent-small-int}]
Let $\bpt x_0\not\in \textsc{BadRecurrent}$, and let $(T_n)_{n\in\NN}\subseteq \RR$ and $(\bpt x_n)_{n\in\NN}\subseteq S$ be the recurrence times and points in $S$, i.e.~$\bpt x_{n+1}=\Phi_S(\bpt x_n)=\phi(\bpt x_{n}, T_n)$, $T_n = T_S(\bpt x_{n-1})$. By assumptions, we have some $N>0$ such that, for all $n\in\NN$, $\bpt x_{N+k}\not\in\textsc{Bad}$. We can estimate
\[\begin{aligned}
\int_0^\infty\delta_1(\phi(\bpt x_0, t))\dd t &= \int_{0}^{T_n}\delta_1(\phi(\bpt x_0, t))\dd t + \sum_{n=N}^{\infty}\int_{0}^{T_n}\delta_1(\phi(\bpt x_{n}, t ))\dd t\\
&\le C(\bpt x_0) + \sum_{n=N}^{\infty}\int_{0}^{T_n}C\exp\{-\frac{C}{r_1^2(\bpt x_n)}\}\delta_1(\bpt x_n)\dd t\\
& \le C(\bpt x_0)+ \sum_{n=N}^{\infty}C\frac{\delta_1}{r_1^2}(\bpt x_n)\\
& \le C(\bpt x_0)+ \sum_{n=N}^{\infty}C\sqrt{\delta_1}(\bpt x_n) <\infty,
\end{aligned}\]
where we used Proposition \ref{prop:neartaub:main} and $\delta_1 < r_1^4$ and the fact that $\delta_1(\bpt x_{n+1})< \frac 1 2 \delta_1(\bpt x_n)$.
\end{proof}
\begin{proof}[Proof of Lemma \ref{lemma:badrecurrent-zeroset}]
Analogously to the proof of Theorem \ref{thm:b8-attractor-global-genericity}, we have $|\textsc{Bad}|_{\omega_3}<\infty$. Then, we can write
\[\begin{aligned}
\left|\textsc{BadRecurrent}\right|_{\omega_3}&=\left|\bigcap_{n\in\NN}\bigcup_{k\ge n} \Phi^{-k}(\textsc{Bad})\right|_{\omega_3}\\
&\le \lim_{n\to 0} \sum_{k\ge n}\left|\Phi^{-k}(\textsc{Bad})\right|_{\omega_3} 
\le \lim_{n\to 0} \sum_{k\ge n}q^{-k}\left|\textsc{Bad}\right|_{\omega_3}<\infty,
\end{aligned}\]
where $q=\inf \{\lambda(\bpt x, T_S(\bpt x)): \bpt x\in S\}>1$.
\end{proof}
\begin{proof}[Proof of Theorem \ref{thm:horizon-formation-alpha-p} ($L^p$ estimates for the generalized localization integral)]\label{proof:thm:horizon-formation-alpha-p}
The claim for $\alpha\in [2,\infty)$ follows trivially from Theorem \ref{thm:local-attractor}. For the other case, we restrict our attention without loss of generality to $\delta_1$. Furthermore, we can without loss of generality assume that we start with $C\subseteq S$ for some Poincar\'e section and estimate the integral with respect to $\omega_3$. Recall the construction in the proof of Theorem \ref{thm:b8-attractor-global-genericity} with a Poincar\'e-section $S_1$ near the heteroclinic $-\taubp_1\to \taubp_1$. We can estimate, for some positive $s\in (2p-1, \alpha p)$: 
\begin{subequations}
\begin{align*}
\int_{\bpt x\in C} \biggl[\int_0^\infty \delta_i^\alpha(&\phi(\bpt x, t))\dd t\biggr]^p \omega_3
= \int_{\bpt x\in C} \left[\sum_{n} \int_{T_{n}}^{T_{n+1}}\delta_i^\alpha(\phi(\bpt x, t))\dd t\right]^p \omega_3\\
&\le \sum_{n} \int_{\bpt x\in C} \left[\int_{T_{n}}^{T_{n+1}}\delta_i^\alpha(\phi(\bpt x, t))\dd t\right]^p \omega_3\numberthis\label{eq:foobar:p}\\
&\le C\sum_{n} \int_{\bpt x\in C} \left[\frac{\delta_i^\alpha}{r_{i}^{2}}(T_n)\right]^p \omega_3\\
&= C\sum_{n} \int_{\bpt x\in C}  \left[ \delta_i^{\alpha p -s} r_i^{-2p +s} \left(\frac{\delta_i}{r_i}\right)^{s}\right](T_n)\omega_3\\
&\le C\sum_{n}\, \sup_{\bpt x\in C} \left(\frac{\delta_i}{r_i}\right)^{s}(T_n)\,\,\cdot\,\,  \int_{\bpt x\in C} \left[\delta_i^{\alpha p -s} r_i^{-2p +s}\right] (T_n)\omega_3.\numberthis\label{eq:foobar:hoelder}
\end{align*}
\end{subequations}

We have used $p<1$ in order to split the integral in \eqref{eq:foobar:p} and the H\"older inequality in \eqref{eq:foobar:hoelder}.
 We continue the estimates by noting that $\sup_{\bpt x\in C}\frac {\delta_i} {r_i}(T_n)$ decreases exponentially in $n$. Hence we only need to bound the second factor. This can be done by using $\alpha p -s>0$ and $-2p+s > -1$ in order to see
\[\begin{aligned}
\int_{\bpt x\in C} \left[\delta_i^{\alpha p -s} r_i^{-2p +s}\right] (T_n)\omega_3 
&\le C \int_{S_1} \left[\delta_i^{\alpha p -s} r_i^{-2p +s}\right] \omega_3\\
&\le C \int_{-0.1}^{0.1} \left[\int_{\beta_2,\beta_3\ge 0} e^{-\frac{(\beta_2+\beta_3)(\alpha p-s)}{2}}|\Sigma_-|^{-2p+s} \dd \beta_2\land\dd \beta_3  \right] \dd \Sigma_-\\
&\le C \int_{-0.1}^{0.1}|\Sigma_-|^{-2p+s} \dd \Sigma_-  < \infty.
\end{aligned}\]
\end{proof}

\section{Physical Properties of Solutions for Bianchi \textsc{VIII} and \textsc{IX}}\label{sect:gr-phys-interpret}
We will now use the results of this work in order to describe some physical properties of Bianchi spacetimes. Recall Section \ref{sect:equations}, where we describe Bianchi spacetimes in terms of the Wainwright-Hsu equations. In Section \ref{sect:append:derive-eq}

\paragraph{Bounded life-time.}
Since the mean curvature $H$ corresponds to the time-derivative of the spatial volume form $\sqrt{g_{11}g_{22}g_{33}}$, the universe described by such a metric is contracting for $H<0$. Physically, we are interested in the behaviour towards the initial (big bang) singularity; this setting is time-reversed to physical time-variables, and we should look at the behaviour of solutions for $t\to +\infty$. Since $|H|$ is at least uniformly exponentially growing, we can immediately see \( \textsc{EigenFuture} = \int_{0}^\infty \sqrt{-g_{00}}\dd t <\infty,\) that is, the universe has only a finite (eigen-) lifetime until $H$ blows up and a singularity occurs. In our coordinates, this singularity is placed at $t=+\infty$.

A priori, we cannot know whether this singularity is a physical singularity (with curvature blow-up) or just a coordinate singularity (it might be that it is just our coordinate system, which blows up, while the actual space-time remains regular). In the case of Bianchi \textsc{VIII} and \textsc{IX}, the singularity is physical (with curvature blow-up). 
It has been proven in \cite{ringstrom2001bianchi} that the singularity is physical and curvature blows up. This is done by considering the so-called Kretschmann scalar
$\kappa = \sum_{\alpha,\beta,\gamma,\delta}R_{\alpha,\beta}^{\ \ \delta,\gamma} R_{\delta,\gamma}^{\alpha,\beta}$ and showing that $\lim_{t\to\infty}\kappa(t)=\infty$. We refer to \cite{ringstrom2001bianchi} for the details.
\paragraph{Bounded spatial metric coefficients.}
Now we restrict our attention to the case of Bianchi \textsc{IX} and \textsc{VIII}, where all $\hat n_i\neq 0$. 

\noindent The coefficients $g_{ii}=\frac 1 {48}\frac {|N_i|}{H^2 |N_1N_2N_3|}$ stay bounded: We know, ssing the global attractor Theorems \ref{thm:b9-attractor-global} and \ref{thm:b8-attractor-global}, that the $|N_i|$ stay bounded. We can compute: 
\[
\DD_t \log |H^2N_1N_2N_3|= -3\Sigma^2 + 1 + 2\Sigma^2 = N^2 \ge -4|N_1N_2N_3|^{\frac 2 3},
\]
which shows that $|H^2N_1N_2N_3|$ stays bounded away from zero, using Lemma \ref{farfrom-A:lemma:uniform}. Hence all three $g_{ii}$ stay bounded for $t\to+\infty$. Indeed, since $N^2>\epsilon>0$ for large amounts of time, we can conclude $\lim_{t\to+\infty} g_{ii}(t)=0$.

\paragraph{Particle Horizons.}
Recall question of particle horizons from the introduction, and the definition of communication cones \eqref{eq:intro:comcone}, which we here adjust to match our convention that the big bang singularity is situated in the future:
\[\begin{aligned}
&\text{Singularity directed light cone of $\bpt p$:}\\
&\quad J^-(\bpt p) = \{\bpt q:\, \text{there is } \gamma:[0,1]\to M\,\,\text{with}\, \gamma(0)=\bpt p, \gamma(1)=\bpt q,\,\text{time-like future directed}\}\\
&\text{Non-singularity directed light cone of $\bpt p$:}\\
&\quad J^+(\bpt p) = \{\bpt q:\, \text{there is } \gamma:[0,1]\to M\,\,\text{with}\, \gamma(0)=\bpt p, \gamma(1)=\bpt q,\,\text{time-like past directed}\}\\
&\text{Communication cone of $\bpt p$:}&\\
&\quad \phantom{\partial}J^+(J^-(\bpt p)) = \bigcup_{\bpt q \in J^-(\bpt p)}J^+(\bpt q)\\
&\text{Cosmic horizon of $\bpt p$:}&\\
&\quad \partial J^+(J^-(\bpt p)) = \text{the topological boundary of the past communication cone.}
\end{aligned}\]
We can now relate the question of particle horizons with our estimates on $\int \sqrt{|N_iN_j|}(t) \dd t$. This gives the physical interpretation of Theorem \ref{thm:horizon-formation}. Remember that time is oriented such that the big bang singularity is situated in the future at $t=\infty$.
\begin{lemma}\label{lemma:horizon-integral}
There is a constant $C>0$, such that, for Bianchi \textsc{IX} and \textsc{VIII} vacuum spacetimes $M$, we can estimate for $\bpt p\in M$ and $t_0\ge t(\bpt p)$
\begin{equation}\label{eq:comcone-estimate}\begin{aligned}
\mathrm{diam}_h \left[J^-(\bpt p)) \cap \{\bpt q\in M: t(\bpt q)=t_0\}\right] &\le C\int_{t(\bpt p)}^{t_0}  \max_{j\neq k} \sqrt{|N_jN_k|}(t)\dd t\\
\mathrm{diam}_h \left[J^+(J^-(\bpt p)) \cap \{\bpt q\in M: t(\bpt q)=t(\bpt p)\}\right] &\le C \int_{t(\bpt p)}^\infty \max_{j\neq k} \sqrt{|N_jN_k|}(t)\dd t,
\end{aligned}\end{equation}
where the diameter is measured with the symmetry metric $h$ given on the surfaces of homogeneity $\{t=\mathrm{const}\}$ by $h=\omega_1 \otimes \omega_1 + \omega_2\otimes \omega_2 + \omega_3\otimes \omega_3$. 
Furthermore, suppose that $M$ is a spacetime corresponding to a Bianchi \textsc{VIII} or \textsc{IX} solution with $\int_0^\infty \delta_i(t)\dd t <\infty$ for all $i\in\{1,2,3\}$, as in the conclusion of Theorem \ref{thm:horizon-formation}. Use the shorthand $J^+(J^-(t_0))=J^+(J^-(\bpt p))$ for some $\bpt p\in M$ with $t(\bpt p)=t_0$.
Then the following holds:
\begin{enumerate}
\item $\lim_{t_0\to\infty} \mathrm{diam}_h\left[J^+(J^-(t_0))\cap \{\bpt q\in M: t(\bpt q)=t_0\}\right]=0$.
\item $\lim_{t_0\to\infty} \mathrm{diam}_g\left[J^+(J^-(t_0))\cap \{\bpt q\in M: t(\bpt q)=t_0\}\right]=0$.
\item For $t_0>0$ large enough, $\partial J^+(J^-(t_0))\neq \emptyset$.
\item For $t_0>0$ large enough, the communication cone is homeomorphic to $(0,1)$ times the three-dimensional unit-ball with boundary. In other words, the following manifolds with boundary are homeomorphic, where $B_1^{\RR^3}(0)$ is the three-dimensional unit-ball:
\begin{multline}\nonumber
\left[J^+(J^-(t_0)) \cap \{\bpt q\in M: t(q) > t_0\}\,,\, \partial J^+(J^-(t_0)) \cap \{\bpt q\in M: t(q) > t_0\}\right]\\
\sim \left[(0,1) \times B_1^{\RR^3}(0)\,,\, (0,1)\times \partial B_1^{\RR^3}(0)\right]\end{multline}
\end{enumerate}
\end{lemma}
\begin{proof}
Any time-like singularity-directed curve $\gamma$ starting in $\bpt p$ must fulfill $|\DD_t \gamma_i| \le \sqrt{-g_{00} g^{ii}}=\sqrt{12}\sqrt{\widetilde N_j \widetilde N_k}$. It is clear that the $h$-length of such a curve must be bounded by $C\max_{j\neq k} \int_{t(\bpt p)}^\infty \sqrt{\widetilde N_j \widetilde N_k}(t)\dd t$ (parametrized over the time $t$ corresponding to \eqref{eq:ode-from-gr}), if the curve only accesses times later than $t_0$. This proves the estimate on $J^-(\bpt p)$. The estimate on $J^+(J^-(\bpt p))$ follows.

If $\int_{t_0}^{\infty}\max_{j\neq k}\sqrt{|N_jN_k|}(t)\dd t <\infty$, then $\lim_{t_0\to\infty}\int_{t_0}^{\infty}\max_{j\neq k}\sqrt{|N_jN_k|}(t)\dd t=0$, which proves $(1)$. $(2)$ follows because the metric coefficients $g_{ij}$ are bounded. $(4)$ follows because the injectivity radius $\mathrm{inj}_h(t)$ of the hypersurfaces of homogeneity is independent of the time $t$, if we measure it with respect to the (time-independent) $h$-metric. $(3)$ follows trivially from $(4)$.
\end{proof}
Hence, Theorem \ref{thm:horizon-formation} really shows that almost every Bianchi \textsc{VIII} and \textsc{IX} vacuum solution forms particle horizons.

\begin{figure}[hbt]
\centering
\begin{subfigure}[t]{0.48\textwidth}
\centering
\includegraphics[width=\textwidth]{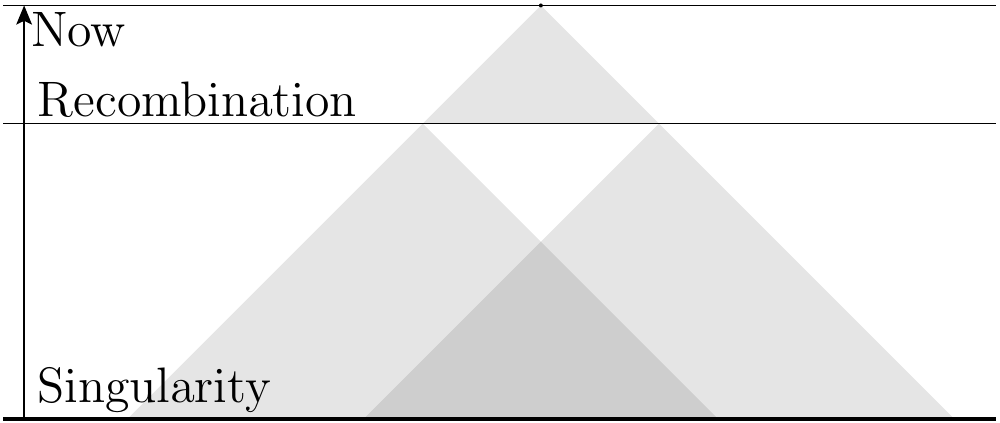}
\caption{A sketched space-time, where homogeneity of the observable universe could be explained by mixing between the big bang and recombination.}
\end{subfigure}~~
\begin{subfigure}[t]{0.48\textwidth}
\centering
\includegraphics[width=\textwidth]{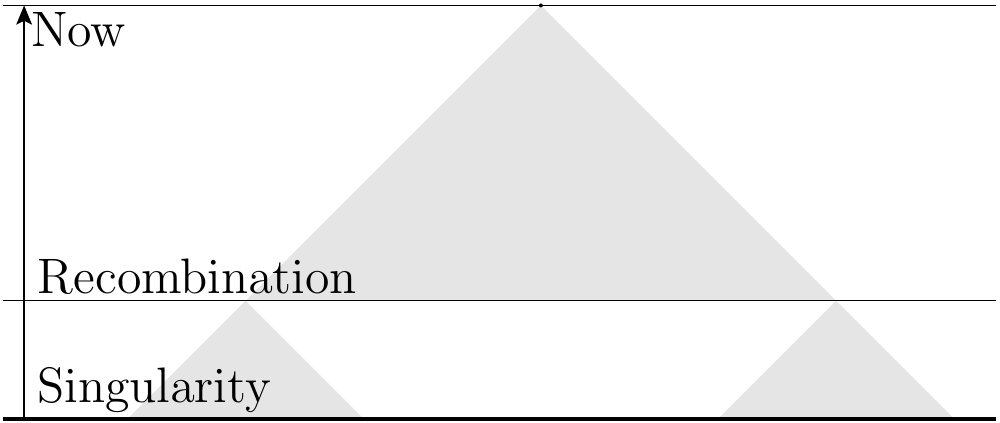}
\caption{A sketched space-time, where homogeneity of the observable universe cannot be explained by mixing between the big bang and recombination.}
\end{subfigure}
\caption{Whether observed homogeneity could be explained by mixing depends on the relation of the conformal distance between recombination and singularity versus the conformal distance between the present time and recombination.}\label{fig:recomb-cone}
\end{figure}
In the introduction, we promised to give some more details on the relation of the formation of particle horizons to homogenization of the universe by mixing. Astronomical observations show that the universe appears to be mostly homogeneous at large scales. However, optical astronomical observations go back only to the recombination, the moment where the primordial plasma condensed to a gas and became transparent to light. Hence, there is only a bounded region of space-time in our past, which is optically accessible. A possible explanation for the observed homogeneity, c.f.~e.g.~\cite{misner1969mixmaster}, might be that the universe, i.e.~matter, radiation, etc, mixed in the time-frame between the initial singularity and recombination. This is only possible if the outer parts of our optical past light-cone have a joint causal past (see Figure \ref{fig:recomb-cone}). This is a reason that the formation of particle horizons does not necessarily spell doom for attempts to explain the observed homogeneity through mixing 
(apart from the actual universe not beeing a homogeneous, anisotropic vacuum spacetime of Bianchi type $\textsc{VIII}$ or $\textsc{IX}$).

Indeed, the same problem is present in the ``standard model'' of cosmology, which is a (homogeneous isotropic) FLRW model. There, the observed homogeneity is typically explained by inflation, i.e.~one postulates a phase of rapid expansion that increases the conformal distance between recombination and singularity, mediated by an exotic, hitherto unobserved matter field.

\newpage
\appendix
\section{Appendix}
\subsection{Glossary of Equations and Notations}\label{sect:eq-cheat-sheet}
For easier reference, we compressed the most frequently referenced equations and notations on two pages. 
The Wainwright-Hsu equations are given by:
\begin{align*}
 N_i' &=  -(\Sigma^2 +2 \langle \taubp_i,\bpt\Sigma\rangle)N_i  \tag{\ref{eq:ode-ni}}\\
 &= -\left(\left|\bpt\Sigma+\taubp_i\right|^2 -1\right)N_i \tag{\ref{eq:ode2-ni}}\\
\bpt\Sigma' &= N^2 \bpt\Sigma + 2 \threemat{\taubp_1}{\taubp_2}{\taubp_3}{\taubp_3}{\taubp_1}{\taubp_2}[\bpt N, \bpt N],\quad\text{where}\tag{\ref{eq:ode2-sigma}}\\
\taubp_1 &= (-1,0) \qquad \taubp_2 = \left(\frac 1 2, -\frac 1 2 \sqrt{3}\right) \qquad \taubp_3 = \left(\frac 1 2, \frac 1 2 \sqrt{3} \right) \tag{\ref{eq:taub-def}}\\
1&\overset{!}{=} \Sigma^2+N^2 = \Sigma_+^2+\Sigma_-^2 + N_1^2+N_2^2+N_3^2-2(N_1N_2+N_2N_3+N_3N_1). \tag{\ref{eq:constraint}}
\end{align*}
Auxilliary quantities are given by:
\begin{align*}
\delta_i &= 2\sqrt{|N_jN_k|} \tag{\ref{eq:delta-r-def-delta}}\\
r_i &= \sqrt{(|N_j|-|N_k|)^2 + \frac{1}{3}\langle\taubp_j-\taubp_k, \bpt\Sigma\rangle^2} \tag{\ref{eq:delta-r-def-r}}\\
\delta_i' &= -\left(\left|\bpt\Sigma-\frac{\taubp_i}{2}\right|^2-\frac{1}{4}\right)\delta_{i}.\tag{\ref{eq:ode2-delta}}
\end{align*}
In polar coordinates, the equations around $r_1 \approx 0$ become for $N_2,N_3>0$:
\begin{align*}
r_1 &= \Sigma_-^2+ N_-^2 \quad N_-=N_3-N_2\quad N_+=N_3+N_2\\
\DD_t \log r_1&= N^2 - (\Sigma_++1) \frac{N_-^2}{r_1^2} +\sqrt{3} N_1 \frac{\Sigma_- N_-}{r_1^2} \tag{\ref{eq:neartaub-b9-q:r}}\\
&=r_1^2\sin^2\psi \frac{-\Sigma_+}{1-\Sigma_+} + N_1\, h_r\tag{\ref{eq:neartaub-b9-t:r}}\\
\DD_t \log \delta_1 &= N^2-(\Sigma_++1)\tag{\ref{eq:neartaub-b9-q:delta}}\\
&= \frac{-1}{1-\Sigma_+}r_1^2\cos^2\psi + \frac{-\Sigma_+}{1-\Sigma_+}r_1^2\sin^2\psi + N_1\,h_\delta \tag{\ref{eq:neartaub-b9-t:delta}}\\
\DD_t \log \frac{\delta_1}{r_1} &= -(\Sigma_+ + 1)\frac{\Sigma_-^2}{r_1^2} -\sqrt{3}N_1\frac{\Sigma_- N_-}{r_1^2} \tag{\ref{eq:neartaub-b9-q:delta-r}}\\
&=\frac{-1}{1-\Sigma_+}r_1^2\cos^2\psi +N_1(h_\delta-h_r)\tag{\ref{eq:neartaub-b9-t:delta-r}}\\
\psi' &=\sqrt{3}r_1 \sqrt{\sin^2\psi + \frac{\delta_1^2}{r_1^2}} - \frac{r_1^2}{1-\Sigma_+} \sin\psi \cos\psi + N_1\sin\psi\, h_{\psi},\tag{\ref{eq:neartaub-b9-t:psi}}
\end{align*}
\noindent
where the terms $|h_r|,|h_\delta|, |h_\psi|$ are bounded (if $|N_i|$, $\Sigma_+<0$, and $\Sigma_-$ are bounded) and given in \eqref{eq:neartaub-b9-t-h}, page \pageref{eq:neartaub-b9-t-h}.

In polar coordinates, the equations around $r_1 \approx 0$ become for $N_2>0$ , $N_3<0$:
\begin{align*}
r_1 &= \Sigma_-^2+ N_-^2 \quad N_-=N_2+N_3\quad N_+=N_2-N_3\\
\DD_t \log r_1&= N^2 - (\Sigma_++1) \frac{N_-^2}{r_1^2} +\sqrt{3} N_1 \frac{\Sigma_- N_+}{r_1^2} \tag{\ref{eq:neartaub-b8-q:r}}\\
&=\frac{-\Sigma_+}{1-\Sigma_+}r_1^2\sin^2\psi +\delta_1^2\frac{\cos^2\psi-\Sigma_+}{1-\Sigma_+} + N_1\, h_r\tag{\ref{eq:neartaub-b8-t:r}}\\
\DD_t \log \delta_1 &= N^2-(\Sigma_++1)\tag{\ref{eq:neartaub-b9-q:delta}}\\
&=\frac{-1}{1-\Sigma_+}r_1^2\cos^2\psi + \frac{-\Sigma_+}{1-\Sigma_+}r_1^2\sin^2\psi +\frac{-\Sigma_+}{1-\Sigma_+}\delta_1^2+ N_1\,h_\delta \tag{\ref{eq:neartaub-b8-t:delta}}\\
\DD_t \log \frac{\delta_1}{r_1} &= -(\Sigma_+ + 1)\frac{\Sigma_-^2}{r_1^2} -\sqrt{3}N_1\frac{\Sigma_- N_+}{r_1^2} \tag{\ref{eq:neartaub-b8-q:delta-r}}\\
&=\frac{-1}{1-\Sigma_+}r_1^2\cos^2\psi - \delta_1^2\frac{\sin^2\psi}{1-\Sigma_+} + N_1(h_\delta-h_r)\tag{\ref{eq:neartaub-b8-t:delta-r}}\\
\psi' &=\sqrt{3}r_1 \sqrt{\cos^2\psi + \frac{\delta_1^2}{r_1^2}} - \frac{r_1^2+\delta_1^2}{1-\Sigma_+} \cos\psi \sin\psi + N_1\cos\psi\, h_{\psi},\tag{\ref{eq:neartaub-b8-t:psi}}
\end{align*}
\noindent
where the terms $|h_r|,|h_\delta|, |h_\psi|$ are bounded (if $|N_i|$, $\Sigma_+<0$, $\Sigma_-$, and $\frac{\delta_1}{r_1}$ are bounded) and given in \eqref{eq:neartaub-b8-t-h}, page \pageref{eq:neartaub-b8-t-h}.

We use $\mathcal M= \{\bpt x\in \RR^5:\,G(\bpt x)=1\}$, and $\mathcal M_{\signN}\subset \mathcal M$ to denote the signs of the three $N_i$, with $\signN\in \{+,-,0\}^3$. If we use $\pm$ in subscripts, the repeated occurences are unrelated, such that $\mathcal M_{\pm\pm\pm}=\{\bpt x\in \mathcal M:\,\text{all three $N_i\neq 0$}\}$. We use the notation $\mathcal T_i = \{\bpt x\in\mathcal M:\, \langle\taubp_j \bpt \Sigma\rangle=\langle\taubp_k,\Sigma\rangle,\,N_j=N_k\}$ for the Taub-spaces, where $i,j,k$ are a permutation of $\{1,2,3\}$. 

We frequently use the following subsets of $\mathcal M$ (with the obvious definition for subscripts $\signN\in \{+,-,0\}^3$ ):
\[\begin{aligned}
\textsc{Basin}[\epsilon]&=\{\bpt x\in\mathcal M: \max_i \frac{\delta_i}{r_i}<\epsilon, \max_i \delta_i < \epsilon\}\\
{\textsc{Cap}}[\epsilon_N, \epsilon_d]&=\{\bpt x\in \mathcal M: \max|N_i|\ge \epsilon_N, \max_i\delta_i \le \epsilon_d\}\\
{\textsc{Circle}}[\epsilon_N, \epsilon_d]&= \{\bpt x\in \mathcal M: \max |N_i|\le\epsilon_N,\, \max_i\delta_i \le \epsilon_d\}\\
{\textsc{Hyp}}[\varepsilon_\taubp, \epsilon_N, \epsilon_d]&= {\textsc{Circle}}[\epsilon_N, \epsilon_d] \setminus \left[B_{\varepsilon_\taubp}(\taubp_1)\cup B_{\varepsilon_\taubp}(\taubp_2)\cup B_{\varepsilon_\taubp}(\taubp_3) \right].
\end{aligned}.
\]

%
%
%
%
%
\subsection{Properties of the Kasner Map}\label{sect:appendix-kasner-map}
We deferred a detailed discussion of the 
Kasner-map $K$ in Section \ref{sect:lower-b-types}, especially the proof of Proposition \ref{prop:kasnermap-homeomorphism-class}.
We will first give a simple proof of Proposition \ref{prop:kasnermap-homeomorphism-class}, and then discuss classical ways of describing the Kasner map.
\begin{figure}[hbpt]
        \begin{subfigure}[b]{0.45\textwidth}
        \centering
        \includegraphics[width=\textwidth]{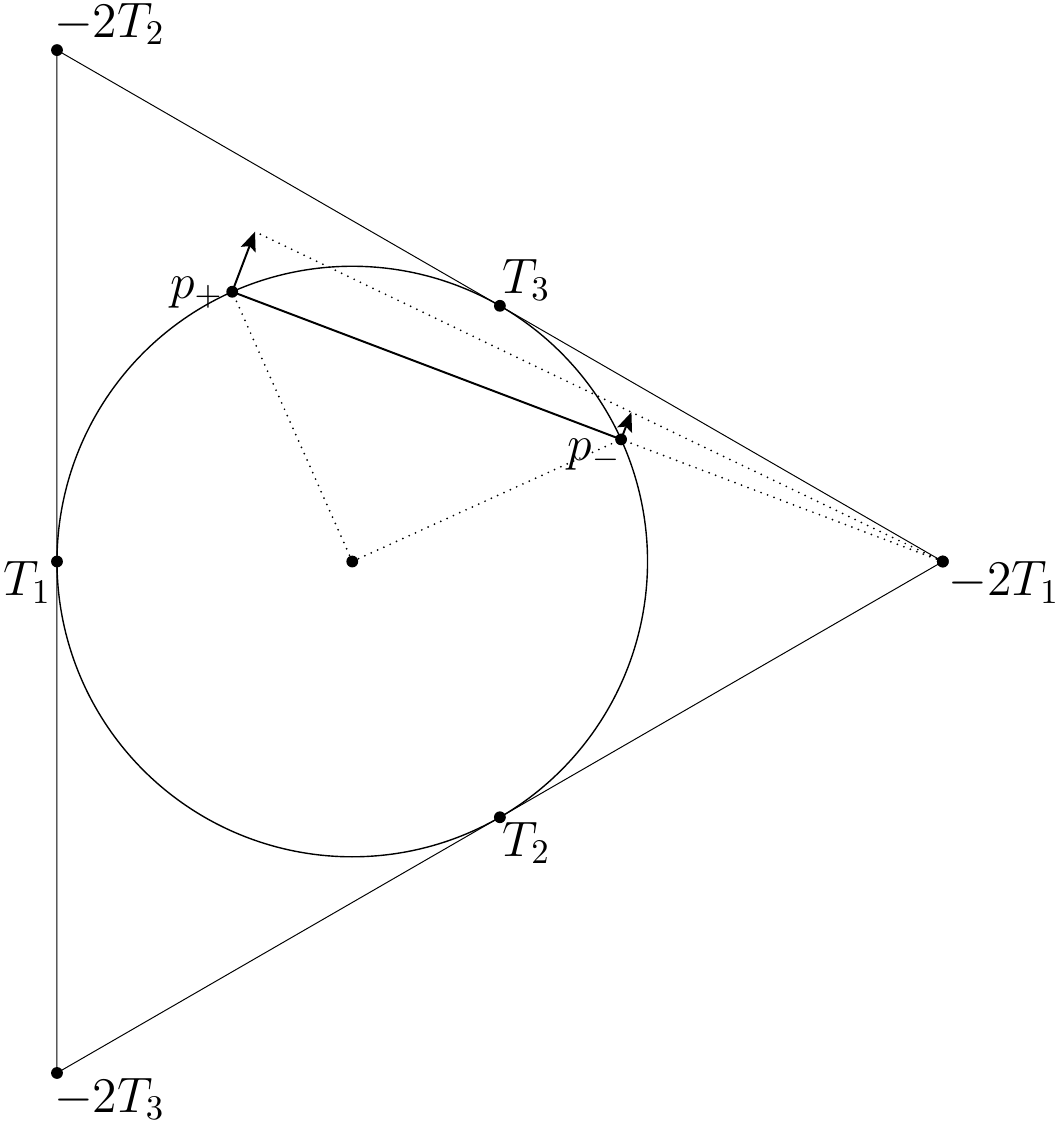}
        \caption{A graphical proof of Lemma \ref{lemma:kasnermap-expansion}.} \label{fig:angle-stuff}
        \end{subfigure}
        ~~%
        \begin{subfigure}[b]{0.45\textwidth}
                \centering
                \includegraphics[width=\textwidth]{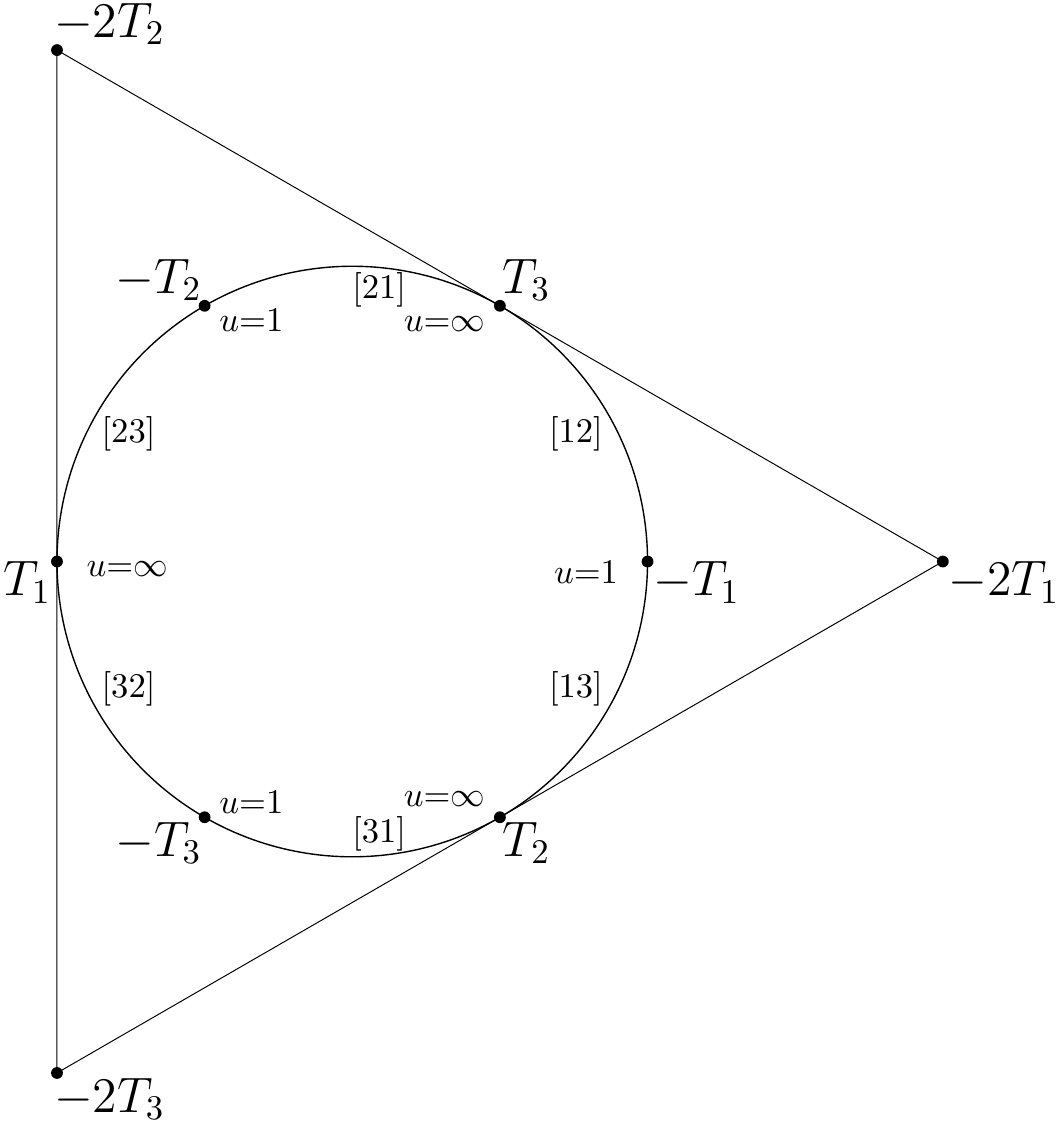}
                 \caption{The Kasner-parameter in the six segments.}
                \label{fig:seg-labels}
        \end{subfigure}
        \caption{Expansion of the Kasner-map}        
\end{figure}
\paragraph{Expansion of the Kasner-map.}
We can see from Figure \ref{fig:double-cover} (p.~\pageref{fig:double-cover}) that the Kasner-map is a double-cover, has three fixed points and reverses orientation. From Figure \ref{fig:angle-stuff}, we can see that $K$ is non-uniformly expanding:

\begin{lemma}\label{lemma:kasnermap-expansion}
Consider the vectorfield $\partial_{\mathcal K}(\bpt p)= (-\Sigma_-(\bpt p), +\Sigma_+(\bpt p))$. Assume without loss of generality that $\bpt p_-$ is such that the Kasner-map $\bpt p_+=K(\bpt p_-)$ proceeds via the $N_1$-cap, i.e.~$d(\bpt p_-, -\taubp_1)<1$ (see Figure \ref{fig:kasner-segments}).

Then the Kasner-map $K$ is differentiable at $\bpt p_-$ and we have
\[
K'(\bpt p_-)= -\frac{|\bpt p_+ +2\taubp_1| }{|\bpt p_- +2\taubp_1|}<-1,\qquad\text{where}\qquad K'(\bpt p'):\quad K_*\partial_{\mathcal K}(\bpt p_-)=K'(\bpt p_-)\partial_{\mathcal K}(\bpt p_+).
\]
\end{lemma}

\begin{proof}[Proof of Lemma \ref{lemma:kasnermap-expansion}]
Informally, differentiability is evident from the construction in Figure \ref{fig:angle-stuff}. The component of $\partial_{\mathcal K}$ which is normal to the line through $\bpt p_+$, $\bpt p_-$ and $-\taubp_1$ gets just elongated by a factor $\lambda =\frac{|\bpt p_++2\taubp_1|}{|\bpt p_-+2\taubp_1|}$. The angle between this line and the Kasner-circle, i.e.~$\partial_{\mathcal K}$, is constant; therefore, the length of the component tangent to $\mathcal K$ must get elongated by the same factor. The negative sign is evident from Figure \ref{fig:angle-stuff}.

Formally, the relation between $\bpt p_+$ and $\bpt p_-$ is described by
\[
|\bpt p_-|=|\bpt p_+|=1\qquad \bpt p_+ +2 \taubp_1= \lambda(\bpt p_-+2\taubp_1).
\]
Setting $\bpt p_-=\bpt p_-(t)$, we obtain a differentiable function $\lambda=\lambda(t)$ by the implicit function theorem (as long as $\langle\bpt p_+, \bpt p_-+2\taubp_1\rangle \neq 0$) and obtain (where we use $'$ to denote derivatives with respect to $t$):
\[
\bpt p_+' = \lambda'(\bpt p_-+2\taubp_1)+\lambda\bpt p_-'.
\]
Assuming $\bpt p_+\neq \bpt p_-$, we can set
\[\bpt v = \frac{\bpt p_++2\taubp_1}{|\bpt p_++2\taubp_1|}= \frac{\bpt p_-+2\taubp_1}{|\bpt p_-+2\taubp_1|}=\frac{\bpt p_+-\bpt p_-}{|\bpt p_+-\bpt p_-|},\]
and compute the projection to the normal component of $\bpt v$:
\((1-\bpt v \bpt v^T)\bpt p_+' = \lambda (1-\bpt v \bpt v^T)\bpt p_-',\) 
since the vector coefficient of $\lambda'$ is parallel to $\bpt v$. Now the vectors $\bpt p_+'$ and $\bpt p_-'$ are tangent to the Kasner-circle; letting $J(\Sigma_+,\Sigma_-)=(-\Sigma_-, \Sigma_+)$ be the unit rotation we can see that
\( \bpt p_+' = \pm |\bpt p_+'|J\bpt p_+ \) and \( \bpt p_-' = \pm |\bpt p_-'|J\bpt p_- \). 
Hence
\[\begin{aligned}
|(1-\bpt v \bpt v^T)\bpt p_+'|^2&=\left(1-\langle v, J\bpt p_+\rangle^2\right)|\bpt p_+'|^2 
&= \left(1-\frac{\langle\bpt p_-, J\bpt p_+\rangle^2}{|\bpt p_+-\bpt p_-|^2}\right)|\bpt p_+'|^2\\
|(1-\bpt v \bpt v^T)\bpt p_-'|^2&=\left(1-\langle v, J\bpt p_-\rangle^2\right)|\bpt p_-'|^2 
&= \left(1-\frac{\langle\bpt p_+, J\bpt p_-\rangle^2}{|\bpt p_+-\bpt p_-|^2}\right)|\bpt p_-'|^2.
\end{aligned}\]
By antisymmetry of the matrix $J$, i.e.~$\langle\bpt p_-, J\bpt p_+\rangle=-\langle J\bpt p_-,\bpt p_+\rangle$, we therefore have
\(
|\bpt p_+'| = \lambda |\bpt p_-'|.
\)
\end{proof}
\paragraph{Symbolic Description.}
For a given $\bpt p_0\in\mathcal K$, we can symbolically encode the trajectory $(\bpt p_n)_{n\in\NN}$ (with $\bpt p_{n+1}=K(\bpt p_n)$) under the Kasner-map. The easiest way to do so is to encode it by $(s_n)_{n\in\NN}\in \{1,2,3\}^\NN$, where $s_n = i$ if $\bpt p_n\to\bpt p_{n+1}$ occurs via the $|N_i|>0$-cap. Then $(s_n)_{n\in\NN}$ has the property that no symbol repeats, i.e.~$s_{n+1}\neq s_n$ for all $n\in\NN$. We have, however, an ambiguity if $\bpt p_N=\taubp_i$ for some $N>0$. If this occurs, then also all later points have $\bpt p_{N+n}=\taubp_i$. We chose to allow both encodings $\bpt p_N=j$ and $\bpt p_N=k$, as long as the property that no symbol repeats is preserved. Factoring out this ambiguity gives us a map
\[\begin{aligned}
\Psi: \mathcal K &\to \{(s_{n})_{n\in\NN}\{1,2,3\}^\NN:\, \text{no sumbol repeats} \}\, / \, \{(*\overline{ij}) = (*\overline{ji})\},\\
\Psi(p_0)&= (s_n)_{n\in\NN},\qquad\text{such that}\  d(\bpt p_n, -\taubp_{s_n})\le 1\ \text{and no symbol repeats},
\end{aligned}\]
where $*$ stands for an arbitrary initial piece and $\overline{jk}$ stands for a periodic tail $(*\overline{jk})=(*jkjkjk\ldots)$.
This map $\Psi$ is continuous (since the Kasner-map is continuous), where we endow the target space $\{1,2,3\}^\NN/\sim$ with the quotient topology of the product topology. Note that, by construction, $\Psi$ semiconjugates the Kasner-map $K$ to the sifht-map $\sigma$:
\[
\Psi\circ K=\sigma\circ \Psi,\quad\text{where}\quad \sigma: (s_0s_1s_2\ldots)\to (s_1s_2\ldots).
\]

In order to see that $\Psi$ is a homeomorphism, we construct a continuous inverse. 
Denote the three segments of $\mathcal K$ as $\mathcal K_i = \{\bpt p\in\mathcal K:\, d(\bpt p, -\taubp_i)\le 1\}$. We can construct inverse maps $K^{-1}_{ij}:\mathcal K_j \to \mathcal K_i$, such that $K\circ K^{-1}_{ij}:\mathcal K_i \to \mathcal K_i = \mathrm{id}$. Then we get an inverse map
\[
\Psi^{-1}: (s_{n})_{n\in\NN} \to \bigcap_{\ell\in\NN}K^{-1}_{s_0 s_1}K^{-1}_{s_1s_2}\ldots K^{-1}_{s_{\ell}s_{\ell+1}}(\mathcal K_{s_{\ell+1}}).
\]
We now need to show that $\Psi^{-1}((s_n)_{n\in\NN})=\{\bpt p\}$ is a single point, which depends continuously on $(s_n)_{n\in\NN}$, and is actually the inverse of $\Psi$.

We first consider a sequence $(s_n)_{n\in\NN}$ which does not end up in a Taub-point. 

In order to see that $\Psi^{-1}((s_n)_{n\in\NN})$ is nonempty, note that it is the intersection of a descending sequence of nonempty compact sets. In order to see that it contains only a single point, note that $K_{ij}^{-1}$ is (nonuniformly) contracting by Lemma \ref{lemma:kasnermap-expansion}; hence the length $\textsc{len}_{\ell}=|K^{-1}_{s_0 s_1}K^{-1}_{s_1s_2}\ldots K^{-1}_{s_{\ell}s_{\ell+1}}(\mathcal K_{s_{\ell+1}})|$ is decreasing. The length $\textsc{len}_{\ell}$ also cannot converge to some $\textsc{len}_{\infty}$ as $\ell\to\infty$, since we have $|K^{-1}_{ij}(I)|< |I|$ for any interval $I$ with $|I|>0$.
In order to see that $\Psi^{-1}$ is continuous, we need to show that $\mathrm{diam}\,\Psi^{-1}\{(r_n)_{n\in\NN}:\,r_n=s_n\, \forall n\le N\} \to 0$ as $N\to\infty$. This also follows from the previous argument of decreasing lengths.

Next, we consider a sequence $(s_n)_{n\in\NN}$ which does end up in a Taub-point $\taubp_i$ at $n=N$. Let $(s'_n)_{n\in\NN}$ denote the other representative of $(s_n)_{n\in\NN}$, i.e. we changed $(s_0,\ldots,s_{N-1},\overline{jk})\leftrightarrow ((s_0,\ldots,s_{N-1},\overline{kj}))$. The previous arguments about single-valuedness and continuity still apply for each of the two representatives; we only need to show that $\Psi^{-1}$ coincides for both. This is obvious.

It is also obvious by our construction that $\Psi$ and $\Psi^{-1}$ are inverse to each other. Hence $\Psi$ and $\Psi^{-1}$ $C^0$-conjugate $K$ to the shift-map
\[
\Psi\circ K\circ\Psi^{-1}=\sigma: (s_0s_1s_2\ldots)\to (s_1s_2\ldots).
\]
The exact same arguments apply in order to conjugate the map $D:\RR/3\ZZ\to\RR/3\ZZ$, $D:[x]_{3\ZZ}\to [-2x]_{3\ZZ}$ to the same shift, where we replaced $\taubp_i=i$ and $\mathcal K_i = [{i+1}, {i+2}]$. Hence, the Kasner-map is $C^0$ conjugate to $D$:
\begin{hproposition}{\ref{prop:kasnermap-homeomorphism-class}}
There exists a homeomorphism $\psi:\mathcal K \to \RR/3\ZZ$, such that $\psi(\taubp_i)=\left[i\right]_{3\ZZ}$ and
\[
\psi(K(\bpt p))=[-2 \psi(\bpt p)]_{3\ZZ}\qquad\forall \bpt p\in\mathcal K.
\]
\end{hproposition}
\paragraph{Kasner Eras and Epochs.}
\def\epoch{\ensuremath{\textsc{S}}}
\def\era{\ensuremath{\textsc{L}}}
A useful and customary description of the symbolic dynamics of $K$ is obtained by distinguishing between ``small'' bounces around a Taub-point, called ``Kasner epochs'' and denoted by the letter \epoch{} in this work, and ``long'' bounces, called ``Kasner eras'' and denoted by the letter \era{} in this work. The \epoch{} and \era{} encoding of an orbit can be obtained from previous $\{1,2,3\}$-encoding by the map
\[\textsc{EpochEra}: \{s\in \{1,2,3\}^\NN:\, s_n\neq s_{n+1}\}\to \{(s_0s_1|r_0r_1\ldots):\, s_0,s_1 \in\{1,2,3\},\,r_n\in\{\epoch, \era\}\}\]
\[(s_0s_1s_2\ldots) \to (s_0s_1| r_0r_1\ldots),\quad\text{where}\quad\left\{\begin{aligned}
r_n &= \epoch \quad\text{if $s_{n}=s_{n+2}$,}\\
r_n &= \era \quad\text{if $s_{n}\neq s_{n+2}$.}
\end{aligned}\right.\]
We can remember the value of $s_n$ as a subindex of $r_n$, such that e.g.
\[
(132121213231\ldots)\to(13|\era_1\era_3 \era_2 \epoch_1 \epoch_2 \epoch_1 \epoch_2 \era_1 \era_3 \epoch_2 \era_3 *_1\ldots).
\]
Then the Kasner-map becomes
\[ (s_0s_1|r_0r_1r_2r_3\ldots)\to \left\{\begin{aligned}
(s_1 i|r_1r_2\ldots) &\quad \text{if $s_0=i$ and $r_0=\epoch$}\\
(s_1 i|r_1r_2\ldots) &\quad \text{if $s_0\neq i\neq s_1$ and $r_0=\era$.}
\end{aligned}\right.\]
Note that the first two indices in $\{1,2,3\}$ describe in which of the six symmetric segments of $\mathcal K$ a point lies, see Figure \ref{fig:seg-labels}.

Another customary way of writing such sequences is to write every symbol \era{} as a semicolon ``;'' and abbreviate the \epoch{} symbols in between by just their number, such that the previous example becomes
\[
(13|\era_1 \era_3 \era_2 \epoch_1 \epoch_2 \epoch_1 \epoch_2 \era_1 \era_3 \epoch_2 \era_3 *_1\ldots)\to 
(13|0;0; 0; 4 ; 0 ; 1 ; \ldots).
\]
\paragraph{The Kasner-Parameter.}
There exists an explicit coordinate transformation, related to the continued fraction expansion, which realizes the conjugacy to the shift-space. This is done via the so-called \emph{Kasner-parameter} $u$, and is the most standard way of discussing the Kasner-map. 

In order to introduce the Kasner-parameter, it is useful to use the coordinates which make the permutation symmetry of the indices more apparent. This is done via
\[\Sigma_i = 2\langle \taubp_i, \bpt \Sigma\rangle \qquad
\Sigma_+ = -\frac 1 2 \Sigma_1 \qquad
\Sigma_- = \frac{1}{2\sqrt{3}}(\Sigma_3-\Sigma_2).\]
These variables are constrained by $\Sigma_1+\Sigma_2+\Sigma_3=0$ and have $\Sigma_1^2+\Sigma_2^2+\Sigma_3^2 =6(\Sigma_+^2+\Sigma_-^2)$.

The parametrization of $\mathcal K$ depends not only on a real parameter $u\in\RR$, but also on a permutation $(i,j,k)$ of $\{1,2,3\}$ and is given by $\bpt \Sigma=\bpt\Sigma(u, (ijk))$ such that
\[\begin{aligned}
\Sigma_i &= -1 + 3\frac{-u}{u^2+u+1} &
\Sigma_j &= -1 + 3\frac{u+1}{u^2+u+1} &
\Sigma_k &= -1 + 3\frac{u^2+u}{u^2+u+1}.
\end{aligned}\]
We can immediately observe that $\langle \taubp_1+\taubp_2+\taubp_3, \bpt \Sigma\rangle=0$, as it should be; hence, the above really defines a function $\psi:\RR\times \textsc{sym}_3\to \RR^2$, where $\textsc{sym}_3$ is the set of permutations of $\{1,2,3\}$. 
 Direct calculation shows that $\Sigma_1^2+\Sigma_2^2+\Sigma_3^2=6$ for all $u\in\RR$.
Hence, the $u$ coordinates do actually parametrize the Kasner circle. We have the noteworthy symmetry properties
\[\begin{array}{rlll}
\bpt \Sigma\left(u, (ijk)\right)
&=\bpt \Sigma\left(u^{-1}, (ikj)\right) 
&=\bpt \Sigma\left(-(u+1), (jik)\right)
&= \bpt\Sigma\left(\frac{-1}{u+1}, (jki)\right)\\
&= \bpt\Sigma\left(-\frac{u}{u+1}, (kji)\right)
&= \bpt\Sigma\left(-\frac{u+1}{u}, (kij)\right).&
\end{array}\]
We can use the symmetry to normalize $u$ to a value $u\in[1,\infty]$, thus giving a parametrization of the Kasner circle as in Figure \ref{fig:seg-labels}.

\paragraph{The Kasner-map in $u$-coordinates.}
We consider without loss of generality a heteroclinic orbit $\gamma\subseteq \mathcal M_{+00}$. In the $\bpt \Sigma$-projection, this heteroclinic orbit is a straight line through $-2\taubp_1$; hence, the quotient $\frac{\Sigma_-}{2-\Sigma_+}$ stays constant. It is given by 
\[
\frac{\Sigma_-}{2-\Sigma_+}(u)=\sqrt{3}\frac{u^2-1}{4(u^2+u+1)-(u^2+4u+1)} = \frac{\sqrt{3}}{3}\frac{u^2-1}{u^2+1}.
\]
If we write the $\alpha$-limit of $\gamma$ with respect to the $123$ permutation, then we must have $u\in [0,\infty]$ (because $N_1$ would not be unstable otherwise). Then the Kasner-map must be given by $(u,(123))\to (-u,123)$. Assume that $u\in[2,\infty]$; then we can renormalize $K(u)$, such that $K(u, 123)= (u-1, 213)$. If instead $u\in [1,2]$, then we can renormalize such that $K(u,123)=\left(\frac{1}{u-1},231\right)$. Applying symmetrical arguments for the other caps yields the following way of writing the Kasner map:
\[\begin{gathered}
\mathcal K = \textsc{sym}_3\times [1,\infty] \,\big/\, \left\{ (1, ijk)=(1,ikj), (\infty, ijk)=(\infty, jik) \right\}\\
K:\mathcal K\to \mathcal K\quad  (u, ijk)\to 
\left\{\begin{aligned}
(u-1, jik) &\qquad \text{if $u\in [2,\infty]$}\\
\left(\frac{1}{u-1}, jki\right)&\qquad \text{if $u\in [1,2]$}.
\end{aligned}\right.
\end{gathered}\]
Note that the Kasner-map is actually well-defined and continuous at $u=2$, due to the identification $(1,ijk)=(1,ikj)$. It is also well-defined at $u=1$, since the identifications $(1,ijk)=(1,ikj)$ and $(\infty, ijk)=(\infty, jik)$ are compatible. It is continuous at $u=1$, since we use the usual compactification at $u=\infty$, such that a neighborhood basis of $(\infty, ijk)$ is given by $\{([R,\infty], ijk)\cup ([R,\infty], jik)\}_{R \gg 0}$. Likewise, the Kasner-map is well-defined and continuous at the three fixed-points $u=\infty$, due to the identification $(\infty, ijk)=(\infty, jik)$ and the compactification at $u=\infty$.
\paragraph{Symbolic description in $u$-coordinates.}
The two local inverses of the Kasner-map in $u$-coordinates are given by
\[\begin{aligned}
\epoch: (u,ijk)&\to (u+1, jik)\\
\era: (u,ijk) &\to \left(1+\frac{1}{u}, kij\right)
\end{aligned}\]
Using the same construction as in the proof of Proposition \ref{prop:kasnermap-homeomorphism-class}, we see that the inverse coding map $\textsc{Cfe}^{-1}$ given by the continued fraction expansion
\[
\textsc{Cfe}^{-1}:(ij|a_0;a_1;a_2;\ldots)\to \left(1+a_0 + \frac{1}{1+a_1+\frac{1}{1+a_2+\ldots}}\,,\, ijk\right),
\]
and the Kasner-map is given by
\[
\textsc{Cfe}\circ K\circ \textsc{Cfe}^{-1}: (ij|a_0;a_1;a_2;\ldots)\to\left\{\begin{array}{ll}
(ji\,|\,a_0-1;a_1;a_2;\ldots) &\quad\text{if $a_0>0$}\\
(jk\,|\,a_1;a_2;\ldots) &\quad\text{if $a_0=0$}.
\end{array}\right.
\]

\subsection{Transformation of Volumes on Manifolds}\label{sect:general-vol-app}
This section contains some basic facts about the transformation of volumes under flows.

\paragraph{Volume Transformation.}
Let $M$ be an $n$-dimensional differentiable manifold.
In the language of differential forms, we can write the transformation law for a non-vanishing volume-form $\omega$ under a diffeomorphism $\Phi:M\to M$ as
\[\begin{aligned}
\lambda(\bpt x) = \frac{\Phi^*\omega(\bpt x)}{\omega(\bpt x)}\\
\mathrm{vol}_\omega(\Phi(U)) = \int_{\Phi(U)}\omega = \int_U \Phi^*\omega = \int_U \lambda \omega.
\end{aligned}\]
Let $X_1,\ldots, X_n$ be a frame, i.e.~a set of vectorfields which form a basis of $TM$. Then can write the density of $\omega$ as $\rho=\omega[X_1,\ldots, X_n]$ and see $\omega[Y_1,\ldots, Y_n]= \rho \det (a_{ij})$, where $Y_i = \sum_j a_{ij}X_j$. This allows us to write 
\[
\lambda(\bpt x) \rho(\bpt x) = \Phi^*\omega [X_1,\ldots X_n] = \omega[\Phi_*X_1,\ldots \Phi_*X_n] = \rho(\Phi(\bpt x))\det J,
\]
where $J$ is the jacobian, i.e.~the matrix of $\DD_x \Phi(\bpt x):T_{\bpt x}M\to T_{\Phi(\bpt x)}M$ with respect to the basis $X_1(\bpt x),\ldots X_n(\bpt x)$ and $X_1(\Phi(\bpt x)),\ldots X_n(\Phi(\bpt x))$, i.e. $J=J_{ij}$ with $\Phi_*X_i(\Phi(\bpt x)) = \DD_x \Phi(\bpt x)\cdot X_i(\bpt x) = \sum_{j}J_{ij} X_j(\bpt x)$.

\paragraph{Volume Transformation under flows.}
We study the behaviour of a volume-form $\omega$ under a flow $\phi:M\times \RR\to M$ corresponding to a vectorfield $f:M\to TM$; we are interested in $\lambda(\bpt x, t)=\frac{\phi^*(\bpt x, t)\omega}{\omega}$.

Given a set of coordinates $x_1,\ldots x_n$ and the vectorfield $f=f_1 \partial_1+\ldots f_n\partial_n$ and volume-form $\omega=\rho \dd x_1\land \ldots \land \dd x_n$ we can therefore write, using $\DD_x\phi$ is a shorthand for the Jacobian with respect to the basis $\partial_1,\ldots, \partial_n$:
\[\begin{aligned}
\lambda(\bpt x, t)&=\frac{\phi^*(\bpt x, t)\omega}{\omega} = \frac{\rho(\phi(\bpt x, t))}{\rho(\bpt x)}\det \DD_x \phi(\bpt x, t)\\
\frac{\dd}{\dd t} \lambda(\bpt x, t)& = \frac{(\DD_f \rho)(\phi(\bpt x, t))}{\rho(\bpt x)}\det \DD_x \phi(\bpt x, t) + \frac{\rho(\phi(\bpt x, t))}{\rho(\bpt x)}\DD_t \det \DD_x \phi(\bpt x, t)\\
&= \lambda(\bpt x, t)\DD_f \log \rho(\phi(\bpt x, t))+ \frac{\rho(\phi(\bpt x, t))}{\rho(\bpt x)}\mathrm{tr}\left[(\DD_t\DD_x\phi(\bpt x, t))\DD_x\phi(\bpt x, t)^{-1}\right] \det \DD_x \phi(\bpt x, t)\\
&= \lambda(\bpt x, t)\left(\DD_f \log \rho(\phi(\bpt x, t)) + \mathrm{tr}\DD_x f(\phi(\bpt x, t))\DD_x\phi(\bpt x, t)^{-1}\right)\\
&= \lambda(\bpt x, t)\left(\DD_f \log \rho(\phi(\bpt x, t)) + \mathrm{tr}\partial_x f(\phi(\bpt x, t))\right)\\
&= \lambda(\bpt x, t)\left(\DD_f \log \rho(\phi(\bpt x, t)) + \sum_i (\partial_i f_i)(\phi(\bpt x, t))\right),
\end{aligned}\]
where we used the general formula for differentiable families $A:\RR\to \RR^{n\times n}$ of invertible matrices:
\[\begin{aligned}
\DD_{t|t=0} \det A(t) &= \DD_{t|t=0}\det A(t)A(0)^{-1}A(0)\\ 
&=\det A(0) \DD_{t|t=0}\det A(t)A(0)^{-1} = \det A(0)\mathrm{tr}\,\DD_{t|t=0}A(t)A(0)^{-1}.
\end{aligned}\]

\paragraph{Restriction to Iso-Surfaces.}
Let $\Phi:M\to M$ be a diffeomorphism on an $n$-dimensional manifold $M$ with volume form $\omega$, which gets transported by $\lambda= \frac{\Phi^*\omega}{\omega}$.
Suppose that $G:M\to \RR$ is a preserved quantity under $\Phi$, i.e.~$G(\Phi(\bpt x))=G(\bpt x)$ and suppose that $1$ is a regular value of $G$. Now we are interested in the isosurface $M_1 = \{\bpt x\in M: G(\bpt x) =1\}$ and in a volume form $\omega_0$ on $M_1$, which gets transported by the same $\lambda$. This can be realized by choosing some vectorfield $X:M_1\to TM$ with $\DD_X G=\mathrm{const}=1$ and setting $\omega_1=\iota_X \omega$. Since $G$ is preserved, we have $\Phi_*X-X \in T M_1$ and hence for a basis $X_1,\ldots,X_{n-1}$ of $TM_1$:
\[\begin{aligned}
\lambda \iota_X \omega[X_1,\ldots X_n] &= \lambda\omega[X,X_1,\ldots,X_n]=\omega[\Phi_*X,\Phi_*X_1,\ldots,\Phi_*X_{n-1}]\\
&=\omega[X,\Phi_*X_1,\ldots,\Phi_*X_{n-1}]+\omega[\Phi_*X-X,\Phi_*X_1,\ldots,\Phi_*X_{n-1}]\\
&=\Phi^*\omega[X_1,\ldots,X_{n-1}].
\end{aligned}\]
It is clear that the induced volume $\omega_1$ does not depend on the choice of $X$, up to possibly a sign.
\paragraph{Contraction with invariant vectorfields.}
Let $\Phi:M\to M$ be a diffeomorphism on an $n$-dimensional manifold $M$ with volume form $\omega$, which gets transported by $\lambda= \frac{\Phi^*\omega}{\omega}$.
Suppose that the vectorfield $Y:M\to TM$ is preserved quantity under $\Phi$, i.e.~$\Phi_*X=X$. Then the contractoed volume-form $\omega_1=\iota_Y\omega$ has $\Phi^*\omega_1=\lambda\omega_1$:
\[\begin{aligned}
\Phi^*\omega_1[X_2,\ldots,X_n]&=\omega[Y, \Phi_*X_2,\ldots,\Phi_*X_n]=\omega[\Phi_*Y,\Phi_*X_2,\ldots,\Phi_*X_n]\\
&=\lambda \omega[Y,X_2,\ldots,X_n]=\lambda\omega_1[X_2,\ldots,X_n].
\end{aligned}\]
Suppose that $\phi:M\times \RR\to M$ is a flow, $U\subseteq M$ is open and $T:U\to \RR$ is smooth. Consider $\Phi:\bpt x\to \phi(\bpt x, T(\bpt x))$. Then $\Phi_* X = \phi(\cdot, T(\cdot))_* X + Y\DD_X T$, and hence
\[\begin{aligned}
\Phi^*\omega_1[X_2,\ldots,X_n]=\omega[Y, \Phi_*X_2+ Y\DD_{X_2} T,\ldots,\Phi_*X_n+Y\DD_{X_n} T]=\lambda \omega_1[X_2,\ldots,X_n].
\end{aligned}\]
\subsection{Derivation of the Wainwright-Hsu equations}\label{sect:append:derive-eq}
The goal of this section is to connect the Einstein field equations of general relativity to the Wainwright-Hsu equations \eqref{eq:ode} discussed in this work.

\subsubsection{Spatially Homogeneous Spacetimes}\label{sect:homogeneous}
We are interested in spatially homogeneous spacetimes. We assume that $(M^4,g)$ is a Lorentz-manifold, and we have a symmetry adapted co-frame: $\{\omega_1,\omega_2,\omega_3, \dd t\}$, corresponding to a frame of vectorfields $\{e_0,e_1,e_2,e_3\}$, such that $e_1,e_2,e_3$ are Killing, i.e.~the metric depends only on $t$. We thus assume that the metric has the form
\[
g = g_{00}(t)\dd t\otimes \dd t + g_{11}(t)\omega_1\otimes \omega_1 + g_{22}(t)\omega_2\otimes \omega_2 + g_{33}(t)\omega_3\otimes \omega_3,
\]
where $g_{00}<0$ and the other three $g_{ii}>0$. The spatial homogeneity is described by the commutators of the three Killing fields $e_1,e_2,e_3$; we assume that it is given (for positive permutations $(i,j,k)$ of $(1,2,3)$) by:
\[
[e_i, e_j] = \gamma_{ij}^k e_k=\hat n_k e_k
\qquad
\dd \omega_i = -\hat n_i \omega_j\land \omega_k,
\]
where $\hat n_i \in \{-1, 0, +1\}$ describe the Bianchi type of the surfaces $\{t=\textrm{const}\}$ of spatial homogeneity.

\paragraph{General equations for the Christoffel symbols.}
The general equations for Christoffel symbols are given by:
\[\begin{aligned}
\nabla_{e_i}e_j &= \sum_k \Gamma_{ij}^k e_k\\
[e_i, e_j] &= \nabla_{e_i}e_j - \nabla_{e_j}e_i\quad\Rightarrow\quad
\Gamma_{ij}^k -\Gamma_{ji}^k = \gamma_{ij}^k\\
\DD_{e_i} g(e_j, e_k) &= \partial_{e_i}g_{jk} 
=\sum_{\ell} g_{\ell k}\Gamma_{ij}^{\ell} + g_{\ell j}\Gamma_{ik}^\ell\\
\partial_{e_i}g_{jk} + \partial_{e_j}g_{ik} - \partial_{e_k}g_{ij} 
&=\sum_{\ell} g_{\ell k} \Gamma_{ij}^\ell + g_{\ell j}\Gamma_{ik}^\ell 
+g_{\ell k} \Gamma_{ji}^\ell + g_{\ell i}\Gamma_{jk}^\ell 
- g_{\ell i} \Gamma_{kj}^\ell - g_{\ell j}\Gamma_{ki}^\ell\\
&=\sum_{\ell} g_{\ell i}\gamma^{\ell}_{jk} + g_{\ell j}\gamma_{ik}^\ell + g_{\ell k}\gamma_{ji}^\ell + 2 g_{\ell k}\Gamma_{ij}^\ell
\end{aligned}\]
We can solve this for the Christoffel symbols by multiplying with the inverse metric:
\newcommand\myrebind[6]{
\gdef\iA{#1}\gdef\iB{#2}\gdef\iC{#3}\gdef\iD{#4}\gdef\iE{#5}\gdef\iF{#6}
}
\myrebind{j}{i}{k}{\ell}{n}{*}
\[\begin{aligned}
\Gamma_{ij}^k&= 
\frac 1 2 g^{k \ell} \left(
\partial_{e_i}g_{j\ell} + \partial_{e_j}g_{i\ell} - \partial_{e_\ell}g_{ij} -\sum_{n} g_{n i}\gamma^{n}_{j\ell} - g_{n j}\gamma_{i\ell}^n - g_{n \ell}\gamma_{ji}^n\right)\\
&=\frac 1 2 g^{k k} \left(
\partial_{e_i}g_{jk} + \partial_{e_j}g_{ik} - \partial_{e_k}g_{ij} - g_{i i}\gamma^{i}_{jk} - g_{j j}\gamma_{ik}^j - g_{k k}\gamma_{ji}^k\right),
\end{aligned}\]
where we used the fact that the metric is diagonal in the last equation.
\paragraph{Christoffel symbols for spatially homogeneous space times.}
If we insert the indices into this equation, we obtain up to index permutations the following non-vanishing Christoffel symbols:
\[\begin{aligned}
\nabla_{e_0}e_0 &= \Gamma_{00}^0 e_0 & \Gamma_{00}^0 &= \phantom{-}\frac 1 2 g^{00}\partial_{e_0} g_{00}\\
\nabla_{e_1}e_2 &= \Gamma_{12}^3 e_3 & \Gamma_{12}^3 &= -\frac 1 2 g^{33}\left(\;g_{11}\gamma^{1}_{23}+ g_{22}\gamma_{13}^2 + g_{33}\gamma_{21}^3\right)\\
& & &=\phantom{-}\frac 1 2 g^{33}\left(g_{11}\hat n_1 \;- g_{22}\hat n_2 \;- g_{33}\hat n_3\right)\\
\nabla_{e_0}e_1 &= \Gamma_{01}^1 e_1 = \Gamma_{10}^1 e_1 & \Gamma_{10}^1 &=\phantom{-}\frac 1 2 g^{11}\partial_{e_0}g_{11}\\
\nabla_{e_1}e_1 &= \Gamma_{11}^0 e_0 & \Gamma_{11}^0 &= -\frac 1 2 g^{00}\partial_{e_0}g_{11}
\end{aligned}\]
\paragraph{Extrinsic Curvature of surfaces of homogeneity.} %
The Weingarten-map $K_{i}^{\ j}$ and second fundamental form $K_{ij}$ of the surfaces of spatial homogeneity given, up to permutation, by:
\[\begin{aligned} 
K(e_1)&=\nabla_{e_1}\frac{1}{\sqrt{-g_{00}}}e_0 =  \sqrt{-g^{00}}\Gamma_{10}^1 e_{1}
\qquad K_i^{ j}= \delta_{i}^{j}\sqrt{-g^{00}}\Gamma_{i0}^i\\
K_{ij}&= g(e_i, \nabla_{e_j}\sqrt{-g^{00}}e_0)= \sqrt{-g_{00}}\Gamma_{ji}^0 =\sum_{k} g_{k i}\sqrt{-g^{00}}\Gamma_{j0}^k,\quad{\text{i.e.,}}\\
\Gamma_{10}^1 &= \sqrt{-g_{00}} K_1^1, \qquad\qquad \Gamma_{11}^0 = \sqrt{-g^{00}}g_{11}K_{1}^1
\end{aligned}\]

The extrinsic curvature corresponds to the normalized time-derivative of the spatial coefficients of the metric:
\[\begin{aligned}
\sqrt{-g^{00}}\nabla_{e_0} \sqrt{g_{ii}} &= \sqrt{-g^{00}}\frac 1 2 \sqrt{g_{ii}}g^{jj}\nabla_{e_0}g_{kk}=\sqrt{-g^{00}}\Gamma_{i0}^i \sqrt{g_{ii}} = K_{i}^i.
\end{aligned}\]
\paragraph{Riemannian Curvature.}
The Riemannian curvature tensor is given by the equation
\[\begin{aligned}
 \sum_{\ell}R_{ijk}^{\ \ \ \ell} e_{\ell} &= \left(\nabla_{e_i}\nabla_{e_j}-\nabla_{e_j}\nabla_{e_i} -\nabla_{[e_i,e_j]}\right) e_k\\
 &=\sum_{\ell, n}\left(\Gamma_{jk}^n\Gamma_{in}^\ell - \Gamma_{ik}^n\Gamma_{jn}^\ell + \nabla_{e_i}\Gamma_{jk}^\ell - \nabla_{e_j}\Gamma_{in}^\ell - \gamma_{ij}^n\Gamma_{nk}^\ell\right)e_{\ell}.
\end{aligned}\]
If we lower the last index, we have
\(
R_{ijk\ell}=g((\nabla_{e_i}\nabla_{e_j}-\nabla_{e_j}\nabla_{e_i}-\nabla_{[e_i,e_j]}) e_k, e_\ell).
\) 
Together with $(\nabla_{e_i}\nabla_{e_j}-\nabla_{e_j}\nabla_{e_i}-\nabla_{[e_i,e_j]}) g(e_k,e_\ell)=0$, this makes apparent the anti-symmetries
\(
R_{ijk\ell}=-R_{jik\ell}=R_{ji\ell k}=-R_{ij\ell k}.
\)

Inserting the indices gives us the following potentially non-vanishing terms of the Riemann tensor, up to permutation and anti-symmetry, which are relevant for the Ricci curvature:
\[\begin{aligned}
R_{121}^{\ \ \ 2} &=  \Gamma_{21}^3\Gamma_{13}^2 - \Gamma_{11}^0\Gamma_{20}^2  - \gamma_{12}^3\Gamma_{31}^2\\
R_{010}^{\ \ \ 1} &=   \Gamma_{10}^1\Gamma_{01}^1 - \Gamma_{00}^0\Gamma_{10}^1 + \nabla_{e_0}\Gamma_{10}^1.
\end{aligned}\]
Raising and using $\widetilde R$ for the intrinsic curvature of the surfaces of homogeneity gives us
\[\begin{aligned}
R_{12}^{\ \ 1 2} &=  g^{11}\Gamma_{21}^3\Gamma_{13}^2- \gamma_{12}^3\Gamma_{31}^2g^{11} - g^{11}\Gamma_{11}^0\Gamma_{20}^2  \\
&= \widetilde{R}_{12}^{\ \ 1 2} -K_{1}^1 K_2^2\\
R_{01}^{\ \ 0 1} &=   g^{00}\Gamma_{10}^1\Gamma_{01}^1 - g^{00}\Gamma_{00}^0\Gamma_{10}^1 + g^{00}\nabla_{e_0}\Gamma_{10}^1\\
&= -K_1^1 K_1^1 + \sqrt{-g^{00}}\Gamma_{00}^0 K_1^1 +g^{00} \nabla_{e_0} \sqrt{-g_{00}}K_1^1\\
&=-K_1^1 K_1^1 -\sqrt{-g^{00}}\nabla_{e_0}K_1^1.
\end{aligned}\]
Setting
\[\begin{aligned}
n_i&=\hat n_i \sqrt{g^{11}g^{22}g^{33}}g_{ii} = \hat n_i \widetilde n_i,\quad\text{i.e.}\quad
\widetilde n_i = \sqrt{g^{11}g^{22}g^{33}}g_{ii},
\end{aligned}\]
the spatial curvature is given by
\[\begin{aligned}
\widetilde R_{12}^{\ \ 12} &=g^{11}\Gamma_{21}^3\Gamma_{13}^2- g^{11}\gamma_{12}^3\Gamma_{31}^2
= g^{11}\left(\Gamma_{21}^3\Gamma_{13}^2 - \Gamma_{12}^3\Gamma_{31}^2 + \Gamma_{21}^3\Gamma_{31}^2\right)\\
&=\frac 1 4\left(-n_1^2 -n_2^2 +3n_3^2 +2n_1n_2 -2n_2n_3 -2n_3n_1\right)\\
\widetilde R_{1}^{\ 1}&= \widetilde R_{12}^{\ \ 12}+ \widetilde R_{13}^{\ \ 13}
=\frac 1 2\left(-n_1^2 + n_2^2+n_3^2 -2 n_2n_3 \right)\\
\widetilde R &= \widetilde R_{1}^{\ 1} + \widetilde R_{2}^{\ 2} + \widetilde R_{3}^{\ 3}
= \frac 1 2(n_1^2+n_2^2+n_3^2 -2(n_1n_2+n_2n_3+n_3n_1)).
\end{aligned}\]
\subsubsection{The Einstein Field equation}\label{sect:gr-efe}
The Einstein equations in vacuum state that the space-time is Ricci-flat, i.e.~
\[\begin{aligned}
R_{0}^{\ 0} &= R_{01}^{\ \ 01}+R_{02}^{\ \ 02}+R_{02}^{\ \ 02} =0\\
R_{1}^{\ 1} &= R_{01}^{\ \ 01}+R_{12}^{\ \ 12}+R_{13}^{\ \ 13} =0\\
R_{2}^{\ 2} &= R_{02}^{\ \ 02}+R_{21}^{\ \ 21}+R_{23}^{\ \ 23} =0\\
R_{3}^{\ 3} &= R_{03}^{\ \ 03}+R_{31}^{\ \ 31}+R_{32}^{\ \ 32} =0.
\end{aligned}\]
Adding the last three equations and subtracting the first gives us an equation which does not contain time-derivatives of $K$ and hence is a constraint equation, called the ``Gauss constraint''. It is given by
\[\begin{aligned}
0 &= R_{12}^{\ \ 12} +R_{23}^{\ \ 23} + R_{31}^{\ \ 31} = \frac 1 2 \widetilde R - K_{1}^1K_2^2 - K_2^2K_3^3 - K_3^3K_1^1.
\end{aligned}\]
The evolution equations are given by
\[\begin{aligned}
\sqrt{-g^{00}}\nabla_{e_0}\widetilde n_i &= \sqrt{-g^{00}} \nabla_{e_0} \sqrt{g_{ii}g^{jj}g^{kk}}
= \sqrt{-g^{00}} (\Gamma_{i0}^i -\Gamma_{j0}^j -\Gamma_{k0}^k)\widetilde n_i= (K_i^i -K_j^j - K_k^k)\widetilde n_i,
\end{aligned}\]
and
\[\begin{aligned}
R_{0i}^{\ \ 0i} &= -R_{ij}^{\ \ ij}-R_{ik}^{\ \ ik}=R_{jk}^{\ \ jk}\\
\sqrt{-g^{00}}\nabla_{e_0} K_i^i &= -K_i^iK_i^i + K_{j}^j K_k^k  - \widetilde R_{jk}^{\ \ jk}.
\end{aligned}\]

\paragraph{Trace-free formulation.}
It is useful to split the variables into their trace and trace-free parts:
\[\begin{aligned}
H &= \frac 1 3(K_1^1+K_2^2+K_3^3),\qquad\text{i.e.~$H$ is the mean curvature of $\{t=\mathrm{const}\}$}\\ 
\sigma_i &= K_i^i -H\\
\frac 1 6 \widetilde R &= \frac 1 3\left(\widetilde R_{12}^{\ \ 12}+\widetilde R_{23}^{\ \ 23}+\widetilde R_{31}^{\ \ 31}\right)
= \frac 1 {12}\left(n_1^2+n_2^2+n_3^2 -2(n_1n_2+n_2n_3+n_3n_1)\right)\\
s_i &= \widetilde R_{jk}^{\ \ jk}-\frac 1 6 \widetilde R
= \frac 1 {3}\left(2n_i^2 - n_j^2 -  n_k^2  - n_in_j +2 n_jn_k - n_kn_i \right),
\end{aligned}\]
where we note that $(\sigma_1+\sigma_2+\sigma_3)^2 = 0=\sigma_1^2+\sigma_2^2+\sigma_3^2 +2(\sigma_1\sigma_2+\sigma_2\sigma_3+\sigma_3\sigma_1)$.
We then obtain
\[\begin{aligned}
0 & = \frac 1 2 \widetilde R - (\sigma_1+H)(\sigma_2+H) - (\sigma_2+H)(\sigma_3+H) - (\sigma_1+H)(\sigma_3+H)\\
&= \frac 1 2 \widetilde R - 3H^2 - (\sigma_1\sigma_2+\sigma_2\sigma_3+\sigma_3\sigma_1)\\
&= \frac 1 2 \widetilde R - 3H^2 +\frac 1 2 (\sigma_1^2+\sigma_2^2+\sigma_3^2)\\
\sqrt{-g^{00}}\nabla_{e_0}\widetilde n_i &= (2\sigma_i-H)\widetilde n_i\\
\sqrt{-g^{00}}\nabla_{e_0} H &= \frac 1 3 \Big[\;-(\sigma_1+H)^2 -(\sigma_2+H)^2 -(\sigma_3+H)^2\\
&\qquad+(\sigma_1+H)(\sigma_2+H) +(\sigma_2+H)(\sigma_3+H) + (\sigma_3+H)(\sigma_1+H) \\
&\qquad- \widetilde R_{23}^{\ \ 23} -  \widetilde R_{13}^{\ \ 13} -  \widetilde R_{12}^{\ \ 12}\Big]\\
&=-\frac 1 2(\sigma_1^2+\sigma_2^2+\sigma_3^2) -\frac 1 6 \widetilde R
= -\frac 1 3(\sigma_1^2+\sigma_2^2+\sigma_3^2) - H^2 \\
\sqrt{-g^{00}}\nabla_{e_0} \sigma_1 &= -(\sigma_1+H)^2 + (\sigma_2+H)(\sigma_3+H) - \widetilde R_{23}^{\ \ 23}  - \sqrt{-g^{00}}\nabla_{e_0}H\\
&=\sigma_2\sigma_3-\sigma_1^2 +(\sigma_2+\sigma_3-2\sigma_1)H - \widetilde R_{23}^{\ \ 23}  - \sqrt{-g^{00}}\nabla_{e_0}H\\
&=(\sigma_1+\sigma_3)(\sigma_1+\sigma_2)-\sigma_1^2 -3\sigma_1H   +\frac 1 2 (\sigma_1^2+\sigma_2^2+\sigma_3^2) +  \frac 1 6 \widetilde R - \widetilde R_{23}^{\ \ 23}\\
&= -3\sigma_1H -s_i.
\end{aligned}\]
\paragraph{Hubble Normalization.}
We can further simplify by Hubble-normalizing to $\overline{\Sigma}_i = \frac{\sigma_i} H$ and $\overline{N}_i = \frac{n_i} H$, and introducing a shorthand for ${\Sigma}^2$ and ${N}^2$
\[\begin{aligned}
\overline{N}_i &= \widetilde{\overline{N}}_i \signN_i\\
{\Sigma}^2 &= \frac 1 6\left({\Sigma}_1^2+\Sigma_2^2+\Sigma_3^2\right)\\
 N^2&= \frac 1 {12} \left[\overline N_1^2+\overline N_2^2+\overline N_3^2-2(\overline N_1\overline N_2+\overline N_2\overline N_3+\overline N_3\overline N_1)\right]= \frac 1 6 \widetilde R H^{-2}\\
 S_i &= -\frac 1 3 \left[-2 \overline N_1^2 + \overline N_2^2 + \overline N_3^2 +\overline N_1\overline N_2 -2\overline N_2\overline N_3 +\overline N_3\overline N_1 \right]= s_i H^{-2}\end{aligned}\]
which yields equations
\[\begin{aligned}
1 &=  \Sigma^2+ N^2\\
\sqrt{-g^{00}}\nabla_{e_0} H &=-\left(2\Sigma^2 +1 \right)H^2\\
\sqrt{-g^{00}}\nabla_{e_0}\widetilde{\overline N}_i &= (2 \Sigma_i - 1 + 2 \Sigma^2 +1)\widetilde{\overline N}_iH=2(\Sigma^2+ \Sigma_i)\widetilde{\overline N}_i H\\
\sqrt{-g^{00}}\nabla_{e_0} \overline \Sigma_1  &=
-3\overline \Sigma_1 H - s_i H^{-1} + (2\overline \Sigma^2+1)H = 2( \Sigma^2-1)\overline\Sigma_1H - S_iH.
\end{aligned}\]
\paragraph{The Wainwright-Hsu equations as used in this work.}
Assuming $H<0$ for an initial surface (which can be obtained if $H\neq 0$ by choosing the direction of the unit normal $\sqrt{-g^{00}}e_0$ and reverting the direction of time), we can set $\sqrt{-g^{00}} = -\frac 1 2 H$. Setting $\widetilde N_i = \sqrt{12}\,\widetilde{\overline N}_i$, we obtain the variant of the Wainwright-Hsu equations used in this work, \eqref{eq:ode-from-gr}, corresponding to the metric \eqref{eq:metric-in-wsh}.

\bibliography{./literaturlist}
\bibliographystyle{alpha}%
%
\end{document}